\documentclass[11pt]{article}
\usepackage[utf8]{inputenc}
\usepackage[authoryear]{natbib}
\usepackage[a4paper]{geometry}
\usepackage[title]{appendix}
\usepackage{xcolor}
\usepackage{setspace}
\usepackage{graphicx,subcaption,caption}
\usepackage{mathtools}
\usepackage{amsthm}
\usepackage{amsmath}
\usepackage{amsfonts}
\usepackage{amssymb}
\usepackage{natbib}
\usepackage{algorithm}
\usepackage{tabularx} 
\usepackage{multirow}
\usepackage{booktabs}
\usepackage{xr}
\usepackage[hidelinks]{hyperref}
\usepackage{algpseudocode}
\usepackage{algorithm}
\usepackage{bbm}
\usepackage{bm}
\usepackage{comment}
\usepackage{makecell}

\numberwithin{equation}{section}
\theoremstyle{plain}

\newtheorem{definition}{Definition}
\newtheorem{remark}{Remark}
\newtheorem{lemma}{Lemma}
\newtheorem{assumption}{Assumption}
\newtheorem{theorem}{Theorem}
\newtheorem{example}{Example}
\newcommand{\norm}[1]{\left\lVert#1\right\rVert}
\DeclarePairedDelimiter\abs{\lvert}{\rvert}%
\makeatletter
\let\oldabs\abs
\def\abs{\@ifstar{\oldabs}{\oldabs*}}
\newcommand{\Btheta}{\boldsymbol{\theta}}
\newcommand{\By}{\mathbf{y}}
\newcommand{\Bu}{\mathbf{u}}

\makeatletter
\newcommand*{\addFileDependency}[1]{
  \typeout{(#1)}
  \@addtofilelist{#1}
  \IfFileExists{#1}{}{\typeout{No file #1.}}
}
\makeatother

\newcommand*{\myexternaldocument}[1]{%
    \externaldocument{#1}%
    \addFileDependency{#1.tex}%
    \addFileDependency{#1.aux}%
}
\myexternaldocument{supplement}

\title{A correlated pseudo-marginal approach to doubly intractable problems}
\author{Yu Yang\thanks{School of Economics, University of New South Wales.} \and Matias Quiroz \thanks{School of Mathematical and Physical Sciences, University of Technology Sydney} \and Robert Kohn$^*$\thanks{Data Analytics for Resources and Environments (DARE), University of Sydney.} \and Scott A. Sisson\thanks{School of Mathematics \& Statistics, University of New South Wales.}\,\,\,\thanks{UNSW Data Science Hub, UNSW Sydney, Australia.}}
\date{}

\begin{document}

\maketitle

\begin{abstract}
Doubly intractable models are encountered in a number of fields, e.g.\ social networks, ecology and epidemiology. Inference for such models requires the evaluation of a likelihood function, whose normalising factor depends on the
model parameters and is assumed to be computationally intractable. The normalising constant of the posterior distribution and the additional normalising factor of the likelihood function result in a so-called doubly intractable posterior, for which it is difficult to directly apply Markov chain Monte Carlo methods. We propose a signed pseudo-marginal Metropolis-Hastings algorithm with an unbiased block-Poisson estimator to sample from the posterior distribution of doubly intractable models. As the estimator can be negative, the algorithm targets the absolute value of the estimated posterior and uses an importance sampling estimator to ensure simulation-consistent estimates of the posterior mean of a function of the parameters. The importance sampling estimator can perform poorly when its denominator is close to zero. We derive a finite-sample concentration inequality that ensures, with high probability, that this pathological case does not occur. Our estimator for doubly intractable problems has three advantages over existing estimators. First, the estimator is well-suited for efficient parallelisation and vectorisation. Second, its structure is ideal for correlated pseudo-marginal methods, which are well known to dramatically increase sampling efficiency. Third, the estimator enables the derivation of heuristic guidelines for tuning its hyperparameters under simplifying assumptions. We demonstrate the superior performance of our method in the standard benchmark example that models correlated spatial data using the Ising model, as well as the Kent distribution model for spherical data.
\vspace{2mm}
\end{abstract}

\newpage 
\section{Introduction}
Markov chain Monte Carlo (MCMC) methods \citep[see, e.g.,][for an overview] {brooks2011handbook} sample from a posterior distribution without evaluating its normalising constant, also known as the marginal likelihood. However, in some settings, the likelihood function itself contains an additional normalising constant that depends on the model parameters, and the resulting so-called doubly intractable posterior distribution falls outside the standard MCMC framework. To distinguish these normalisation quantities, we refer to the first as a normalising constant and the latter as a normalising function. 
Many well-known models give rise to doubly intractable posteriors, such as the exponential random graph models for social networks \citep{hunter2006inference} and non-Gaussian Markov random field models in spatial statistics, including the Ising model and its variants \citep{lenz1920beitrvsge,ising1925beitrag, hughes2011autologistic}. Doubly intractable models are recognised as among the most challenging problems in statistics, as they involve an intractable likelihood and also make it difficult to simulate from the model for fixed parameter values \citep{rudolf2024perturbations}.

Several algorithms are available to tackle the doubly intractable problem in Bayesian statistics; see \cite{park2018bayesian} for a review. These algorithms are classified into two main categories, with some overlap between them. The first category of methods introduces cleverly chosen auxiliary variables that cancel the normalising function when carrying out the MCMC sampling, and standard MCMC such as the Metropolis-Hastings (MH) algorithm
 \citep{metropolis1953equation, hastings1970monte} can thus be applied. This approach is model-dependent and cannot always be applied. The second category of methods, which applies more generally, approximates the likelihood function (including the normalising function) and substitutes the approximation in place of the exact likelihood in the estimation procedure. The
pseudo-marginal (PM) method \citep{beaumont2003estimation,andrieu2009pseudo}  is often used when a positive and unbiased estimator of the likelihood is available through Monte Carlo simulation. However, in some problems, including doubly intractable models, forming an unbiased estimator that is almost surely positive is only possible under unrealistic assumptions \citep{jacob2015nonnegative}. 
The so-called Russian roulette estimator  \citep{lyne2015russian} is an example of a method that can be used to unbiasedly estimate the likelihood function in doubly intractable models, although the estimate is not necessarily positive. Our paper focuses on methods that, asymptotically (in terms of chain length), yield samples from the exact doubly intractable posterior or provide estimates of expectations with respect to it.

We propose a method for exact inference on posterior expectations in doubly intractable problems based on the approach in \cite{lyne2015russian}, where an unbiased, but not necessarily positive, estimator of the likelihood function is used. The algorithm targets a posterior density that uses the absolute value of the likelihood, resulting in iterates from a perturbed target density. We follow \cite{lyne2015russian} and reweight the samples from the perturbed target density using importance sampling to obtain simulation-consistent estimates of the expectation of a function of the parameters with respect to the true posterior density. By simulation-consistent, we mean that the posterior expectation can be estimated to an arbitrary precision by increasing the number of iterations of the algorithm. While our method does not sample from the target of interest, we refer to it as exact due to its simulation-consistent property. 

Our main contribution is to explore the use of the block-Poisson (BP) estimator \citep{quiroz2021block} in the context of estimating doubly intractable models using the signed PMMH approach. Moreover, we contribute to the signed PMMH literature by deriving a finite-sample (in terms of chain length) result that guarantees, with high probability, that the denominator of the importance sampling estimator --- used to compute expectations under the doubly intractable posterior --- remains bounded away from zero. A near–zero denominator of the importance sampling estimator is known to be detrimental for signed PMMH methods \citep{lyne2015russian,quiroz2021block}. Our method provides the following advantages over the Russian roulette method. First, the BP estimator's simpler structure enables greater computational efficiency through parallelisation and vectorisation.  Second, the block form of our estimator makes it possible to correlate the estimators of the doubly intractable posterior at the current and proposed draws in the MH algorithm. Introducing such correlation dramatically improves the efficiency of PM algorithms \citep{tran2016block, deligiannidis2018correlated}. Finally, under simplifying assumptions, some statistical properties of the logarithm of the absolute value of our estimator are derived and used to obtain heuristic guidelines to tune the hyperparameters of the estimator. We demonstrate empirically that our method outperforms \cite{lyne2015russian} when estimating the parameter of an Ising model following the settings in the review paper \cite{park2018bayesian}. In a real data application for directional data, our method has a significantly shorter computing time compared to competing Bayesian methods. To the best of our knowledge, our method and that of \cite{lyne2015russian} with its extensions are the only alternatives in the PM framework to perform exact inference (in the sense of simulation-consistent estimates of posterior expectations) for general doubly intractable problems. Compared to algorithms using auxiliary variables to avoid evaluating the normalising function, signed PMMH algorithms are more widely applicable and generic as they do not require exact sampling from the likelihood, which may be hard to implement. Exact sampling refers to the ability to draw independent (data) samples exactly distributed according to the likelihood.

The rest of the paper is organised as follows. Section \ref{sec:Doubly_intractable} introduces the doubly intractable problem and discusses previous research. Section \ref{sec: method} introduces our methodology and presents a theoretical result (Theorem \ref{thm: sum far from zero}) that helps avoid a pathological case arising in finite-length MCMC chains when implementing signed pseudo-marginal methods---a practical concern that appears to have received little attention in the literature. Section \ref{sec: method} establishes the guidelines for tuning the hyperparameters of the estimator. Section \ref{sec: demo} reports on a replicated simulation study from the review paper \cite{park2018bayesian} for the Ising model and, additionally, the Kent distribution for modelling directional data. Section \ref{sec: empirical study} analyses four real-world datasets using the Kent distribution. Section \ref{sec: discussion} concludes and outlines future research.  The paper has a supplement that contains all proofs and details of the simulation studies, as well as the details of the other methods applied. We refer to equations, sections, lemmas in the main paper as (1.1), Section~1, Lemma~1 etc., and to equations, sections and lemmas, etc. in the supplement as (S1.1), Section~S1 and Lemma~S1, etc.

\section{Doubly intractable problems}\label{sec:Doubly_intractable}
\subsection{Doubly intractable posterior distributions\label{sec: doubly intract prob}}
Let $p(\By|\Btheta) $ denote the density of the data vector $\By$,  where $\Btheta$ is the vector of model parameters. Suppose $p(\By |\Btheta) = f(\By |\Btheta)/Z(\Btheta) $, where $f(\By |\Btheta)$ is computable while the normalising function  $Z(\Btheta)$ is not.  The reason that $Z(\Btheta)$  is intractable may be that it is prohibitively expensive to evaluate numerically, or lacks a closed form because $Z(\Btheta)$ is defined as an integral over a complex or high-dimensional space which is hard to evaluate, or it involves summing over an intractably large number of terms.  Two examples are given below to demonstrate the intractability for both discrete and continuous observations $\By$. 

\begin{example}[The Ising model \citep{ising1925beitrag}]\label{Ex: Ising}

 Consider an $L \times L$ lattice with binary observation $y_{ij}\in \{ -1,1\}$ in row $i$ and column $j$. The likelihood of  $\theta \in \mathbb{R}$ is
\begin{align} \label{eq: Ising}
    p(\By|\theta) &= \frac{1}{Z(\theta)}\exp ( \theta S(\By)) ; \quad 
    S(\By) = \sum_{i=1}^{L} \sum_{j=1}^{L-1} y_{i,j} y_{i,j+1} + \sum_{i=1}^{L-1} \sum_{j=1}^L y_{i,j} y_{i+1,j};\\
    \mbox{with } 
    Z(\theta)  &= \sum_{\By} \exp(\theta S(\By))\notag.
\end{align}
The normalising function $Z(\theta)$ in the Ising model is a sum over $2^{L^2}$ terms, making it computationally intractable even for moderate values of $L$. See  Section \ref{sec: ising model} for a further discussion.
\end{example}

\begin{example} [The Kent distribution \citep{kent1982fisher}] \label{Ex: Kent}

The density of the Kent distribution 
 for $\By \in \mathbb{R}^3, \norm{\By} = 1$, is
 \begin{align} \label{eq: Kent}
     f(\By|\boldsymbol{\gamma_1},\boldsymbol{\gamma_2},\boldsymbol{\gamma_3},\beta,\kappa) &= \frac{1}{c(\kappa,\beta)} \exp\left\{\kappa \boldsymbol{\gamma_1}^\top \cdot \By + \beta \left[(\boldsymbol{\gamma_2}^\top \cdot\By)^2 - (\boldsymbol{\gamma_3}^\top \cdot\By)^2\right]\right\};\\
      \mbox{with } c(\kappa, \beta) &= 2\pi \sum_{j=0}^{\infty} \dfrac{\Gamma(j+0.5)}{\Gamma(j+1)} \beta^{2j} (0.5\kappa)^{-2j-0.5} I_{2j+0.5}(\kappa),\; \notag
 \end{align}
where $I_\nu(.)$ is the modified Bessel function of the first kind and $\boldsymbol{\gamma_1},\boldsymbol{\gamma_2},\boldsymbol{\gamma_3}$ form a set of 3-dimensional orthonormal vectors. The normalising function $c(\kappa,\beta)$ is an infinite sum and thus intractable to evaluate. See  Section \ref{subsec: kent_dist} for a further discussion.
\end{example}

The doubly intractable posterior density of $\Btheta$ is
\begin{align}\label{eq:target_doubly_intractable}
\pi(\Btheta|\By)  = \dfrac{f(\By|\Btheta) \pi(\Btheta)}{Z(\Btheta) p(\By)} \propto \dfrac{f(\By|\Btheta) \pi(\Btheta)}{Z(\Btheta)},
\end{align}
where $\pi(\Btheta)$ is the prior for $\Btheta$ and  \begin{align}\label{eq:normalising_const}
    p(\By) & =\int \dfrac{f(\By|\Btheta) \pi(\Btheta)}{Z(\Btheta)} d\Btheta
\end{align} 
is the normalising constant for the posterior. Suppose we devise a Metropolis-Hastings algorithm to sample from \eqref{eq:target_doubly_intractable} with a proposal density $q(\cdot|\Btheta)$. The probability of accepting a proposed sample $\Btheta'$ is
\begin{align}\label{eq:MH_accept}
\alpha(\Btheta',\Btheta) = \min \left\{1,
     \dfrac{\pi(\Btheta')f(\By|\Btheta')/Z(\Btheta')}{\pi(\Btheta)f(\By|\Btheta))/Z(\Btheta)}\times  \dfrac{q(\Btheta|\Btheta')}{q(\Btheta'|\Btheta)}\right\}. 
\end{align}
The marginal likelihood in \eqref{eq:normalising_const} cancels in \eqref{eq:MH_accept}, but the normalising function does not. Since $Z(\Btheta)/Z(\Btheta')$ is computationally intractable, \eqref{eq:MH_accept} cannot be evaluated and thus MCMC sampling via the Metropolis-Hastings algorithm is impossible.

\subsection{Previous research} \label{sec: previous research}
Previous research on doubly intractable problems is mainly divided into the auxiliary variable approach and the likelihood approximation approach; see \cite{park2018bayesian} for an excellent review of both approaches.

The auxiliary variable approach cleverly chooses the joint transition kernel of the parameters and the auxiliary variables so that the normalising function cancels in the resulting MH acceptance ratio. The most well-known algorithms in this space are the exchange algorithm \citep{murray2006mcmc} and the auxiliary variable method \citep{moller2006efficient}. \cite{Andrieu2020MetropolisHastingsWA} extend the exchange algorithm by leveraging averaged acceptance ratios for improved stability. These algorithms are model-dependent and, crucially, rely on the sampling technique used to draw observations from the likelihood function. Perfect sampling \citep{propp1996exact} is often used to generate data samples from the model without knowing the normalising function. However, for some complex models, such as the Ising model on a large grid, perfect sampling is prohibitively expensive.  To overcome this issue,  \cite{liang2010double} and \cite{liang2016adaptive} relax the requirement of exact sampling and propose the double MH sampler and the adaptive exchange algorithm. However, the former generates inexact inference results, while the latter suffers from memory issues as many intermediate variables need to be stored within each iteration. 

Methods belonging to the likelihood approximation approach can be simulation-consistent. One example is \cite{atchade2013bayesian}, who directly approximate $Z(\Btheta)$ through multiple importance sampling. Their approach also depends on an auxiliary variable technique, but does not require perfect sampling. The downside is similar to that of the adaptive exchange algorithm; a large memory is usually required to store the intermediate variables generated in each iteration. An alternative method is to approximate $1/Z(\Btheta)$ directly using the signed PMMH algorithm to replace the likelihood function by its unbiased estimator as proposed in \cite{lyne2015russian}. To obtain the unbiased estimator, $1/Z(\Btheta)$ is expressed as a geometric series, which is truncated using a Russian roulette (RR) approach. The RR method first appeared in the physics literature \citep{carter1975particle} and is useful for obtaining an unbiased estimator through a finite stochastic truncation of an infinite series. To implement RR, a tight upper bound for $Z(\Btheta)$ is required; otherwise, the convergence of the geometric series is slow and makes the algorithm inefficient. In practice, an upper bound is usually unavailable, which may lead to negative estimates of the likelihood, and thus a signed PMMH approach is necessary, although this inflates the asymptotic variance of the resulting importance sampling estimator of posterior expectations. In particular, having nearly half of the estimates being negative is detrimental for the asymptotic variance \citep{lyne2015russian,quiroz2021block}. It is therefore crucial to control the probability of a negative estimate, which is difficult for the RR estimator as, to the best of our knowledge, no analytical expression exists for this quantity. In contrast, our estimator is more tractable, and the probability of a positive estimate is analytically derived under simplifying assumptions. Besides the upper bound, a few other hyperparameters of the RR estimator need to be determined. We are unaware of any guidelines for selecting these based on analytical expressions. We provide such guidelines for the hyperparameters in the estimator proposed in our paper. \cite{wei2017markov} combine RR with Markov chain coupling to produce an estimator with lower variance and a larger probability of producing positive estimates. However, their estimator is not tractable enough to readily provide tuning guidelines. \cite{caimulti} propose a multi-fidelity MCMC method to approximate the doubly intractable target density, which, like the Russian roulette method, stochastically truncates an infinite series and uses slice sampling \citep{murray2016pseudo}. However, similarly to \cite{lyne2015russian}, the method lacks guidelines for tuning the hyperparameters.

Finally,  although our focus is on exact approaches, we note that several approximate methods for estimating doubly intractable posteriors are available. \cite{Alquier2016NoisyMC} propose the noisy exchange algorithm, which uses multiple auxiliary variables to obtain an estimate of the acceptance ratio in the exchange algorithm. As the number of auxiliary variables goes to infinity, the algorithm samples from the exact doubly intractable posterior. However, for a finite number of auxiliary variables, the method is approximate. Variational Bayes approaches include \cite{tan2020bayesian} and \cite{lee2024stein}. \cite{Park2020emulation} propose approximating the normalising function at several parameter values using importance sampling, and interpolating it at other parameter values via Gaussian process-based emulation.

\section{Methodology} \label{sec: method}

\subsection{The block-Poisson estimator}\label{sec: blockPois}
The block-Poisson estimator \citep{quiroz2021block} was proposed for estimating the likelihood unbiasedly given an unbiased estimator of the log-likelihood obtained by data subsampling. The BP estimator builds on the Poisson estimator \citep{Wagner1988336,papaspiliopoulos2011monte}, which is useful for estimating $\exp(B)$ unbiasedly given an unbiased estimator $\widehat{B}$ of $B$, i.e.\  $E(\widehat{B}) = B$. The estimator $\widehat{B}$ can be obtained by Monte Carlo integration based on $M$ samples. The Poisson estimator of $\exp(B)$ is
\begin{align}
   \exp(m+a)\prod_{h=1}^{\chi} \frac{\widehat{B}^{(h)}-a}{m},\: \chi \sim \mathrm{Pois}(m),\: m\in\mathbb{Z},\: a \in \mathbb{R},\label{eq:orig_Pois}
\end{align}
where $\widehat{B}^{(h)}$ are independent copies of $\widehat{B}$.

The block-Poisson estimator also estimates $\exp(B)$ unbiasedly and consists of $\lambda$ Poisson estimators similar to \eqref{eq:orig_Pois}, however, estimating $\exp(B/\lambda)$ unbiasedly. The idea behind using blocks of Poisson estimators, instead of a single one as in \eqref{eq:orig_Pois}, is to allow for correlation between successive iterates in the PM algorithm as described in Section \ref{sec: signed PMMH}. Similarly to the likelihood approximation approaches discussed above, the BP estimator is implemented in combination with an auxiliary variable $\nu$, and an estimator of the normalising function. Omitting details of the auxiliary variable method that are explained in Section \ref{sec: signed PMMH}, assume $B(\Btheta) = -\nu Z(\Btheta)$ where $\nu \sim \mathrm{Exp}(Z(\Btheta))$. Our procedure estimates the likelihood for given $\nu$, and thus $\nu$ is treated as fixed (non-random) in the results derived for our estimator that we use for the tuning guidelines. The BP estimator produces (for a fixed $\nu$) an unbiased estimator of $\exp(-\nu Z(\Btheta))$ using unbiased estimators of the normalising function $Z(\Btheta)$. One advantage of the BP estimator over the RR estimators is that its simple form, together with simplifying assumptions, e.g.\ that the estimator of the normalising function is normal, enables hyperparameter tuning based on analytical expressions. As a result, the BP estimator is more likely to produce positive estimates if tuned following our guidelines in Section \ref{sec: tuning para}. Controlling the signs of the estimates is desirable for efficient estimation based on MCMC output as discussed in Sections \ref{sec: signed PMMH} and \ref{sec: tuning para}.

 Definition \ref{def: BP} describes the BP estimator $\widehat{L}_B$ of the likelihood we use for doubly intractable problems.
 Lemma \ref{lemma: block poisson est} gives the expectation and variance of $\widehat{L}_B$. Lemmas \ref{lemma: positive prob} and \ref{lemma: log variance} establish useful results for tuning the hyperparameters of the estimator (see Section \ref{sec: tuning para}). The proofs are in 
 Section \ref{app: BPproof} in the supplement.

\begin{definition} \label{def: BP}
The block-Poisson estimator is defined as 
\begin{equation}\label{eq: block_pois_est}
\widehat{L}_B(\Btheta) = \prod_{l=1}^{\lambda} \xi_l(\Btheta), \,\, 
\xi_l(\Btheta) =\exp(a/ \lambda + m) \prod_{h=1}^{\chi_l} \dfrac{\widehat{B}^{(h,l)}(\Btheta)-a}{m\lambda}, 
\end{equation}
where $\lambda$ is the number of blocks, $\chi_l \sim \mathrm{Pois}(m)$, a Poisson distribution with mean $m$, $a$ is an arbitrary constant and $m$ is the expected number of estimators used within each block. The unbiased estimates of $B$, $\widehat{B}^{(h,l)}$ are independent with the indexes $h=1,\dots, \chi_l, l= 1,\dots, \lambda$. 

\end{definition}

\begin{remark}
If $m=1$ then \eqref{eq: block_pois_est} is the estimator in
\cite{quiroz2021block}. To ensure a positive estimator with probability 1, $a$ needs to be a lower bound for all $\widehat{B}^{(h,l)}$.
\end{remark}

\begin{lemma}\label{lemma: block poisson est}
Denote $\sigma^2_{\widehat{B}} =\mathrm{Var}(\widehat{B}(\Btheta))$, and assume $\sigma^2_{\widehat{B}} < \infty$ and $E(\widehat{B}(\Btheta)) = B(\Btheta)$. The following properties hold for $\widehat{L}_B(\Btheta)$ in \eqref{eq: block_pois_est}:
\begin{itemize}
    \item[](i) $E(\widehat{L}_B(\Btheta)) = \exp(B(\Btheta))$.

    \item[](ii) $
    \mathrm{Var}(\widehat{L}_B(\Btheta)) = \exp\bigg[ \dfrac{(B(\Btheta)-a)^2 + \sigma^2_{\widehat{B}}}{m\lambda } + 2a + m\lambda \bigg] - \exp(2B(\Btheta)) $.
    
    \item[](iii) $\mathrm{Var}(\widehat{L}_B(\Btheta))$ is minimised at $a = B(\Btheta) - m\lambda$, given fixed $m$ and $\lambda$.
\end{itemize}
\end{lemma}

Part (i) of Lemma \ref{lemma: block poisson est} shows that given an unbiased estimator $\widehat{B}(\Btheta)$ of $B(\Btheta)$, the BP estimator is unbiased for  $\exp(B(\Btheta))$ for any values of the hyperparameters. Let $a_{\mathrm{opt}}$ denote the value of $a$ that minimises the variance of the estimator in Part (ii) (for a fixed $m$ and $\lambda$). Part (iii) of Lemma \ref{lemma: block poisson est} shows that $a_{\mathrm{opt}} = B(\Btheta) - m\lambda$.

Similarly to the RR estimator, the BP estimator is not necessarily positive unless $a$ is a lower bound of $\widehat{B}(\Btheta)$, i.e.\ $\widehat{B}(\Btheta)>a$. This implies that $a < 0$ if $\widehat{Z}(\Btheta) >0$ (recall that $\widehat{B}(\Btheta)=-\nu\widehat{Z}(\Btheta),\,\nu\geq 0$) which usually holds in doubly intractable problems. Selecting a large negative value of $a$ (to account for extreme outcomes of $\widehat{B}(\Btheta)$), Part (iii) suggests that $m\lambda$ should also be large to keep the variance small (typically, $|\widehat{B}(\Btheta)|\ll m\lambda$). This translates to a computationally costly estimator due to many products in the BP estimator. We follow \cite{quiroz2021block} and advocate the use of a soft lower bound, i.e.,\ one that may lead to negative estimates, but still gives a $\Pr(\widehat{L}_B(\Btheta) \geq 0)$ close to one. Lemma \ref{lemma: positive prob} shows that the probability $\Pr(\widehat{L}_B(\Btheta) \geq 0)$ has an analytical expression. It is crucial to have this probability close to one for the algorithm to be efficient.  
\begin{lemma}\label{lemma: positive prob} Suppose that $a=a_\mathrm{opt}=B(\Btheta) - m\lambda$. Then,
\begin{align}\label{eq:prob_pos_estimator}
    \Pr(\widehat{L}_B(\Btheta) \geq 0) & = \dfrac{1}{2}\bigg( 1 + (1-2 \mathsf{\Psi}(a,m,\lambda, M))^\lambda \bigg),
\end{align}
with $\Psi(a,m,\lambda, M) = \Pr (\xi < 0) = \dfrac{1}{2} \sum_{j=1}^\infty \left( 1 - (1-2 \Pr(A_m \leq 0))^j \right) \Pr(\chi_l = j)$, $\chi_l \sim \mathrm{Pois}(m)$ and $A_m = [\widehat{B}(\Btheta) - B(\Btheta)]/ (m\lambda) + 1$, and $M$ is the number of Monte Carlo samples to estimate a single $\widehat{B}(\Btheta)$.
\end{lemma}
\begin{remark}\label{rem:lemma_positive_prob}
To compute the probability $\Pr(A_m \leq 0)$ in practice, we assume that $\widehat{B}(\Btheta)$ is normal with variance $\sigma^2_{\widehat{B}}$ as in Lemma \ref{lemma: log variance}. When the Monte Carlo samples $M$ are independent, then $\sigma^2_{\widehat{B}}\propto 1/M$.
\end{remark}

Figure \ref{fig:prob_pos_estimator} illustrates some terms in Lemma \ref{lemma: positive prob} under the assumptions in Remark \ref{rem:lemma_positive_prob} and $m=1$. The left panel shows the probability of an individual term in the block-Poisson estimator, i.e.\ $\Pr(\xi <0)$ in \eqref{eq: block_pois_est}, being negative as a function of $\sigma^2_{\widehat{B}}$ ($\log_{10}$ scale) for various $\lambda$ values. The right panel shows the probability of the overall block-Poisson estimator being positive, i.e.\ $\Pr(\widehat{L}_B \geq 0)$ in \eqref{eq:prob_pos_estimator}, for the same $\sigma^2_{\widehat{B}}$ and $\lambda$. As $\lambda$ increases, $\Pr(\xi < 0)$ remains near zero over a wider span of $\sigma^2_{\widehat{B}}$ (left panel). Consequently, the pathological regime---detailed later---where $\Pr(\widehat{L}_B \geq 0) \approx 0.5$ is pushed to much larger $\sigma^2_{\widehat{B}}$. Thus, this regime can be avoided by increasing $\lambda$ or by reducing $\sigma^2_{\widehat{B}}$ (the latter via increasing $M$). Section \ref{sec: tuning para} develops a tuning strategy that maximises an objective function including $\Pr(\widehat{L}_B \geq 0)$.

\begin{figure}[t]
    \centering \includegraphics[width=1.0\linewidth]{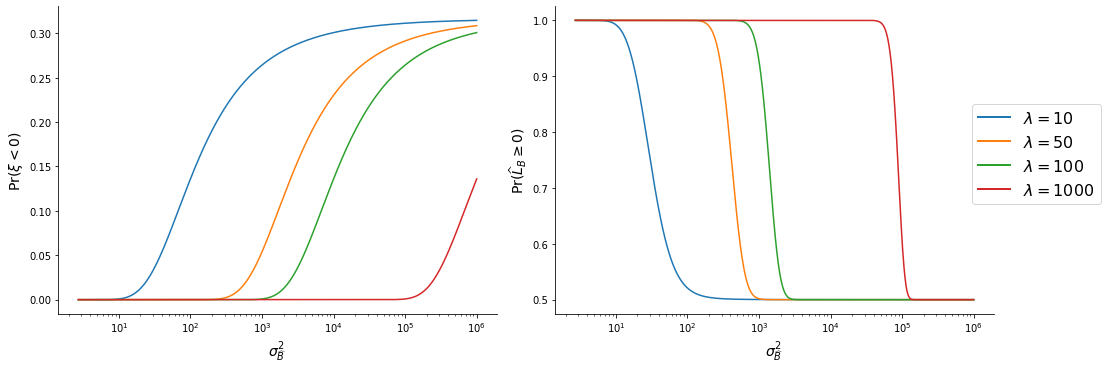}
    \caption{Plots of $\Pr(\xi<0)$ (left panel) and $\Pr(\widehat{L}\geq 0)$ (right panel) in Lemma \ref{lemma: positive prob} when $m=1$ as a function of   $\sigma^2_{\widehat{B}}$ ($\log_{10}$ scale) for various $\lambda$ values; see legend.}
\label{fig:prob_pos_estimator}
\end{figure}

Lemma \ref{lemma: log variance} derives the variance of the logarithm of the absolute value of the block-Poisson estimator by assuming that $\widehat{B}^{(h,l)}(\Btheta)$ is normal. 

\begin{lemma}\label{lemma: log variance}
If $\widehat{B}^{(h,l)}(\Btheta) \stackrel{\mathrm{iid}}{\sim} N(B(\Btheta), \sigma^2_{\widehat{B}})$ for all $h$ and $l$, when $a =a_\mathrm{opt}= B(\Btheta) - m\lambda$, then the variance of $\log |{\widehat{L}_B}|$ is
\[ \sigma^2_{\log \abs*{\widehat{L}_B}} =  m\lambda (\nu_B^2 + \eta_B^2),\]
where
     \[ \eta_B =  \log(\sigma_{\widehat{B}}/(m\lambda)) + 0.5 \left( \log 2 + E_J(\psi^{(0)} (0.5 + J)) \right)  \]
and
     \[\nu_B^2 = 0.25 \left(  E_J(\psi^{(1)} (0.5 +J)) + \mathrm{Var}_J(\psi^{(0)}(0.5+J)) \right), \]
with $J \sim \mathrm{Pois}( (m\lambda)^2/(2\sigma^2_{\widehat{B}}))$ and $\psi^{(q)}$ is the polygamma function of order $q$. 
\end{lemma}
\begin{remark}
In our method, $\widehat{B}^{(h,l)}(\Btheta)=-\nu\widehat{Z}^{(h,l)}(\Btheta)$. The assumption holds if $\widehat{Z}^{(h,l)}(\Btheta)$ is normal (recall that $\nu$ is treated as non-random in the tuning procedure). 
\end{remark}

Section \ref{sec: tuning para} tunes the hyperparameters using the results above. However, we have not yet dealt with the fact that $a_\mathrm{opt}=B(\Btheta)-m\lambda$ is itself intractable as it includes the normalisation function. A sensible approach is to estimate the normalising function, i.e.\ $\widehat{a}_\mathrm{opt}=\widehat{B}(\Btheta)-m\lambda$. One can show that Part (i) still holds if $a$ is random,  however, the extra randomness in $\widehat{L}_B(\Btheta)$ may cause problems such as an infinite $\mathrm{Var}(\widehat{L}_B(\Btheta))$. We therefore consider the non-random value $a_{\mathrm{sub}} = -1 -m\lambda$, where the subscript highlights that it is a sub-optimal choice. The choice is motivated by $E(\widehat{B}(\Btheta))=-\nu Z(\Btheta)$,  with $\nu$ replaced by its expected value $1/Z(\Btheta)$. We note that this sub-optimal choice does not have implications for the exactness of the algorithm. Moreover, Section \ref{subsec:simplified_lb} in the supplement includes a simulation study showing that given our assumptions, the sub-optimal soft lower bound $a_\mathrm{sub}$ often results in similar quantities used by the tuning procedure as those of the optimal choice $a_\mathrm{opt}$ and its estimated version $\widehat{a}_\mathrm{opt}$.

 \subsection{Signed block PMMH with the BP estimator}\label{sec: signed PMMH}
 \cite{lyne2015russian} use an auxiliary variable $\nu$ to cancel the reciprocal of the normalising function in \eqref{eq:target_doubly_intractable} and end up with $\exp(-\nu Z(\Btheta))$ (instead of the reciprocal ). Specifically,
 assuming that $\nu \sim \text{Exp}(Z(\Btheta))$, the joint density of $\Btheta$ and the auxiliary variable $\nu$ is
\begin{align}
    \pi(\Btheta, \nu|\By) &= Z(\Btheta) \exp(-\nu Z(\Btheta)) \dfrac{f(\By| \Btheta)}{Z(\Btheta)} \pi(\Btheta) \dfrac{1}{p(\By)} \nonumber \\
    &\propto \exp(-\nu Z(\Btheta)) f(\By|\Btheta) \pi(\Btheta) \label{eq:augmented_posterior}.
\end{align}
    
We can use the BP estimator in Section \ref{sec: blockPois} to obtain an unbiased estimator of the augmented posterior in \eqref{eq:augmented_posterior}, up to a normalising constant for a fixed (non-random) $\nu$. Denote the unbiased estimator of $\exp(-\nu Z(\Btheta))$ by $\widehat{\exp}(-\nu Z(\Btheta))$. To emphasise the source of randomness in the estimator, let $\Bu$ be a set of random numbers with density $\pi(\Bu)$ (assumed independent of $\Btheta$) and write the estimator (with a small abuse of notation) as $\widehat{\exp}(-\nu Z(\Btheta)|\Bu)$ for a fixed $\nu$. In the block-Poisson estimator, $\Bu$ includes all the random numbers used to generate the estimates $\widehat{B}$ and the Poisson variables $\chi$. The unbiasedness of the estimator is with respect to the density $\pi(\Bu)$, i.e.\
\begin{align}\label{eq:unbiased_estimator}
    \exp(-\nu Z(\Btheta)) & = \int_{\Bu} \widehat{\exp}(-\nu Z(\Btheta)|\Bu) \pi(\Bu) d\Bu.
\end{align}
The augmented version of the posterior density in \eqref{eq:augmented_posterior} is
\begin{align}\label{eq:augmented_posterior_u}
\widehat{\pi}(\Btheta,\Bu, \nu| \By) & = \widehat{\exp}(-\nu Z(\Btheta)|\Bu) \pi(\Bu) f(\By|\Btheta) \pi(\Btheta)\frac{1}{p(\By)}.    
\end{align}
It is easy to show that, under the unbiasedness condition in \eqref{eq:unbiased_estimator}, integrating out $\Bu$ in \eqref{eq:augmented_posterior_u} gives the marginal density of interest in \eqref{eq:augmented_posterior} for $\Btheta, \nu$. However, we cannot sample from \eqref{eq:augmented_posterior_u} using a pseudo-marginal algorithm as the BP estimates may be negative and hence it is not a valid density. We follow \cite{lyne2015russian} and consider the target density
\begin{align}\label{eq:absolute_target}
\overline{\pi}(\Btheta,\Bu,\nu)=|\widehat{\pi}(\Btheta,\Bu, \nu|\By)|  = |\widehat{\exp}(-\nu Z(\Btheta)|\Bu)| \pi(\Bu) f(\By|\Btheta) \pi(\Btheta)\frac{1}{p(\By)}.    
\end{align}

Integrating out $\Bu$ in \eqref{eq:absolute_target} does not give the marginal density of interest in \eqref{eq:augmented_posterior} because $|\widehat{\exp}(-\nu Z(\Btheta)|\Bu)|$ is biased. \cite{lyne2015russian} propose reweighting the MCMC iterates using importance sampling to obtain a simulation-consistent estimate of the expectation of an arbitrary function $\psi(\Btheta)$ (assuming the expectation exists) with respect to the posterior density $\pi(\Btheta|\By)$, i.e.\
\begin{align}\label{eq:expectation_doubly_intractable}
    E_{\pi} (\psi(\Btheta)|\By) = \int_{\Btheta} \psi(\Btheta) \pi(\Btheta|\By) d\Btheta,
\end{align}
which we now outline in some detail. We can write
\begin{align*}
     E_{\pi} (\psi(\Btheta)|\By) 
     &= \int_{\Btheta}\psi(\Btheta)\int_{\nu} \pi(\Btheta,\nu|\By) d\nu d\Btheta\\
     &= \int_{\Btheta}\int_{\nu}\psi(\Btheta) \pi(\Btheta,\nu|\By) d\nu d\Btheta\\ 
     &= \frac{\int_{\Btheta}\int_{\nu}\psi(\Btheta) \exp(-\nu Z(\Btheta)) f(\By|\Btheta) \pi(\Btheta)d\Bu d\nu d\Btheta}{\int_{\Btheta}\int_{\nu}\exp(-\nu Z(\Btheta)) f(\By|\Btheta) \pi(\Btheta) d\nu d\Btheta}  \\
     &= \frac{\int_{\Btheta}\int_{\nu}\int_{\Bu}\psi(\Btheta) \mathrm{sign}(\widehat{\pi}( \Btheta, \Bu, \nu|\By)) \left|\widehat{\exp}(-\nu Z(\Btheta)|\Bu)\right| \pi(\Bu)f(\By|\Btheta) \pi(\Btheta)d\Bu d\nu d\Btheta}{\int_{\Btheta}\int_{\nu}\int_{\Bu}\mathrm{sign}(\widehat{\pi}(\Btheta, \Bu, \nu|\By))\left|\widehat{\exp}(-\nu Z(\Btheta)|\Bu)\right| \pi(\Bu) f(\By|\Btheta) \pi(\Btheta)d\Bu d\nu d\Btheta} \\
     & = \frac{\int_{\Btheta}\int_{\nu}\int_{\Bu}\psi(\Btheta) \mathrm{sign}(\widehat{\pi}( \Btheta, \Bu, \nu|\By))\left|\widehat{\pi}(\Btheta, \Bu, \nu|\By)\right|d\Bu d\nu d\Btheta}{\int_{\Btheta}\int_{\nu}\int_{\Bu}\mathrm{sign}(\widehat{\pi}(\Btheta, \Bu, \nu|\By))\left|\widehat{\pi}(\Btheta, \Bu, \nu|\By)\right| d\Bu d\nu d\Btheta},
\end{align*}
i.e.\ a ratio of expectations with respect to $\overline{\pi}=\left|\widehat{\pi}(\Btheta, \Bu, \nu|\By)\right|$ in \eqref{eq:absolute_target}, 
where $\mathrm{sign}(x) = 1$ if $x >0$, or $\mathrm{sign}(x) = -1$ if $x < 0$. Thus, we can sample from $\left|\widehat{\pi}(\Btheta, \Bu, \nu|\By)\right|$ by a pseudo-marginal MH algorithm to compute the expectations in the ratio. We denote the expectation with respect to the augmented target posterior $\left|\widehat{\pi}(\Btheta, \Bu, \nu|\By)\right|$ by $E_{\overline{\pi}}$, and we note that it differs from $E_{\pi}$ except when the likelihood estimator is almost surely positive \citep{quiroz2021block}. Since the function $\psi(\Btheta)$ is independent of $\nu$, we only store $\Btheta^{(i)}$ and the sign of the likelihood estimate evaluated at the accepted $\Btheta^{(i)}, \Bu^{(i)}, \nu^{(i)}$ at the $i$th iterate. The estimate of the expectation in \eqref{eq:expectation_doubly_intractable} with respect to the true doubly intractable posterior in \eqref{eq:target_doubly_intractable} is
\begin{equation}\label{eq: final_est}
    \widehat{E}_{\pi}(\psi(\Btheta)) = \frac{\widehat{E}_{\overline{\pi}}\left(\psi(\Btheta)S(\Btheta, \Bu, \nu)\right)}{\widehat{E}_{\overline{\pi}}\left(S(\Btheta, \Bu, \nu)\right)}= \dfrac{\sum_{i=1}^N \psi(\Btheta^{(i)}) s^{(i)}}{\sum_{i=1}^N s^{(i)}},
\end{equation}
where $S(\Btheta, \Bu, \nu)= \mathrm{sign}(\widehat{\pi}(\Btheta, \Bu, \nu|\By))$ and $s^{(i)} = S(\Btheta^{(i)}, \Bu^{(i)}, \nu^{(i)})$. We often use the shorthand $S= S(\Btheta, \Bu, \nu)$ for the sign of the posterior estimate (inherited from the likelihood estimate).

\cite{quiroz2021block} and \cite{lyne2015russian} prove a central limit theorem for the estimator in \eqref{eq: final_est}. The resulting asymptotic variance is finite if (i) $\mathrm{E}_{\overline{\pi}}(S) \neq 0$, corresponding to $\Pr_{\overline{\pi}}(S=+1)\neq0.5$, which is akin to $\Pr\left(\widehat{L}_B(\Btheta) \geq 0\right) \neq 0.5$ in Lemma \ref{lemma: positive prob},  and (ii) the variance and inefficiency factor (under $\overline{\pi}$) of $\psi S$ are finite; see \cite[Theorem S1]{quiroz2021block} for details. Although a finite asymptotic variance is reassuring, it does not reflect the performance of the estimator in \eqref{eq: final_est} for a finite number of samples, which may still be poor despite the favorable asymptotic behavior. This issue was not addressed in \cite{lyne2015russian,quiroz2021block}. In particular, the variance might be large if the denominator in \eqref{eq: final_est} is close to zero. The minimal condition for a finite variance of \eqref{eq: final_est} is that $\left|\sum_{i=1}^N s^{(i)}\right| > 0$; however, for reliable performance, we would prefer $\left|\sum_{i=1}^N s^{(i)}\right| \gg 0$, say 
\begin{align}\label{eq:sum away from zero}
\left|\sum_{i=1}^N s^{(i)}\right| > c N, \quad \text{ for some } 0 < c < |\mu|,    
\end{align}
where $\mu = E_{\overline{\pi}}(S)=2\tau-1$, with $\tau = \Pr_{\overline{\pi}}(S = +1)$. To ensure that the sum in \eqref{eq:sum away from zero} is bounded away from zero, it is necessary that $\tau \neq 0.5$, i.e. $\mu \neq 0$ (otherwise $c = 0$).

Theorem \ref{thm: sum far from zero} below shows that the probability of \eqref{eq:sum away from zero} can be made arbitrarily close to 1 by increasing $N$. The proof of the theorem uses the Bernstein-type concentration inequality for Markov chains in \cite{paulin2015concentration}. The following assumptions are made.

\begin{assumption}\label{ass: assumptions}
We assume that:
\begin{itemize}
    \item[(i)] $\{\Btheta^{(i)},\Bu^{(i)},\nu^{(i)}\}_{i=1}^N$ is a realisation of a stationary, reversible Markov chain on $\Omega= \mathbb{R}^{\dim(\Btheta)} \times \mathbb{R}^{\dim(\Bu)} \times \mathbb{R}_{>0}$ with stationary distribution $\overline{\pi}$ in \eqref{eq:absolute_target}.
    \item[(ii)] The Markov chain in (i) has spectral gap  $0 < \delta < 1$.
    \item[(iii)] Define $S: \Omega \rightarrow \Omega^\prime$, $S=\mathrm{sign(\widehat{\pi}(\Btheta, \Bu, \nu|\By))}$ with $\widehat{\pi}$ in \eqref{eq:augmented_posterior_u} and $\Omega^{\prime}= \{-1, +1\}$. Let $\eta$ be a measure on $\Omega^{\prime}$ with
    \begin{align}\label{eq:pibar_induces_eta}
    \eta\left( \{+1\}\right) = \int_{\Btheta}\int_{\nu}\int_{\Bu} \mathbbm{1}\left(S(\Btheta, \Bu, \nu) = +1\right) \overline{\pi}(\Btheta, \Bu, \nu)d\Bu d\nu d\Btheta,
\end{align}
where $\mathbbm{1}(\cdot)$ is the indicator function. We assume that $\eta\left( \{+1\}\right)\neq \eta\left( \{-1\}\right)$, i.e. $\eta\left( \{+1\}\right)=\Pr_{\overline{\pi}}(S=+1)\neq 0.5$, and hence
\begin{align}\label{eq:expected_sign}
\mu &= E_{\overline{\pi}}(S)=2\eta\left(\{+1\}\right)-1 \neq 0. 
\end{align} 
\end{itemize}
\end{assumption}
\begin{remark}
    Note that $\tau$ defined after \eqref{eq:sum away from zero} and $\eta\left( \{+1\}\right)$ are the same, $\tau = \eta\left( \{+1\}\right)$; \eqref{eq:pibar_induces_eta} shows how the measure $\eta$ is induced from $\overline{\pi}$.
\end{remark}
Assumption \ref{ass: assumptions} is justified as follows. Theorem S1(ii) in \cite{quiroz2021block} states that the signed block pseudo-marginal (block introduced below) algorithm converges in total variation norm to $\overline{\pi}$, which justifies the assumed stationary distribution in Assumption~\ref{ass: assumptions}(i). For Assumption ~\ref{ass: assumptions}(ii), loosely speaking, the spectral gap quantifies the rate at which a Markov chain mixes, that is, how rapidly it converges to its stationary distribution; see \cite{levin2009markov}[Ch. 12] for details. The assumption $0<\delta<1$ ensures that we exclude Markov chains that do not mix ($\delta =0$) or are independent $(\delta = 1)$. The justification for excluding the first case is that all chains in our experiments are empirically observed to mix well. The second case is excluded since this implies independent sampling, which is unrealistic.
Assumption~\ref{ass: assumptions}(iii), $\tau =\eta(\{+1\})= \Pr_{\overline{\pi}}(S=+1)\neq 0.5$, is reasonable because it can be imposed via the analytical expression in Lemma \ref{lemma: positive prob}; Section \ref{sec: tuning para} shows that the tuning procedure gives $\tau \gg 0.5$.

\begin{theorem}\label{thm: sum far from zero} Suppose that Assumption \ref{ass: assumptions} holds. Then, for any $\varepsilon > 0$ and any $0<c<|\mu|$, with $\mu$ in \eqref{eq:expected_sign}, there exists a constant $N_0$ such that,
\begin{align}\label{eq:pr_larger_than_cN}
    \textstyle \Pr_{\overline{\pi}}\left(\left|\sum_{i=1}^N s^{(i)}\right| > c N \right) \geq 1 - \varepsilon, \quad \text{for all } N > N_0.
\end{align}
Moreover, the convergence of the probability in \eqref{eq:pr_larger_than_cN} to 1 is exponentially fast in $N$.
\end{theorem}
\begin{remark}
    The constant $N_0$ is computable and depends on $\varepsilon$, $c$, $\tau$, and $\delta$; see the proof of the theorem in Section \ref{app: FiniteSampleBoundsIS} in the supplement for an expression. 
\end{remark}
The concentration result in Theorem \ref{thm: sum far from zero} ensures that, with high probability, the denominator of the estimator in \eqref{eq: final_est} remains bounded away from zero. Consequently, the finite-$N$ estimator $\widehat{E}_{\pi}(\psi(\Btheta))$ avoids divisions by near-zero values with high probability, and exhibits controlled variability in practice. The proof of the theorem is in Section \ref{app: FiniteSampleBoundsIS} in the supplement. Figure \ref{fig:N0_scenarios} illustrates how $N_0$ varies with $c$ for $\varepsilon=0.001$ (ensuring the probability in \eqref{eq:pr_larger_than_cN} is at least 0.999). The left panel corresponds to a well-tuned case ($\tau \gg 0.5$), while the right panel represents a case near the pathological scenario ($\tau \approx 0.5$). Each panel shows the results under three mixing scenarios; see the caption for details. In the well-tuned case, choosing $N\geq 1{,}000$ ($N\gg 1{,}000$ in our examples) comfortably avoids the pathological regime under moderate mixing $\delta=0.30$ (chains mix moderately well in our examples with a random walk proposal). In contrast, in the near-pathological case, the required $N$ to remain far from the pathological regime becomes prohibitively large (note the log-scale on the y-axis), even under strong mixing. This highlights the importance of our tuning strategy to avoid such cases (we ensure $\tau \approx 0.99$ in our examples).

\begin{figure}[ht]
    \centering
    \includegraphics[width=1\linewidth]{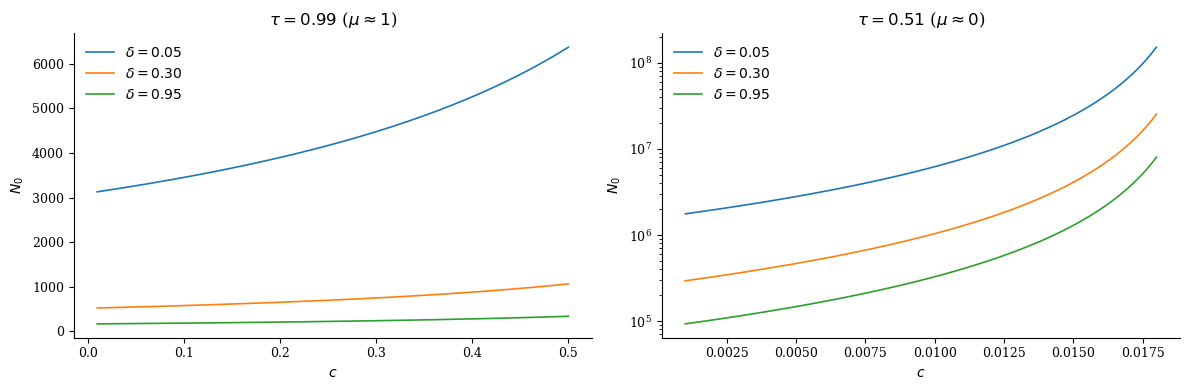}
    \caption{Values of $N_0$ from Theorem \ref{thm: sum far from zero} as a function of $c$ ($0<c<|\mu|$) when $\varepsilon=0.001$. The left and right panels correspond to $\tau=0.99$ and $\tau=0.51$, respectively. Three mixing scenarios for the Markov chain $s^{(i)}$ are shown: $\delta = 0.05$ (slow), $\delta = 0.3$ (moderate), and $\delta = 0.95$ (strong). The right panel uses a log-scale for the vertical axis.}
    \label{fig:N0_scenarios}
\end{figure}

Finally, to make the signed pseudo-marginal algorithm for sampling from \eqref{eq:absolute_target} more efficient, we correlate the estimators at the current and proposed draws to decrease the variability of the difference of the log of the likelihood estimators. This provides a substantial advantage over the standard pseudo-marginal method that proposes $\Bu$ independently in each iteration \citep{deligiannidis2018correlated, tran2016block}. We follow the approach in \cite{tran2016block}, where the correlation is induced by blocking the random numbers and only updating one of the blocks when evaluating the likelihood at the proposed value, while keeping the rest of the blocks fixed. The BP estimator uses the random numbers $\Bu_l$ to estimate $\xi_l$, $l=1,\dots,\lambda$, and group them as $\Bu = (\Bu_1,\dots,\Bu_\lambda) = \Bu_{1:\lambda}$. Note that each $\Bu_l$ includes random numbers of different sizes depending on the value of $\chi_l \sim \mathrm{Pois}(m)$. If the number of blocks $\lambda$ is sufficiently large, the correlation $\rho$ between the logarithms of the likelihood estimators evaluated at the current and proposed draws is approximately $1-1/\lambda$ \citep{quiroz2021block}.  We can adjust the number of blocks $\lambda$ to achieve a prespecified correlation between the log of the likelihood estimates.

\begin{algorithm}
\caption{One iteration of the signed block PMMH update with the BP estimator.}\label{algo: PMMH_DI}
\begin{algorithmic}[1]
\State \textbf{Input:}  Current values of $\nu,\Btheta, \Bu_{1:\lambda}$.
\State \textbf{Output:} Updated values of $\nu,\Btheta,\Bu_{1:\lambda}$ and  $\mathrm{sign}(\widehat{\pi}(\Btheta, \Bu_{1:\lambda}, \nu |\By))$.
\State Generate $\Bu'_{1:\lambda} \gets  \Bu_{1:\lambda}$ from $q(\Bu'_{1:\lambda}|  \Bu_{1:\lambda})$ by updating one block of random numbers. 

\State Generate $\Btheta'$ from $q(\Btheta'|\Btheta)$.

\State Compute the unbiased estimates, $\widehat{Z}(\Btheta')$, and use them to construct the BP estimator via \eqref{eq: block_pois_est}. The proposal distribution of the auxiliary variable $\nu'$ is an exponential distribution with mean $1/\widehat{Z}_P(\Btheta')$ : $$q(\nu'|\Btheta', \Bu') =  \widehat{Z}_P(\Btheta') \exp(-\nu' \widehat{Z}_P(\Btheta')),$$ where $\widehat{Z}_P(\Btheta')$ is the average of the $\widehat{Z}(\Btheta')$s used in the BP estimator.

\State Set $\Btheta \gets \Btheta'$, $\nu \gets \nu'$ and $\Bu_{1:\lambda} \gets \Bu'_{1:\lambda}$ with probability
\begin{equation} \label{eq: ar_PMMH}
    \min \left\{1, 
    \frac{\lvert \widehat{\pi}(\Btheta', \nu', \Bu'_{1:\lambda}|\By) \rvert}{\lvert \widehat{\pi}(\Btheta, \nu,\Bu_{1:\lambda}|\By) \rvert}
    \dfrac{q(\Btheta| \Btheta')q(\Bu_{1:\lambda}|\Bu'_{1:\lambda})}{q(\Btheta'|\Btheta)q(\Bu_{1:\lambda}'|\Bu_{1:\lambda})}
     \dfrac{\widehat{Z}_P(\Btheta)}{\widehat{Z}_P(\Btheta')}
    \dfrac{\exp(-\nu \widehat{Z}_P(\Btheta))}{\exp(-\nu' \widehat{Z}_P(\Btheta'))} \right\}, 
\end{equation}
 where $$\widehat{\pi}(\Btheta,\nu, \Bu_{1:\lambda}|\By) = \widehat{\exp}(-\nu Z(\Btheta)|\Bu_{1:\lambda}) f(\By|\Btheta) \pi(\Bu_{1:\lambda})\pi(\Btheta) p^{-1}(\By),$$ and $\widehat{\exp}(-\nu Z(\Btheta)|\Bu_{1:\lambda})$ is obtained by the BP estimator.  
 
\State Record  $s = \mathrm{sign}(\widehat{\pi}(\Btheta, \Bu_{1:\lambda}, \nu|\By))$  which is also the sign of $\widehat{\exp}(-\nu Z(\Btheta)|\Bu_{1:\lambda})$. 
\end{algorithmic}
\end{algorithm}

Algorithm \ref{algo: PMMH_DI} outlines one iteration of our method when using an exponential proposal for the auxiliary variable $\nu$\footnote{An anonymous reviewer suggested an alternative deterministic proposal for $\nu$, which is outlined in Section \ref{sec: discussion of algo} in the supplement.}. Rewriting
\eqref{eq: ar_PMMH} as\footnote{\citet{tran2016block} show that $\pi(\Bu'_{1:\lambda})q(\Bu_{1:\lambda}|\Bu'_{1:\lambda})=\pi(\Bu_{1:\lambda})q(\Bu'_{1:\lambda}|\Bu_{1:\lambda})$ for a block proposal.}
 \begin{align}
     \frac{\pi(\Btheta') f(\By|\Btheta')}{\pi(\Btheta) f(\By|\Btheta)} \times \dfrac{q(\Btheta|\Btheta')}{q(\Btheta'|\Btheta)} \times \frac{\widehat{Z}_P^{-1}(\Btheta')}{\widehat{Z}_P^{-1}(\Btheta)}  \times \dfrac{|\widehat{\exp}(-\nu'  Z(\Btheta')|\Bu'_{1:\lambda})| /  \exp(-\nu' \widehat{Z}_P(\Btheta'))}{|\widehat{\exp}(-\nu Z(\Btheta)|\Bu_{1:\lambda} )| / \exp(-\nu \widehat{Z}_P(\Btheta))} \notag,
 \end{align}
 we observe that $$\frac{|\widehat{\exp}(-\nu'  Z(\Btheta')|\Bu'_{1:\lambda})| /  \exp(-\nu' \widehat{Z}_P(\Btheta'))}{|\widehat{\exp}(-\nu Z(\Btheta)|\Bu_{1:\lambda} )| / \exp(-\nu \widehat{Z}_P(\Btheta))}$$
 acts as a bias-correction for the bias induced when estimating $Z^{-1}(\Btheta')/Z^{-1}(\Btheta)$ by $\widehat{Z}_P^{-1}(\Btheta')/\widehat{Z}_P^{-1}(\Btheta)$. When forming the $\widehat{Z}_P$ estimators, we recommend using the average of the corresponding $\widehat{Z}(\Btheta)$s used in the BP estimator. This does not affect the unbiasedness property of the BP estimator and is computationally efficient, as the $\widehat{Z}(\Btheta)$s are already computed and the extra cost in obtaining the average is negligible.

\subsection{Tuning the signed block PMMH with the BP estimator}\label{sec: tuning para}

 \cite{pitt2012some} provide guidelines to tune the number of particles, i.e.\ the number of samples used in the likelihood estimation procedure, in a pseudo-marginal algorithm with a positive unbiased estimator to achieve an optimal trade-off between computing time and MCMC efficiency as measured by the integrated autocorrelation time (IACT), also known as the inefficiency factor (IF). Suppose that $\theta^{(j)}$, $j=1,2,\dots$, are the iterates after convergence of the MCMC and let $\vartheta^{(j)} = \psi(\theta^{(j)})$ be a scalar function of the iterates. Let $r_\tau$ be the correlation between $\vartheta^{(j)}$ and $\vartheta^{(j+\tau)}$. In pseudo-marginal algorithms, $r_\tau$ depends on the variance of the log of the likelihood estimator $\widehat{L}$, which we denote by $\sigma^2_{\log \widehat{L}}$. The inefficiency factor is defined as
 $$\mathrm{IF}(\sigma^2_{\log \widehat{L}}) = 1 + 2\sum_{\tau=1}^\infty r_\tau(\sigma^2_{\log \widehat{L}}).$$
A larger $\sigma^2_{\log \widehat{L}}$ results in a stickier chain and thus $\mathrm{IF}(\sigma^2_{\log \widehat{L}})$ is an increasing function of $\sigma^2_{\log \widehat{L}}$; see \cite{pitt2012some} for details. To also take the computing time into account when determining the number of particles to use in the estimation of $\log \widehat{L}$, \cite{pitt2012some} show that the number of particles is inversely proportional to $\sigma^2_{\log \widehat{L}}$ and define the computational time $\mathrm{CT}(\sigma^2_{\log \widehat{L}})=\mathrm{IF}(\sigma^2_{\log \widehat{L}})/\sigma^2_{\log \widehat{L}}$. This measure takes into account both the mixing of the chain (through IF) and the cost of computing the estimator (through the number of particles, which is inversely proportional to $\sigma^2_{\log \widehat{L}}$). \cite{pitt2012some} show that, under certain simplifying assumptions, $\sigma^2_{\log \widehat{L}}\approx 1$ is optimal, and thus the guideline is to choose the number of particles to achieve this.

\cite{quiroz2021block} extend the guidelines in \cite{pitt2012some} to cases when the likelihood estimator is not necessarily positive. The derivation of our guidelines follows those in \cite{quiroz2021block}, with modifications that account for a different estimator. Following \cite{quiroz2021block}, we tune the hyperparameters by minimising the following computational time (CT) 
\begin{equation}\label{eq: CT}
    \mathrm{CT} = m\lambda M \frac{\mathrm{IF}_{\abs{\widehat{\pi}} ,\psi s}\left( \sigma^2_{\log |\widehat{L}_B|} (m, \lambda, M |\gamma)\right)}{(2\tau(m,\lambda,M) -1)^2},
\end{equation}
where the dependence on $\Btheta$ is omitted for $\gamma$ (defined in \eqref{eq: gamma_expression} below) and $\tau(\cdot)=\Pr(\widehat{L}_B(\Btheta) >0)$ in Lemma \ref{lemma: positive prob}. The first term $m\lambda M$ in \eqref{eq: CT} is proportional to the expected cost per iteration since there are $\lambda$ blocks in total and each block includes $m$ estimates on average with $M$ Monte Carlo samples in each. The latter refers to the number of samples used to produce a single $\widehat{B}^{(h, l)}$ in \eqref{eq: block_pois_est}. The denominator in \eqref{eq: CT} shows that it is important to have a large proportion of estimates of the same sign and that having close to half of the estimates is detrimental for the CT. The numerator in \eqref{eq: CT} is the inefficiency factor, which measures the MCMC sampling efficiency of drawing $\psi$'s from the target distribution $|\widehat{\pi}|$. 
\citet[Section S2]{quiroz2021block} derives the specific form of the IF. The IF is determined by the variance of the log of the absolute likelihood estimator $\sigma^2_{\log \abs{\widehat{L}_B}}$ (recall the discussion when tuning using a positive likelihood estimator above) in Lemma \ref{lemma: log variance}, which in turn depends on the hyperparameters $m,\lambda, M$, and in addition $\gamma$. We define $\gamma(\Btheta)$ as the variance of a single Monte Carlo sample, say $-\nu \widehat{Z}_i(\Btheta)$, i.e. such that $\sigma^2_{\widehat{B}}=\gamma(\Btheta)/M$ 
for the Monte Carlo estimate $\widehat{B}=\sum_{i=1}^M-\nu \widehat{Z}_i(\Btheta)/M$ based on $M$ independent Monte Carlo samples. Note that $\gamma(\Btheta)$ does not depend on $M$; however, it depends on $\nu$ (recall $\nu$ is treated as non-random) as
\begin{align}\label{eq: gamma_expression}
    \gamma(\Btheta) = \mathrm{Var}(-\nu \widehat{Z_i}(\Btheta)) &  =  \nu^2 \mathrm{Var} (\widehat{Z_i}(\Btheta)) \nonumber \\ 
    & = \frac{2}{Z(\Btheta)^2}\mathrm{Var}(\widehat{Z_i}(\Btheta)).
\end{align}
Similarly to when we replaced $\nu$ by its expected value when determining $a_{\mathrm{sub}}$ in Section \ref{sec: blockPois}, $\nu^2$ in \eqref{eq: gamma_expression} is replaced by its second moment when tuning the hyperparameters. 

\begin{figure}
    \centering
    \subfloat[\centering $\rho =0$]{\includegraphics[width=0.8\linewidth]{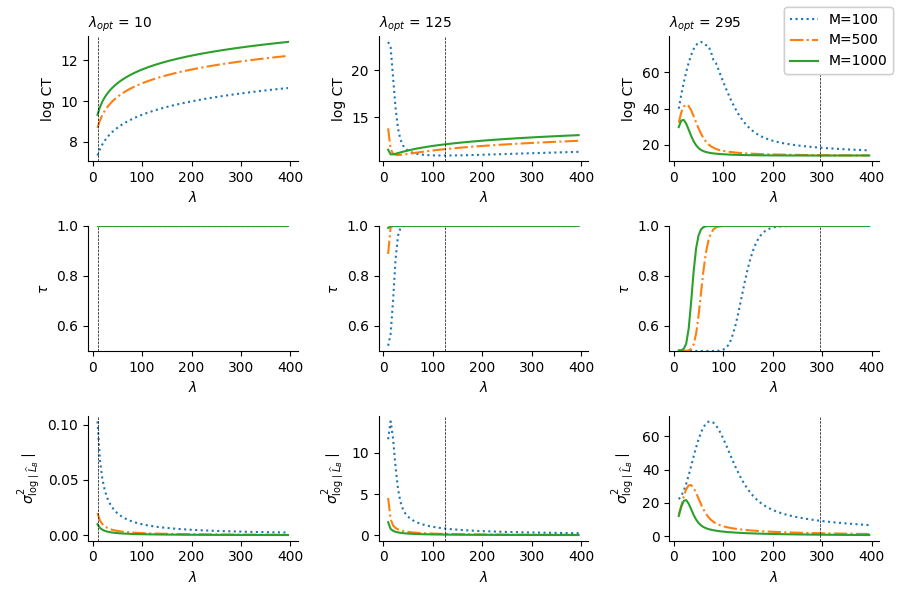} }%
    \qquad
    \subfloat[\centering $\rho = 0.99$]{\includegraphics[width=0.8\linewidth]{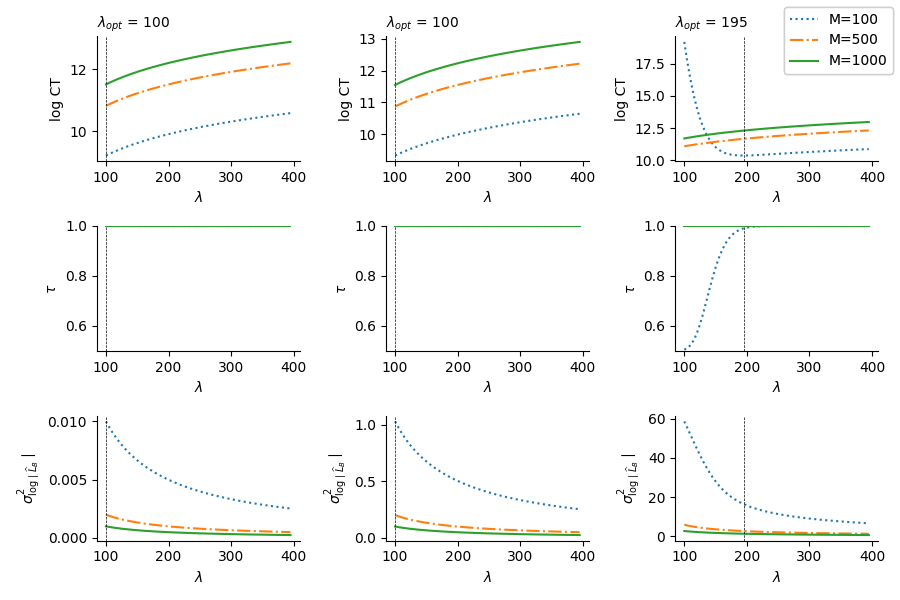}}%
    \caption{The effect of the number of blocks $\lambda$ on the logarithm of CT,  $\tau$ and $\sigma^2_{\log\left|\widehat{L}_B\right|}$. The Poisson parameter $m$ is fixed at 1 for each set of panels (a,b). The correlation term is set to $\rho = 0$ (upper panel), 0.99 (bottom panel).  Columns from left to right correspond to three different settings of $\gamma = 10^2$, $100^2$, and $500^2$. The top, middle and last rows of each panel show the log of the CT in \eqref{eq: CT}, the probability of obtaining a positive estimator $\tau(m,\lambda,M)$ (see Lemma \ref{lemma: positive prob}) and the variance of log of the absolute value of the likelihood estimate (see Lemma \ref{lemma: log variance}). The vertical line on each plot represents $\lambda_{\mathrm{opt}}$, the optimal $\lambda$, which minimises the log of the CT within each of the settings.}\label{fig:CT_dependence}
\end{figure}

Figure \ref{fig:CT_dependence} shows the effects of the number of blocks ($\lambda$) and Monte Carlo samples ($M$) on the logarithm of CT, $\tau$ and  $\sigma^2_{\log \abs{\widehat{L}_B}}$.  We consider the three cases $\gamma = 10^2, 100^2, 500^2$  (left to right columns respectively) which  show that  
the optimal $\lambda$ (corresponding to the minimal CT) varies with different values of $M$ and increases with $\gamma$ (top row).
The minimum CT is associated with a high probability of a positive estimator ($\tau$) (middle row). The last row indicates that $\sigma^2_{\log \abs{\widehat{L}_B}}$ decreases as a function of $\lambda$ for large $\lambda$. Comparing the top nine panels with the bottom nine, a high correlation $\rho = 0.99$, reduces $\lambda_{\mathrm{opt}}$ from 295 (no correlation, $\rho=0$) to 195 for $\gamma = 500^2$. Conversely, $\rho = 0.99$ requires at least 100 blocks. So when the variance $\gamma$ is small, introducing a high correlation increases the CT as more blocks are required compared to the uncorrelated case. Our implementation follows the approach in \cite{tran2016block} which sets the correlation $\rho$ to a value close to 1. Comparing the first row of the top panel (a) in Figure \ref{fig:CT_dependence} with that of the bottom panel (b), shows that a high correlation significantly reduces the CT ($y$-axis is in log-scale) per iteration for large $\gamma$.  

We conclude that the tuning depends on $\gamma$ in \eqref{eq: gamma_expression}. For conservative tuning, we set $\gamma$ to a large value $\gamma_{\mathrm{max}}$ by using a grid search over possible $\Btheta$, meaning the process is calibrated to guard against worst‑case scenarios. Note that grid search scales poorly beyond two dimensions; in such cases, Bayesian optimisation \citep{Shahriari2016bayesopt} offers a practical alternative for the parameter spaces typical of doubly intractable problems. The tuning process starts with fixed values of $\lambda$ and $m$ to find the optimal value for $M$ by minimising \eqref{eq: CT}. In Figure \ref{fig:optimal_M}, we fix the values of $\lambda$ and $m$, with $\lambda = 50, 100$ (the corresponding $\rho$ are 0.98 and 0.99 respectively), and $m = 1$. A standard optimiser is used to find the optimal value $M_{\mathrm{opt}}$ for each of the $\gamma$. The dots in the left panel of Figure \ref{fig:optimal_M} plot $M_{\mathrm{opt}}$ for various values of $\sqrt{\gamma}$\footnote{Taking the square root puts the quantity on the standard deviation scale.}. The figure shows that $M_{\mathrm{opt}}$ increases as a function of  $\sqrt{\gamma}$. The right panel shows how the minimised log CT increases as a function of $\sqrt{\gamma}$. To estimate the relationship between $M_{\mathrm{opt}}$ and $\sqrt{\gamma}$, a quadratic polynomial is fitted to the points in the left panel.

\begin{figure}[ht]
    \centering
    \includegraphics[width=0.8\linewidth]{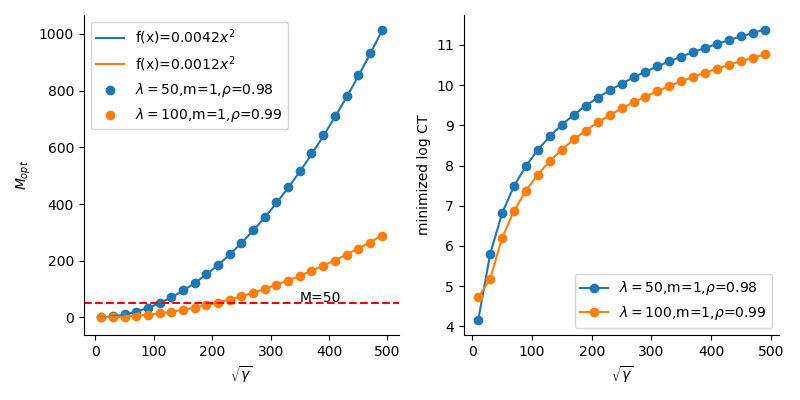}
    \caption{Left panel: The optimal value $M_{\mathrm{opt}}$ as a function of $\sqrt{\gamma}$. The lines are quadratic polynomials fitted to the scattered dots. The horizontal dashed line represents the threshold $M=50$, which is the minimal number of blocks required in the algorithm.  Right panel: The minimised log CT as a function of $\sqrt{\gamma}$. }
    \label{fig:optimal_M}
\end{figure}

The tuning below is based on $\gamma_{\mathrm{max}}$, leading to a conservative tuning of $M_{\mathrm{opt}}$. 
\begin{enumerate}
    \item [1.] Obtain a rough understanding of the support of the posterior distribution of $\Btheta$. Section \ref{subsec:bias_corrected} in the supplement proposes an approximate method that relies on a normality assumption of $\widehat{Z}(\Btheta)$, which can be used in the tuning phase. Alternatively, we can apply another approximate method tailored for the specific problem (e.g.\ the variational mean-field approximation of the Ising model in \citealp{Jain2018meanfield}). It is also possible to optimise the posterior distribution by plugging the biased estimator ($1/\widehat{Z}(\Btheta)$). 
    \item [2.] Estimate $\gamma(\Btheta)$ in \eqref{eq: gamma_expression} for the $\Btheta$ candidates from Step 1. The estimator $\widehat{Z_i}(\Btheta)$ replaces the unknown $Z(\Btheta)$.  
    \item [3.] Obtain the maximum value $\gamma_{\mathrm{max}}(\Btheta)$ of $\gamma(\Btheta)$ from Step 2. A sensible first tuning attempt sets $\lambda = 100, m= 1, \rho = 0.99$ and $M_{\mathrm{opt}} = \max \{50, 0.0012 \times \gamma_{\max}(\Btheta) \}$. 

    If  $\gamma_{\max}(\Btheta)$ is small to moderately large, e.g.\ $\gamma_{\max}(\Btheta) < 100^2$, having many blocks increases CT (e.g.\ top left and top middle panels in Figure \ref{fig:CT_dependence}). In this case, a weaker correlation also produces an efficient algorithm with smaller CT. A suitable setting in this scenario is $\lambda = 50, m = 1, \rho = 0.98$ and $M_{\mathrm{opt}}=\max\{50, 0.0042 \times \gamma_{\max}(\Btheta)\}$.
    
    For an even smaller $\gamma(\Btheta)$, the correlation can be relaxed further. In the Ising model example, $\lambda = 10$ is sufficient when the variability is low; see Section \ref{sec: ising model} and Section \ref{app: ising model} in the supplement.
    
\end{enumerate}

The tuning is theoretically sub-optimal as it is based on simplifying assumptions in its derivation and practical implementation. Section \ref{subsec:simplified_lb} in the supplement presents a simulation study demonstrating that tuning based on the simplified lower bound closely approximates the intractable one, which is reassuring. In our experience, the main objective of the tuning is met by our guidelines, namely: avoiding the pseudo-marginal chain getting stuck. Finally, we note that the
algorithm is still exact under the simplified tuning.

\section{Simulation study} \label{sec: demo}

This section empirically compares the performance of our method on two simulated examples with some competing methods. All our comparisons are on an empirical level rather than a theoretical level. Theoretical comparisons are either impossible because the competing methods are analytically intractable or outside the scope of this paper. Examples of the former are that expressions of the variance of the log of the absolute value or the probability of a positive estimator are not readily available for Russian roulette estimators. An example of the latter includes extending the ordering of pseudo-marginal MCMC chains \citep{andrieu2016establishing} to a signed pseudo-marginal setting.

The first example is the Ising model, which is the usual benchmark example for doubly intractable problems, as perfect sampling is efficient for this model on small grids, and thus the exchange algorithm \citep{murray2006mcmc} can provide a ground truth. We use the same settings as the survey paper \cite{park2018bayesian}, in which the Ising model serves as the main benchmark example. The second example considers the Kent distribution, where the intractable normalising function is an infinite sum. Unlike the Ising model, efficient perfect sampling is hard for this model, and thus the exchange algorithm fails. Section \ref{app: ex and rr intro} in the supplement provides the implementation details for the exchange algorithm. To the best of our knowledge, exact Bayesian inference has not been considered for the Kent distribution due to its intractability. 

Our method, abbreviated BP, is a correlated signed pseudo-marginal method that utilises the block-Poisson (BP) estimator and is compared to the signed pseudo-marginal methods introduced in \citet{lyne2015russian}. Two implementations of the latter are considered, which differ by their use of auxiliary variables and what form of the normalising function they estimate. RR-aux, uses a Russian roulette (RR) estimator of the exponent of the normalising function and requires auxiliary variables (-aux) to turn the reciprocal into an exponent as in \eqref{eq:augmented_posterior}. RR, on the other hand, estimates the reciprocal directly and does not require auxiliary variables. Section \ref{app: ex and rr intro} in the supplement provides the implementation details for both Russian roulette methods.

Table \ref{tab:methods summary} lists the requirements and features of the exact methods considered in our paper. The multiple observations case is relevant for the Kent distribution example, and Section \ref{eq:Kent_posterior} in the supplement discusses its scalability for the different estimators. Section \ref{app:computational_aspects} in the supplement provides a detailed computational analysis, including the structural differences between the BP, RR, and RR-aux estimators, and the derivation of their arithmetic complexities. In summary, the BP estimator offers computational advantages due to its independent product structure that enables efficient parallelisation and vectorisation via single instruction, multiple data (SIMD) instructions \citep[Ch. 4]{warne2022BAtutorial,hennessy2011computer}. The Russian roulette estimators RR and RR-aux, in contrast, involve a sum of dependent, nested products, and thus cannot exploit parallelisation or vectorisation as effectively, due to their sequential computation pattern. As shown in Section \ref{app:computational_aspects}, achieving a similar parallelism as BP by recomputing each nested product---thereby creating an independent sum---incurs a quadratic computational cost in the number of sum terms, compared to the linear cost of the BP product. 

Finally, we emphasise again that the signed pseudo-marginal methods are exact in providing simulation-consistent estimates of expectations under the exact doubly intractable posterior. By contrast, the exchange algorithm provides the usual exact inference MCMC methods do, i.e.\ samples from the invariant distribution (the posterior density) after burn-in. However, for a finite number of MCMC iterations, the exchange algorithm might not have properly converged due to mixing problems, which can arise if sampling from the likelihood is inefficient. In the Ising example, this occurs for large grids, and in such cases, signed pseudo-marginal algorithms are tractable alternatives \citep{lyne2015russian}. For the Kent distribution example, we find that the accept/reject method in \cite{kent2013new} does not give efficient sampling from the likelihood in any of our settings.

\begin{table}[ht]
\centering
\begin{tabular}{p{6cm}cccc}
\toprule
Requirements and features  / Method & BP  & RR &  RR-aux &  Exchange \\
         \midrule
Auxiliary variable(s) required &  Yes  & No  & Yes  &Yes  \\
\makecell[{{l}}]{Sampling from the likelihood without \\ knowing the normalising function} & No  & No & No & Yes\\
\makecell[{{l}}]{Correlated pseudo-marginal (PM) \\
\,\,\,\,\,\,\,\,\,\,\,(only PM methods)} & Yes  &  No& No & NA \\
\makecell[{{l}}]{Estimator scales with multiple\\ observations} & Yes  & No & Yes & NA \\
\makecell[{{l}}]{Estimator utilises vectorisation \\ and parallelisation} & Yes & No & No & NA\\
\bottomrule \\
\end{tabular}
\caption{Requirements (first 2) and features (bottom 3) for the methods. BP= block-Poisson, RR = Russian roulette, RR-aux = Russian roulette with an auxiliary variable. NA stands for not available.}
    \label{tab:methods summary}
\end{table}

\subsection{The Ising model} \label{sec: ising model}
The Ising model \citep{lenz1920beitrvsge, ising1925beitrag} has widespread applications, such as understanding phase transitions in thermodynamic systems \citep{fredrickson1984kinetic}, interactive image segmentation in vision problems \citep{kolmogorov2004energy} and modelling small-world networks \citep{herrero2002ising}. It is the typical benchmark example in the literature to evaluate methods for tackling the doubly intractable problem; see e.g.\  \cite{moller2006efficient, lyne2015russian, atchade2013bayesian, park2018bayesian}. However, most of the existing methods use auxiliary variable approaches, as it is feasible to draw observations from the likelihood function perfectly, so-called perfect sampling, for moderately small grids. The signed pseudo-marginal methods, such as RR and our approach, do not require perfect sampling, which makes them applicable to more general problems. \cite{lyne2015russian} report that their pseudo-marginal approach handles larger grids than the exchange algorithm.

Recall Example \ref{Ex: Ising} in Section \ref{sec: doubly intract prob}: consider an $L \times L$ lattice with binary observations $y_{ij}$ of row $i$ and column $j$  ($y_{ij} \in \{ -1,1\})$. The model is
\begin{align}
\begin{aligned}
p(\By|\theta) &= \frac{1}{Z(\theta)}\exp ( \theta S(\By)),\\
 \text {with } S(\By) &= \sum_{i=1}^{L} \sum_{j=1}^{L-1} y_{i,j} y_{i,j+1} + \sum_{i=1}^{L-1} \sum_{j=1}^L y_{i,j} y_{i+1,j}\\
\text{and } Z(\theta)  &= \sum_{\By} \exp(\theta S(\By));\notag
\end{aligned}
\end{align}
$\theta$ is a scalar parameter and $S(\By)$  imposes spatial dependence; a stronger interaction between observations is associated with a larger $\theta$. Obtaining $Z(\theta)$ is computationally expensive with a sum over $2^{L^2}$ possible configurations. 
The data simulations are conducted using perfect sampling \citep{propp1996exact}, which samples exactly without evaluating the normalising function. Perfect sampling uses coupling to guarantee that the samples are generated from a Markov chain which has already converged to its equilibrium distribution. Following the settings in \cite{park2018bayesian}, two scenarios are considered on a $10 \times 10$ grid, with $\theta = 0.2$ and $ 0.43$; see Figure \ref{fig:ising model demo} for an illustration.

\begin{figure}[ht]
\centering
\includegraphics[width=0.8\linewidth]{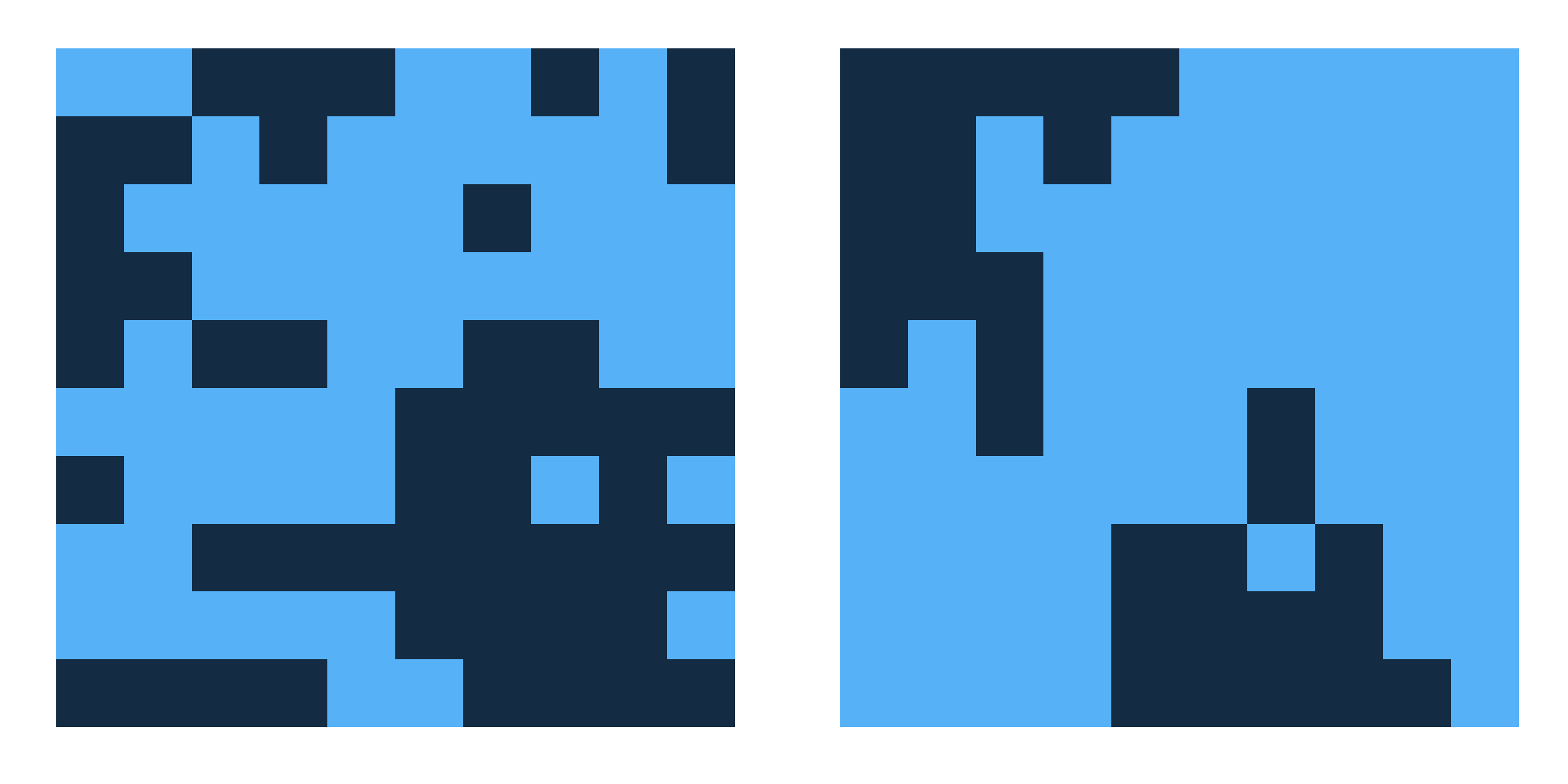}
\caption{Illustrating an Ising model on a $10\times 10$ grid. The samples are drawn using perfect sampling with $\theta=0.2$ (left) and $\theta=0.43$ (right). The light and dark blue squares correspond to the values 1 and $-1$. }
\label{fig:ising model demo}
\end{figure}

For all the algorithms considered, a uniform distribution on $[0,1]$ is selected as the prior for $\theta$. We adopt a random walk proposal centred at the current  $\theta$ with a step size $0.07$. The signed pseudo-marginal methods (RR, RR-aux, BP) require an unbiased estimator for $Z(\theta)$. We use annealed importance sampling (AIS) \citep{neal2001annealed} to obtain the estimate of $Z(\theta)$. The method starts by sampling from a tractable distribution  (the prior) and ends at the intractable target (the posterior) via a sequence of intermediate distributions.  The transitions between the distributions are completed via Gibbs updates, and the weights associated with the transitions finally constitute the normalising function of interest; see \cite{neal2001annealed} for details of AIS in general and Section \ref{app: ising model} in the supplement for its implementation for the Ising model.

To obtain the ``gold'' standard to evaluate the accuracy of the results,  we follow \cite{park2018bayesian}, where an exchange algorithm with $1{,}010{,}000$ iterations is performed. The first $10{,}000$ iterations are discarded for burn-in and the remaining iterates are thinned so that $10{,}000$ posterior samples remain. We use the same set of hyperparameters in RR and RR-aux; see details in Section \ref{app: ising model} in the supplement. For the tuning of the BP method, the two scenarios $\theta=0.2$ and $\theta=0.43$ result in different $\gamma_{\max}$ values, with the former yielding a smaller and the latter a larger value for the guidelines. We find that when $\theta = 0.2$, the AIS method gives a sufficiently low value for $\gamma_{\mathrm{max}}$ so that $\lambda = 10$ is appropriate.  When $\theta = 0.43$, on the other hand, the strong dependence leads to higher variability in $\widehat{Z}(\theta)$ (see Section \ref{app: ising model} in the supplement). We therefore increased the number of blocks to 50 for the BP estimator as per the tuning guidelines. To ensure a fair comparison, we also increased the number of particles (number of samples) in the importance samplers of AIS from 100 to 500 for RR and RR-aux to decrease the variance.

\begin{table}
\centering
\small{
\begin{tabular}{rccccccc}  
\toprule
\multicolumn{7}{c}{$\theta_{\mathrm{true}}$ = 0.2}\\
\midrule
Method & Mean & 95\%HPD & IACT & Time(s) & ESS/s &$\lambda$ & $\#$ particles \\ 
\midrule
\makecell[{{r}}]{Gold \\ (Exchange)} & 0.205 & (0.075, 0.337) & 1 & -& - & - & - \\ 
BP & 0.203 & (0.075, 0.334) & 8.39 & 606 & 3.9 & 10 & 100 \\  
RR & 0.199 & (0.072, 0.324) & 7.59 & 9,825 & 0.27 & - & 100 \\ 
RR-aux & 0.199 & (0.069, 0.332) & 9.16 & 7,447 & 0.32 & - & 100\\
\midrule
\multicolumn{7}{c}{$\theta_{\mathrm{true}}$ = 0.43}\\
\midrule
Method & mean & 95\%HPD & IACT & time(s) & ESS/s &$\lambda$ & $\#$ particles \\ 
\midrule
\makecell[{{r}}]{Gold \\ (Exchange)} & 0.433 & (0.330, 0.533) & 1.04 & - & - & - &-\\ 
BP & 0.435 & (0.324, 0.540) & 7.30 & 4,072 & 0.67 &50 & 100 \\ 
RR & 0.435 & (0.330, 0.548) & 7.16 &  10,016 & 0.29 &- & 500 \\ 
RR-aux & 0.633 & (0.633, 0.633) & NA &  87,525 & NA &- & 500 \\ 
\bottomrule
\\
\end{tabular}}
\caption{Results for the Ising model. All the chains, except for the ``gold standard'', ran for 20{,}000 iterations using the algorithms (Gold=exchange algorithm, BP= block-Poisson, RR = Russian roulette, RR-aux = Russian roulette with an auxiliary variable). For BP, RR and RR-aux, the mean estimates are corrected for the negative estimates using \eqref{eq: final_est}. The highest posterior density (HPD) intervals are calculated by the \texttt{coda} package in \texttt{R}. The IACT calculation is based on all the samples as the chains start at the true value. For BP, RR and RR-aux, the calculation of the IACT accounts for the negative estimates via  \eqref{eq: CT}. Time denotes the running time in seconds. ESS/s is the effective sample size per second.  For BP, $\lambda$ refers to the number of blocks; $\#$ particles is the number of particles (number of samples) used in the AIS. NA indicates not available due to sampler nonconvergence.} \label{table: ising inference}
\end{table}

Table \ref{table: ising inference} summarises the simulation results. When $\theta=0.2$, all the estimates are close to those of the gold standard, which is expected since the methods are simulation-consistent and no major mixing problems are encountered. The BP method has the smallest computing time and the second smallest IACT, and obtains a factor of roughly 14$\times$ improvement in terms of effective sample size per unit of computing time (in sec) compared to the competing Russian roulette approaches. When $\theta=0.43$, the results of the BP and RR match well with those of the gold standard; however, not for RR-aux due to severe mixing issues; the chain gets stuck after around 1,000 iterations as reported in Figure \ref{fig:Ising_traceplots} in the supplement. We were unable to find settings for RR-aux to work when $\theta= 0.43$, which emphasizes the importance of having tuning guidelines. In this example, BP is a factor of roughly 2$\times$ more efficient than RR in terms of effective sample size per computing time.

\subsection{The Kent distribution}\label{subsec: kent_dist}

Directional statistics involves the study of density functions defined on unit vectors in the plane or sphere. The Kent distribution, also known as the 5-parameter Fisher-Bingham distribution ($\mathrm{FB}_5$),  is an analogue to the bivariate normal distribution to model asymmetrically distributed data on a spherical surface \citep{kent1982fisher}. It has 5 parameters: $\boldsymbol{\gamma}_1, \boldsymbol{\gamma}_2,\boldsymbol{\gamma}_3,\beta$, and $\kappa$, where $\boldsymbol{\gamma}_1,\boldsymbol{\gamma}_2,\boldsymbol{\gamma}_3$ form a 3-dimensional orthonormal basis, representing the mean, major and minor axes; $\kappa$ is the concentration parameter,  and $\beta$ is a measure of its eccentricity, with the constraint $0 \leq \beta < \kappa/2$ to ensure that the distribution is unimodal. 

Recall Example \ref{Ex: Kent} in Section \ref{sec: doubly intract prob}: the density of the Kent distribution is 
\begin{equation*}
     f(\By|\boldsymbol{\gamma_1},\boldsymbol{\gamma_2},\boldsymbol{\gamma_3},\beta,\kappa) = \frac{1}{c(\kappa,\beta)} \exp\left\{\kappa \boldsymbol{\gamma_1}^\top \cdot \By + \beta \left[(\boldsymbol{\gamma_2}^\top \cdot \By)^2 - (\boldsymbol{\gamma_3}^\top \cdot\By)^2\right]\right\},
\end{equation*}
where $\By \in \mathbb{R}^3, \| \By\|  = 1$. The normalising function is
\begin{equation*}
c(\kappa, \beta) = 2\pi \sum_{j=0}^{\infty} \dfrac{\Gamma(j+0.5)}{\Gamma(j+1)} \beta^{2j} (0.5\kappa)^{-2j-0.5} I_{2j+0.5}(\kappa),  
\end{equation*}
where $I_\nu(\cdot)$ is the modified Bessel function.

The normalising function is an intractable infinite sum. Due to the complex form of the density function,  \cite{kent1982fisher} proposes a consistent moment estimator of the parameters. The moment estimation of $\boldsymbol{\gamma}_i, i =1,2,3 $ is independent of $\beta$ and $\kappa$. Estimating $\beta$ and $\kappa$ requires an approximation that utilises the limiting case when $2\beta/\kappa$ is small or $\kappa$ is large, provided that the moment estimates of the $\boldsymbol{\gamma}_i$ are available.  Alternatively, $\kappa$ and $\beta$ can be obtained numerically. \cite{kume2005saddlepoint} adopt saddle point techniques to obtain the approximation for the normalising function. \cite{kasarapu2015modelling} uses the Bayesian framework to model a mixture of $\mathrm{FB}_5$ distributions. The infinite sum in $c(\kappa,\beta)$ is truncated in the sense that the successive term to be added is less than a prefixed threshold. However, this approach induces a bias in the computed distribution relative to the true posterior, i.e.\ this approach is not simulation-consistent. In contrast, the signed block PMMH with the BP, RR and RR-aux estimators provide exact Bayesian inference for the parameters.

We use the approach by \cite{papaspiliopoulos2011monte} to obtain an unbiased estimator for $c(\kappa,\beta)$.  Rewrite $c(\kappa,\beta) $ as $ \sum_{j=0}^{\infty} \phi_j (\kappa,\beta)$; then the estimator  $\widehat{c}(\kappa,\beta) = \phi_k/q_k$ is unbiased, where $k$ is a non-negative discrete random variable with probability mass function $q_k$. Either a Poisson or a geometric distribution is suitable, as $k$ is a non-negative integer.  It is straightforward to verify that $E(\widehat{c}(\kappa,\beta)) = \sum_k \phi_k/q_k \times q_k = c(\kappa,\beta)$. As $\phi_j(\kappa,\beta)$ is a decreasing function in $j$, to reduce the variability, we compute the first $K$ terms exactly and perform a truncation of the remaining terms. Specifically,  $c(\kappa,\beta)$ is decomposed as $$\sum_{j=0}^{K-1} \phi_j (\kappa,\beta) + \sum_{j=K}^{\infty} \phi_j (\kappa,\beta). $$ 
The first sum is evaluated, and the second is estimated via the truncation procedure described above.

We apply the parameterisation in \cite{kasarapu2015modelling}, where the orthonormal basis $\boldsymbol{\gamma}_1,\boldsymbol{\gamma}_2,\boldsymbol{\gamma}_3$ is reparameterised as $\psi \in [0,\pi], \alpha \in [0,2\pi], \eta \in [0,\pi]$. An adaptive Gaussian random walk proposal is used for all the parameters with the algorithm proposed in \cite{garthwaite2016adaptive}. To achieve this, we further transform $\psi,\alpha,\eta$ into $\psi^*, \alpha^*, \eta^*$ which take unconstrained values using the transformations
$$\psi^* = \log\left( \frac{\psi}{\pi - \psi} \right);  \alpha^* = \log \left( \frac{\alpha}{2\pi - \alpha} \right) \quad {\rm and} \quad \eta^* = \log \left( \frac{\eta}{\pi - \eta }\right).$$ 
We also work with the logarithms of $\beta$ and $\kappa$ as they are unconstrained. We follow \cite{dowe1996mml} and set the prior for $\kappa$ as $4\kappa^2/\pi(1+\kappa^2)^2$. For a given $\kappa$, the prior for $\beta $ is uniform on $[0,\kappa/2)$. The priors for $\psi$, $\alpha$ and $\eta$ follow \cite{kasarapu2015modelling}. The joint prior on all the parameters, $\psi,\alpha,\eta,\beta$ and $\kappa$ is
\begin{equation*}
\pi(\psi,\alpha,\eta,\beta,\kappa) = \dfrac{2\kappa \sin \alpha}{\pi^3 (1+\kappa^2)^2} \mathbbm{1}(0\leq 2\beta/\kappa < 1).
\end{equation*}

In the simulation, we generate $n$ observations $\By$ from $\mathrm{FB}_5,$  i.e.\ $\By_1, \dots \By_n$, for a variety of settings for $\beta$ and $\kappa$. This is different from the Ising model, in which only a single data observation (matrix) $\By$ is generated. Section \ref{eq:Kent_posterior} in the supplement outlines the resulting expression of the doubly intractable posterior. The data generation is performed by the R package \texttt{Directional}, which implements the acceptance-rejection method in \cite{kent2013new}. We set $n = $ 10, 100, 1{,}000 in combination with $\beta/\kappa  = $ 0.01, 0.25, 0.49, with $\kappa$ fixed as 5. The lower and upper bounds for $\beta/\kappa$ are 0 and 0.5 to ensure the unimodality of the data \citep{kent1982fisher}. In addition to the Bayesian methods (BP, RR, RR-aux and Exchange) and the moment estimation method, we also consider the maximum likelihood estimation (MLE) method, which uses the saddle point technique \citep{kume2005saddlepoint} to approximate the likelihood. We refer the reader to Section \ref{app: kent distr} in the supplement for the implementation details.  

\begin{table}[]
    \centering
    \small{\begin{tabular}{r l l r|l l r|l l r}
    \toprule
               & \multicolumn{9}{c}{$\beta/\kappa$ = 0.01} \\
           \midrule
               & \multicolumn{3}{c}{$n=10$} & \multicolumn{3}{c}{$n=100$} &\multicolumn{3}{c}
 {$n=1{,}000$}\\
  & $\mathrm{RMSE}_\kappa$ & $\mathrm{ESS}_{\kappa} $ & $\mathrm{ESS}_{\kappa}$/s & $\mathrm{RMSE}_\kappa$ & $\mathrm{ESS}_{\kappa} $ & $\mathrm{ESS}_{\kappa}$/s   & $\mathrm{RMSE}_\kappa$ & $\mathrm{ESS}_{\kappa} $ & $\mathrm{ESS}_{\kappa}$/s \\
    \cmidrule(lr){2-4}\cmidrule(lr){5-7}\cmidrule(lr){8-10} 
BP& 2.46 & 130.9 &14.4 & 0.51 & 104.1 &11.3 & 0.17 & 137.0 &14.0\\
RR& 2.89 & 13.0 &0.3 & 0.41 & 4.1 & $<$0.1 & 0.27 & 3.8 & $<$0.1\\
RR-aux& 2.44 & 156.2 &13.7 & 0.51 & 138.2 &4.2 & 0.17 & 44.0 &0.8\\
Moment& 4.17 & - &- & 0.53 & - &-& 0.17 & - &-\\
MLE& 4.40 & - &- & 0.54 & - &-& 0.18 & -&-\\
Exch.& 4.60 & 88.4 &7.2 & 0.59 & 96.6 &6.8 & 0.21 & 87.1 &3.7\\
\midrule
           & \multicolumn{9}{c}{$\beta/\kappa$ = 0.25} \\
           \midrule
               & \multicolumn{3}{c}{$n=10$} & \multicolumn{3}{c}{$n=100$} &\multicolumn{3}{c}
 {$n=1{,}000$}\\
  & $\mathrm{RMSE}_\kappa$& $\mathrm{ESS}_{\kappa}$ &$\mathrm{ESS}_{\kappa}$/s  & $\mathrm{RMSE}_\kappa$& $\mathrm{ESS}_{\kappa}$ &$\mathrm{ESS}_{\kappa}$/s  & $\mathrm{RMSE}_\kappa$& $\mathrm{ESS}_{\kappa}$ &$\mathrm{ESS}_{\kappa}$/s \\
    \cmidrule(lr){2-4}\cmidrule(lr){5-7}\cmidrule(lr){8-10} 
BP & 2.04 & 118.5 &13.1 & 0.54 & 120.2 &13.2 & 0.18 & 133.4 &13.6\\
RR & 2.19 & 15.5 &0.4 & 0.66 & 4.1 & $<0.1$ & 0.45 & 2.1 &$<0.1$\\
RR-aux& 2.02 & 150.7 &13.4 & 0.54 & 162.4 &5.0 & 0.17 & 30.2 &0.5\\
Moment& 3.45 & - &- & 0.53 &-&- & 0.18 & - &-\\
MLE & 4.04 & - &- & 0.59 & - &- & 0.26 & - &-\\
Exch.& 4.94 & 89.0 &7.1 & 2.35 & 82.8 &5.6 & 1.94 & 70.5 &2.9\\
\midrule
           & \multicolumn{9}{c}{$\beta/\kappa$ = 0.49} \\
           \midrule
               & \multicolumn{3}{c}{$n=10$} & \multicolumn{3}{c}{$n=100$} &\multicolumn{3}{c}
 {$n=1{,}000$}\\
  & $ \mathrm{RMSE}_\kappa$& $\mathrm{ESS}_{\kappa}$ &$\mathrm{ESS}_{\kappa}$/s   & $\mathrm{RMSE}_\kappa$& $\mathrm{ESS}_{\kappa}$ &$\mathrm{ESS}_{\kappa}$/s   & $\mathrm{RMSE}_\kappa$& $\mathrm{ESS}_{\kappa}$ &$\mathrm{ESS}_{\kappa}$/s  \\
    \cmidrule(lr){2-4}\cmidrule(lr){5-7}\cmidrule(lr){8-10} 
BP& 1.67 & 99.6 &11.0 & 0.57 & 107.6 &11.8 & 0.18 & 69.6 &7.1\\
RR& 2.07 & 10.2 &0.2 & 0.64 & 3.7 & $<$0.1 & 0.59 & 3.9 & $<$0.1\\
RR-aux& 1.68 & 111.4 &10.0 & 0.56 & 150.7 &4.6 & 0.21 & 12.9 &0.2\\
Moment& 2.82 & - &- & 0.56 & - & - & 0.48 & - &-\\
MLE& 3.75 & - &- & 0.68 & - & - & 0.21 & - &-\\
Exch.& 5.61 & 75.3 &6.0 & 5.03 & 109.0 &7.1 & 4.36 & 103.4 &4.0\\
\bottomrule
\\
    \end{tabular}}
\caption{Simulation results when estimating $\kappa$ using 100 independent replications from an $\mathrm{FB}_5$ distribution. The RMSE numbers are with respect to the true value ($\kappa = 5$). The $\mathrm{ESS_\kappa}$/s is the effective sample size (altered with sign correction) per second. }
\label{tab:FB5_simulation}
\end{table}

Table \ref{tab:FB5_simulation} shows the root mean squared error (RMSE), effective sample size (ESS), and ESS per second (ESS/s), for all methods based on 100 independent data replicates when estimating $\kappa$. The corresponding tables for $\beta$ and $\beta/\kappa$ are provided in Section \ref{app: results of FB5_sim} in the supplement. For BP, RR and RR-aux, the RMSE is calculated using the posterior mean with the sign correction in \eqref{eq: final_est}. We now summarise the results for each metric.
 
For the $\mathrm{RMSE}$ for $\kappa$, the BP and RR-aux perform equally well in all cases of $\beta/\kappa$ with minor differences. The RR has a higher RMSE compared to BP and RR-aux for all but one case ($n=100$ and $\beta/\kappa=0.01$), ranging from 5\% to as large as 250\%, attributable to its poor mixing (low $\mathrm{ESS}_\kappa$) reported in the next paragraph. Note that the exchange method, which serves as the gold standard in the previous example, has the highest RMSE in all cases but one ($n=1{,}000$ and $\beta/\kappa=0.01$). This is likely attributed to the difficulty of efficiently simulating from the likelihood in the Kent distribution. We apply the \texttt{rkent} function of the \texttt{Directional} package, which uses the acceptance-rejection method \citep{kent2013new}. Both the moment estimation method and the MLE estimation method behave similarly: when $n=10$, they produce a large $\mathrm{RMSE}_\kappa$, but as $n$ increases to 100 or 1,000, their RMSE decreases and approaches those of Bayesian methods. The exception is the boundary case $\beta/\kappa = 0.49$, where the moment estimation method performs poorly; \cite{kent1982fisher} note that the moment estimator should avoid the limiting case $\beta/\kappa =0.5$. Figure \ref{fig:kent_uni_and_bimodal} in the supplement shows examples of the Kent distribution near the limiting case on both sides, i.e.\ $\beta/\kappa \lessgtr 0.5$.

For the $\mathrm{ESS}_{\kappa}$,  RR-aux and BP have the largest values. RR-aux has larger values than BP for a small number of observations ($n=10, 100$), ranging from approximately 10-20\%. However, when $n=1{,}000$, the $\mathrm{ESS}_{\kappa}$ of RR-aux shrinks drastically, and now BP has larger values than RR-aux, ranging from approximately 300-550\%. RR has the lowest ESS across all settings. The variance of the log of the absolute value of the likelihood estimates is a key factor determining the ESS, and a decline in ESS is attributed to an increase in this variance (assuming the same proposal and a similar ratio of negative signs). We note that the tuning guidelines help the BP method, especially when $n = 1{,}000$ and the variance of the log of the absolute value of the estimator variance of RR and RR-aux is large (evidenced by low ESS); the method maintains a relatively high ESS. The exchange method exhibits a lower effective sample size than the signed pseudo-marginal methods (except RR) in most cases.

Finally, the $\mathrm{ESS}_{\kappa}/s$ incorporates the time taken to produce the corresponding $\mathrm{ESS}_{\kappa}$ results. The BP method is clearly better in this metric: its $\mathrm{ESS}_\kappa/s$ is substantially higher than that of both Russian‐roulette variants because BP is suitable for parallelisation and vectorisation, which markedly reduces runtime. In contrast, the structure of the RR and RR-aux estimators cannot readily employ these features, see Section \ref{app:computational_aspects} in the supplement for details. In addition, RR performs worse than RR-aux with respect to computing time because it does not scale to multiple observations (recall Table \ref{tab:methods summary}). Note that as the sample size increases, the $\mathrm{ESS}_\kappa/s$ gap between BP and the RR methods becomes even more pronounced, with factors of roughly $17\times$--$35\times$ improvement compared to RR-aux, and more than $70\times$--$140\times$ compared to RR when $n=1{,}000$.

To conclude, among the Bayesian methods, the top-performing method is BP. RR-aux is an alternative to BP, albeit at a significantly higher computational cost, hence leading to a much lower ESS/s, especially for larger $n$. Both RR and Exchange fail to deliver competitive results. RR is unsuitable for an $\mathrm{FB_5}$ distribution with multiple observations, as it requires substantially more computing resources to compute $c^{-1}(\kappa,\beta)$ multiple times; see Section \ref{eq:Kent_posterior} for details. Finally, for the frequentist methods, moment estimation fails when $\beta/\kappa$ approaches the limiting value 0.5. While the MLE method gives better results, it remains inferior to the Bayesian methods in terms of RMSE when $n$ is small.

\section{An empirical study on spherical data} \label{sec: empirical study}
We now analyse four real spherical datasets using the Kent distribution and our method. Each data set contains samples from two groups that are formed naturally from the sample collection process. Figure \ref{fig:FB5_data} plots the spherical datasets. 
\begin{enumerate}

    \item \textbf{Palaeomagnetic} (Palaeo)  \citep{wood1982bimodal}: Thirty-three estimates of previous magnetic pole positions were obtained using palaeomagnetic techniques.  Each estimate is associated with a different site in  Tasmania. The data is originally from \cite{schmidt1976non} and the author points out that the data is likely to fall mainly into two groups of distinct geographical regions. Following \cite{figueiredo2009discriminant}, the first group contains 9 observations with the indices 9, 10, 11, 12, 14, 16, 23, 24, 30. The second group has 24 observations. 
        
    \item \textbf{Magnetic} \cite[Table B8]{fisher1993statistical}:  Measurements of magnetic remanence from a set of 62 specimens is obtained. The specimens are from Mesozoic Dolerite from Prospect, New South Wales, after successive partial demagnetisation stages ($200^{\circ}$ and $350^{\circ}$). 
    An experiment was conducted to determine the blocking temperature spectrum of the magnetisation components. 
    
    \item \textbf{Sandstone}  \cite[Table B23]{fisher1993statistical}: Measurements of natural remanent magnetisation in Old Red Sandstone rocks in Pembrokeshire, Wales. The measurements consist of specimens from two sites with the number of observations 35 and 13, respectively. 
    
    \item  \textbf{Stone} \cite[Table B25]{fisher1993statistical}:  Measurements of the longest axis and shortest axis (101 observations) orientations of tabular stones on a slope at Windy Hills, Scotland.
    
\end{enumerate}

\begin{figure}[ht]
    \centering\includegraphics[width=1.0\linewidth]{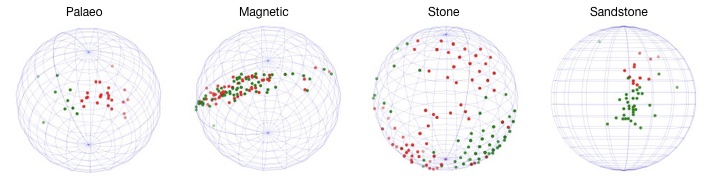}
    \caption{Illustration of the datasets. Green points and red points refer to the observations from groups 1 and 2, respectively.}
    \label{fig:FB5_data}
\end{figure}

The two groups are modelled separately by assuming a non-hierarchical structure on the prior for all the parameters. The data is modelled in the same way as in Section \ref{subsec: kent_dist} using the density function in \eqref{eq: Kent}. 

\begin{table}{}
\centering
\small{\begin{tabular}{lccc|ccc}
\toprule
Paleo & \multicolumn{3}{c}{Group 1 (n=9)} &\multicolumn{3}{c}{Group 2 (n=24)} \\
    & $\beta/\kappa$ & $\mathrm{ESS}_{\beta/\kappa}$ & $\mathrm{ESS}_{\beta/\kappa}$/s & $\beta/\kappa$ & $\mathrm{ESS}_{\beta/\kappa}$ & $\mathrm{ESS}_{\beta/\kappa}$/s  \\
\cmidrule(lr){2-4} \cmidrule(lr){5-7}
BP &   0.15 &   235.49 &       9.33 &   0.14 &   201.70 &       21.12 \\
       RR &   0.15 &   49.47 &       1.37 &   0.12 &   201.60 &       2.36 \\
   RR-aux &   0.19 &   158.03 &       11.11 &   0.10 &   165.48 &       9.14 \\
 Exchange &   0.27 &   118.14 &        8.80 &   0.19 &   75.02 &       5.76 \\
   Moment &   0.15 &       - &           - &   0.15 &       - &           - \\
      MLE &    0.17 &       - &           - &   0.16 &       - &           - \\
  \midrule
Magnetic & \multicolumn{3}{c}{Group 1 (n=62)} &\multicolumn{3}{c}{Group 2 (n=62)} \\
    & $\beta/\kappa$ & $\mathrm{ESS}_{\beta/\kappa}$ & $\mathrm{ESS}_{\beta/\kappa}$/s & $\beta/\kappa$ & $\mathrm{ESS}_{\beta/\kappa}$ & $\mathrm{ESS}_{\beta/\kappa}$/s \\
   \cmidrule(lr){2-4} \cmidrule(lr){5-7}
       BP &   0.48 &   184.78 &       18.50 &   0.49 &   140.27 &       14.43 \\
       RR &     0.48 &   44.38 &      0.26 &   0.50 &    1.90 &    $<0.01$ \\
   RR-aux &   0.48 & 146.95 &       5.37 &   0.49 &   153.40 &   5.26 \\
 Exchange & 0.49 &155.64 &       10.11 &0.49 &117.07 &7.63 \\
   Moment & 0.35 &  - &  - & 0.38 &       - & - \\
      MLE & 0.50 & - &- &   0.50 & - & -\\
      \midrule
Sandstone & \multicolumn{3}{c}{Group 1 (n=35)} &\multicolumn{3}{c}{Group 2 (n=13)} \\
    & $\beta/\kappa$ & $\mathrm{ESS}_{\beta/\kappa}$ & $\mathrm{ESS}_{\beta/\kappa}$/s & $\beta/\kappa$ & $\mathrm{ESS}_{\beta/\kappa}$ & $\mathrm{ESS}_{\beta/\kappa}$/s \\
   \cmidrule(lr){2-4} \cmidrule(lr){5-7}
             BP &  0.08 &   147.92 &       5.87 &   0.20 &   97.11 &        10.28 \\
       RR &  0.06 &   98.09 &      0.88 &   0.17 &    149.16 &        2.16 \\
   RR-aux &  0.08 &   218.60 &        10.43 &   0.17 &   111.36 &       7.38 \\
 Exchange &   0.17 &   95.56 &       6.65 &   0.28 &   7.24 &      0.56 \\
   Moment &  0.09 &       - &           - &   0.27 &       - &           - \\
      MLE &   0.11 &       - &           - &    0.29 &       - &           - \\
      \midrule
Stone & \multicolumn{3}{c}{Group 1 (n=101)} &\multicolumn{3}{c}{Group 2 (n=101)} \\
    & $\beta/\kappa$ & $\mathrm{ESS}_{\beta/\kappa}$ & $\mathrm{ESS}_{\beta/\kappa}$/s & $\beta/\kappa$ & $\mathrm{ESS}_{\beta/\kappa}$ & $\mathrm{ESS}_{\beta/\kappa}$/s \\
   \cmidrule(lr){2-4} \cmidrule(lr){5-7}
             BP &   0.13 &   82.67 &       3.31 &   0.49 &   197.48 &        20.54 \\
       RR &   0.11 &   10.81 &     0.03 &   0.49 &   11.79 &     0.04 \\
   RR-aux &   0.14 &   75.69 &        2.05 &   0.49 &   138.70 &       3.78 \\
 Exchange &   0.43 &   71.26 &       5.01 &   0.49 &   239.20 &       18.49 \\
   Moment &  0.06 &       - &           - &   0.21 &       - &           - \\
      MLE &   0.14 &       - &           - &        0.5 &       - &           - \\
\bottomrule
\end{tabular}}
\caption{Results for the Kent model for the four datasets when estimating $\beta/\kappa$. All the chains ran for 10,000 iterations. The term $\beta/\kappa$ refers to the posterior mean for the Bayesian methods and estimates for the frequentist methods. The sign correction is applied for BP, RR and RR-aux. }\label{tab:FB5_empirical}
\end{table}

\begin{figure}[ht]
    \centering
    \includegraphics[width=0.8\linewidth]{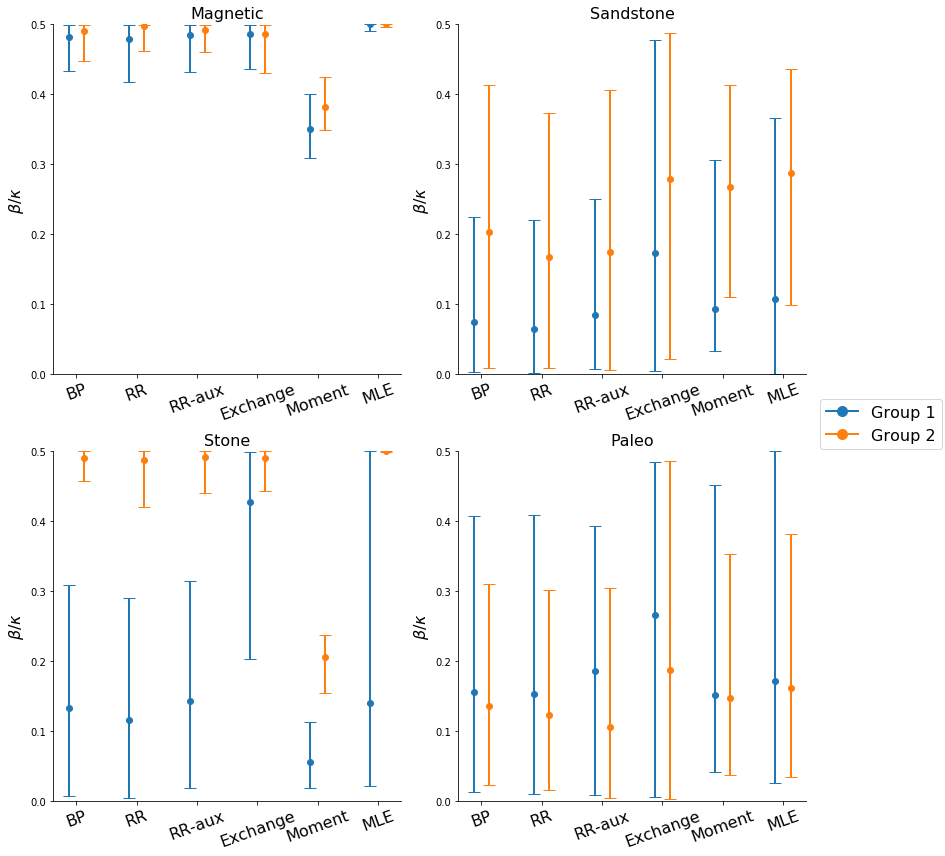}
    \caption{Posterior mean/estimate and $95\%$ credible/confidence intervals for each method by group.}
    \label{fig:FB5_beta}
\end{figure}

Table \ref{tab:FB5_empirical} shows the posterior mean/estimate, ESS and ESS per second for $\beta/\kappa$.  The results for $\beta$ and $\kappa$ are in Section \ref{app: results of FB5} in the supplement.  The table shows that BP, RR and RR-aux have similar estimates and credible intervals for all four datasets. Exchange has a significantly larger posterior mean than those of the other Bayesian methods. It also has wider credible intervals for all datasets. The abnormal performance of Exchange is consistent with the results in Section \ref{subsec: kent_dist}, where we attribute this issue to the sampling inefficiency from a $\mathrm{FB_5}$ distribution. In addition, while BP and RR-aux occasionally trade places when ranked by ESS, the ESS/s metric yields a more decisive ordering: BP outperforms RR-aux in all cases but two (Group 1 for Sandstone and Paleo). The factors of improvement span a wide range of roughly 1.6$\times$ to 7$\times$. In contrast, RR yields the lowest ESS as well as ESS/s in every scenario, demonstrating its poor performance on the Kent distribution and its lack of scalability to multiple observations. The moment estimation demonstrates significant deviations from both the Bayesian results and MLE estimates in the limiting case where $\beta/\kappa \approx 0.5$ occurs (Magnetic and Stone, Group 2).

Figure \ref{fig:FB5_beta} shows the posterior mean/estimate of $\beta/\kappa$ with 95\% credible intervals (confidence interval for Moment and MLE). The signed-pseudo marginal methods deliver similar results, while the results of the exchange method differ substantially (recall the poor mixing due to inefficient sampling from the likelihood).

\section{Conclusions and future research} \label{sec: discussion}

We propose the signed block PMMH with the block-Poisson estimator to carry out exact inference in general doubly intractable problems. Our method requires only an unbiased estimator of the normalising function, which makes it applicable to a wider range of problems than its competitors, such as the exchange algorithm, which requires perfect sampling from the model. Moreover, we derive a finite-sample result that ensures, with high probability, that the denominator of the importance sampling estimator does not approach zero, eliminating a critical breakdown scenario in signed PMMH algorithms overlooked in the literature.

Compared with the Russian roulette method in \cite{lyne2015russian}, the block-Poisson estimator achieves a smaller variance of the 
logarithmic difference in the likelihood estimates in the MH acceptance ratio by its use of correlated pseudo-marginal updates. Moreover, the Russian roulette method lacks guidelines on how to tune its hyperparameters. We derive heuristic guidelines based on analytically derived statistical properties of our estimator.  The Ising model example in Section \ref{sec: ising model} suggests that our approach is 2 to 14 times as efficient (time normalised) as the Russian roulette methods. For the Kent distribution in Section \ref{subsec: kent_dist} with the largest number of observations case ($n=1{,}000$), the improvement factors range between 17 to 35 times against RR-aux and more than 70-140 against RR.

 Despite its broad applicability, the signed PMMH algorithm---whether using block-Poisson or Russian roulette estimators---can be computationally costly, particularly when unbiased estimation of the normalising function is time-consuming. The Ising model exemplifies this, as each iteration requires multiple annealed importance sampling estimates, which dominate the computational cost. In such cases, the difference between BP and Russian roulette methods is less pronounced, since parallelisation and vectorisation contribute little to the overall runtime. In contrast, when unbiased normalising function estimates are computationally cheap, as in the Kent distribution model, the advantages of BP become more pronunced, as arithmetic operations account for a substantial proportion of the total execution time.
 
 Future research will explore other doubly intractable problems and develop tuning strategies for cases when the estimator of the normalising function is not normal. Another area of future research is to obtain tighter bounds for the concentration inequality in Theorem \ref{thm: sum far from zero}. To this end, the improved Bernstein-type bounds via iterated Poincaré inequalities proposed in \cite{huang2024bernstein} may prove useful.

\section*{Acknowledgements}
Yu Yang was financially supported by a University International Postgraduate Award from UNSW Sydney. Robert Kohn was partially supported by the Australian Research Council (IC190100031, DP210103873).
Scott Sisson is supported by the Australian Research Council (FT170100079). We thank Chris Sherlock for helpful comments on an earlier version of this manuscript. We thank the Associate Editor and two referees for helpful comments that significantly improved the manuscript.

\bibliographystyle{apalike}
\addcontentsline{toc}{section}{\refname}
\bibliography{references}

\clearpage
\appendix

\setcounter{section}{0}
\setcounter{figure}{0}
\setcounter{table}{0}
\setcounter{equation}{0}
\setcounter{lemma}{0}
\setcounter{remark}{0}


\renewcommand{\thesection}{S\arabic{section}}

\renewcommand{\thesubsection}{\thesection.\arabic{subsection}}

\renewcommand{\theequation}{S\arabic{section}.\arabic{equation}}

\renewcommand{\thefigure}{S\arabic{figure}}
\renewcommand{\thetable}{S\arabic{table}}
\renewcommand{\thelemma}{S\arabic{lemma}}
\renewcommand{\theremark}{S\arabic{remark}}

\section*{Supplementary material}
\section{Properties of the block-Poisson estimator} \label{app: BPproof}

\begin{proof}[Proof of Lemma \ref{lemma: block poisson est}]

The block-Poisson estimator is expressed as   $$\widehat{L}_B(\Btheta) = \prod_{l=1}^{\lambda} \xi_l(\Btheta) $$ with 
\begin{equation*}
\xi_l(\Btheta) =\exp(a/ \lambda + m) \prod_{h=1}^{\chi_l} \dfrac{\widehat{B}^{(h,l)}(\Btheta)-a}{m\lambda},
\end{equation*}
where $\lambda$ is the number of blocks with $\chi_l \sim \mathrm{Pois}(m)$ and $a$ is an arbitrary constant. For notational convenience, dependence on $\Btheta$ is omitted for $\widehat{L}_B$, $\widehat{B}$ and $\xi$. 

The following proofs closely follow the proofs in \citet[Section S8]{quiroz2021block} who assume $m = 1$, whereas here $m$ can be any non-negative integer. The two properties below are useful for the proof. Suppose that $X \sim \mathrm{Pois}(m) $ and $A < \infty$. Then,
\begin{enumerate}
    \item[(i)] $E_X(A^X) = \exp((A-1)m)$.
    \item[(ii)] $\mathrm{Var}_X (A^X) = \exp(-m) [\exp(A^2m) - \exp(2Am -m)]$.
\end{enumerate}

\textbf{Proof of unbiasedness}
\begin{align*}
    E(\xi_l) &= \exp(a/\lambda + m) E\Bigg[ \prod_{h=1}^{\chi_l} \dfrac{\widehat{B}^{(h,l)}-a}{m\lambda} \Bigg] \\
    &=  \exp(a/\lambda + m) E_\chi E_{\widehat{B}| \chi} \Bigg[ \prod_{h=1}^{\chi_l} \dfrac{\widehat{B}^{(h,l)}-a}{m\lambda} \Bigg] \\
    &= \exp(a/\lambda + m) E_\chi \Bigg[ \dfrac{B-a}{m\lambda} \Bigg]^\chi\\
    &= \exp(a/\lambda + m) \exp( (B-a)/\lambda -m)\\
    &= \exp(B/ \lambda).
\end{align*}
 Hence, $E(\widehat{L}_B) = \exp(B)$, 
as $\xi_1,\dots,\xi_\lambda$ are independent.

In the implementation, we use $a = \widehat{B} - m \lambda $, where $\widehat{B}$ is an estimate of $B$ which is independent of $\widehat{B}^{(h,l)}$. Such a choice of $a$ preserves the unbiasedness of the block-Poisson estimator.  

Treating $a$ as a random variable, the expectation of the  block-Poisson estimator can be expressed as
\[
    E(\widehat{L}_B(\Btheta)) = E_a E_{\xi_1,\dots,\xi_\lambda|a} \prod_{l=1}^{\lambda} \xi_l(\Btheta) = E_a \left( \prod_{l=1}^{\lambda} E_{\xi_l|a} \xi_l(\Btheta) \right).
    \]
    The conditional expectation of $E_{\xi_l|a} \xi_l$ is (omitting the dependence on $\Btheta$)
    \[E_{\xi_l|a} \xi_l = E_{\xi_1|a} \exp(a/\lambda+m) \prod_{h=1}^{\chi_l} \dfrac{\widehat{B}^{(h,l)}-a}{m\lambda}. \] 
The conditional expectation is the same as the derived $E(\xi_l)$, which is independent of $a$ as $a$ cancels out in the process. Hence, treating $a$ as a random variable still guarantees the unbiasedness of the block-Poisson estimator.

\textbf{Derivation of the variance}

From the definition of $\widehat{L}_B$, 
\begin{align*}
    \mathrm{Var}(\widehat{L}_B) &= \mathrm{Var}\left(\prod_{l=1}^{\lambda} \xi_l\right).
\end{align*}

For a collection of independent random variables $\xi_1,\hdots,\xi_\lambda$,
\begin{align*}
     \mathrm{Var}\left(\prod_{l=1}^{\lambda} \xi_l\right) = \prod_{l=1}^{\lambda}\left( \mathrm{Var}(\xi_l) + E(\xi_l)^2\right) - \prod_{l=1}^{\lambda} E(\xi_l)^2
\end{align*}
with 
\begin{align*}
    \mathrm{Var}(\xi_l) &= \exp(a/\lambda +m) \left[ E_\chi \mathrm{Var}_{\widehat{B}| \chi} \left( \prod_{h=1}^{\chi}\dfrac{\widehat{B}^{(h,l)} - a }{m\lambda }\right) + \mathrm{Var}_\chi E_{\widehat{B}| \chi} \left(
     \prod_{h=1}^{\chi}\dfrac{\widehat{B}^{(h,l)} - a }{m\lambda }\right) \right].
\end{align*}

For the first term in the brackets, making the use of independence of $\widehat{B}^{(h,l)},h=1,\dots,\chi_l$, $$ \mathrm{Var}_{\widehat{B}| \chi} \left( \prod_{h=1}^{\chi}\dfrac{\widehat{B}^{(h,l)} - a }{m\lambda }\right)$$ is simplified as 
\begin{align*}
    \mathrm{Var}_{\widehat{B}| \chi} \left( \prod_{h=1}^{\chi}\dfrac{\widehat{B}^{(h,l)} - a }{m\lambda }\right) &= \prod_{h=1}^{\chi} \left[  \mathrm{Var} \left(\dfrac{\widehat{B}^{(h,l)} -a }{m \lambda} \right)   +  E\left( \dfrac{\widehat{B}^{(h,l)} -a }{m \lambda}\right)^2\right] \\
    & - \prod_{h=1}^{\chi} E\left( \dfrac{\widehat{B}^{(h,l)} -a }{m \lambda}\right)^2\\
    &= \prod_{h=1}^{\chi} \left( \frac{\sigma_{\widehat{B}}^2 + (B-a)^2}{(m\lambda)^2} \right) - \left( \frac{B-a}{m\lambda} \right)^{2\chi}.
\end{align*}

Taking the expectation with respect to $\chi$ and using property (i) twice for the two terms, we have
\begin{align*}
     E_\chi \mathrm{Var}_{\widehat{B}| \chi} \left( \prod_{h=1}^{\chi}\dfrac{\widehat{B}^{(h,l)} - a }{m\lambda }\right) &= \exp\left[ \left(\frac{\sigma_{\widehat{B}}^2 + (B-a)^2}{(m\lambda)^2 }   -1 \right)m \right] - \exp\left[ \left( \frac{(B-a)^2}{(m\lambda)^2 } -1 \right)m \right]\\
     &=  \exp\left[ \left( \frac{(B-a)^2}{(m\lambda)^2} -1 \right)m \right] \left[ \exp 
     \left(\frac{\sigma_{\widehat{B}}^2}{m\lambda^2} \right)-1 \right].
\end{align*}

The second term is derived similarly,
\begin{align*}
   \mathrm{Var}_\chi E_{\widehat{B}| \chi} \left( \prod_{h=1}^{\chi}\dfrac{\widehat{B}^{(h,l)} - a }{m\lambda }\right)  &= \mathrm{Var}_\chi \left( \frac{B-a}{m\lambda}\right)^\chi \\
   &= \exp\left(-m + \frac{(B-a)^2}{m\lambda^2} \right) - \exp\left( 2(B-a)/\lambda -2m \right).
\end{align*}

Combining the two terms, we have 
\begin{align*}
    \mathrm{Var}(\xi_l) = \exp\bigg[ \frac{(B-a)^2 + \sigma_B^2}{m\lambda^2} -m \bigg] - \exp(2(B-a)/\lambda -2m ).
\end{align*}

Deriving $E(\xi)^2$ is straight forward,
\begin{align*}
    E(\xi_l)^2  &=  \left[ E_\chi E_{B| \chi} \prod_{h=1}^{\chi} \left(\frac{\widehat{B}^{(h,l)}-a}{m\lambda}\right) \right]^2\\
    &= \exp( 2(B-a)/\lambda - 2m).
\end{align*}
Combining all the terms, after some algebra, the variance of the block-Poisson estimator is 
\begin{align*}
    \mathrm{Var}(\widehat{L}_B) = \exp\left[  \frac{(B-a)^2 + \sigma_B^2}{m\lambda } + 2a + m\lambda \right] - \exp(2B).
\end{align*}

\textbf{Choice of the constant $a$}

The optimal value minimising $\mathrm{Var}(\widehat{L}_B)$ is $a = B - m \lambda$,  which is obtained by solving the equation  $\partial \mathrm{Var}(\widehat{L}_B)/ \partial a =0$. 
\end{proof}

\begin{proof}[Proof of Lemma \ref{lemma: positive prob}]

The proof of Lemma \ref{lemma: positive prob} 
is the same as Lemma 3 in \cite{quiroz2021block}. We sketch the key points here. 

Note that 
\[ \Pr(\widehat{L}_B \geq 0) = \Pr \left( \prod_{l=1}^\lambda \xi_l \geq 0 \right) \]

and $\widehat{L}_B$ only occurs in the case of an odd number of negative terms $\xi_l, l= 1,\dots,\lambda$. Based on results in \cite[p.~277]{feller-vol-2}, we obtain 
\[ \Pr(\widehat{L}_B > 0) = \frac{1}{2} \left( 1 + (1- 2\Pr(\xi_l < 0))^\lambda. \right)\]

Notice that 
\[ \Pr(\xi_l <0) = \sum_{j=1}^{\infty} \Pr \left((\prod_{l=1}^{j} (A_{m} < 0))^j \right) \Pr(\chi_l = j), \:\chi_l \sim \mathrm{Pois}(m). \]

Apply the result in \cite{feller-vol-2} again, we have

\[\Pr \left((\prod_{l=1}^{j} (A_{m}\leq 0))^j \right) = \frac{1}{2} \left( 1 -(1-2\Pr(\xi_l < 0))^j \right), \]
where $A_m = (\widehat{B(\Btheta)} - B(\Btheta))/(m\lambda)+1$, concluding the proof.

\end{proof}

\begin{proof}[Proof of Lemma \ref{lemma: log variance}]

The variance of the log of the likelihood estimator is
\begin{align*}
    \mathrm{Var}( \log \abs*{\widehat{L}_B}) &= \mathrm{Var}\bigg(\sum_{l=1}^\lambda \sum_{h=1}^{\chi_l} \log \abs*{\frac{\widehat{B}^{(h,l)}-a}{m\lambda}} \bigg) \\
    &= E_{\chi_{1,\hdots,\lambda}} V_{B| \chi_{1,\hdots \lambda}} \log \abs*{\frac{\widehat{B}^{(h,l)}-a}{m\lambda}} + \mathrm{Var}_{\chi_{1,\hdots,\lambda}} E_{B| \chi_{1,\hdots \lambda}} \log \abs*{\frac{\widehat{B}^{(h,l)}-a}{m\lambda}}.
\end{align*}
Suppose $\widehat{B}^{(h,l)} \sim N(B, \sigma^2_{\widehat{B}})$ and $a = B -m\lambda$; then 
\begin{align*}
    \log \abs*{\frac{\widehat{B}^{(h,l)}-a}{m\lambda}} &= \log \abs*{\frac{\sigma_{\widehat{B}}Z}{m\lambda} + 1}\\
    & =  \log(\sigma_{\widehat{B}}/(m\lambda)) + \log\abs*{ Z + m \lambda/\sigma_{\widehat{B}}} \\
    &=  \log(\sigma_B/(m\lambda)) + \frac{1}{2} \log(( Z + m \lambda/\sigma_{\widehat{B}})^2)\\
    &= \log(\sigma_{\widehat{B}}/(m\lambda)) + \frac{1}{2} \log W^{(h,l)}, \quad  W^{(h,l)} \sim \chi^2 (1, (m\lambda/\sigma_{\widehat{B}})^2),
\end{align*}
where $\chi^2(k,\lambda)$ denotes the non-central $\chi^2$ distribution with $k$ degrees of freedom and non-centrality parameter $\lambda$. Lemma S12 in \cite{quiroz2021block} provides the moments of $\log W$.

Let $\eta_B$ and $\nu_B^2$ be the expectation and the variance of $\log \abs{(\widehat{B}^{(h,l)}-a)/{m\lambda}}$ respectively. We have
\begin{align*}
     \eta_B = E \bigg( \log \abs*{\frac{\widehat{B}^{(h,l)}-a}{m\lambda}} \bigg)  &=   \log(\sigma_B/(m\lambda)) + \frac{1}{2} \log(2 + E_J[\psi^{(0)} (0.5 + J)], \\
     \nu_B^2 = \mathrm{Var} \bigg( \log \abs*{\frac{\widehat{B}^{(h,l)}-a}{m\lambda}} \bigg) &= \frac{1}{4} \left[  E_J[\psi^{(1)} (0.5 +J)] + \mathrm{Var}_J[\psi^{(0)}(0.5+J)] \right],
\end{align*}
where $J \sim \mathrm{Pois}( (m\lambda)^2/2\sigma_{\widehat{B}}^2)$ and $\psi^{(q)}$ is the polygamma function of order $q$. 

Then,
\begin{align*}
     \mathrm{Var}( \log \abs*{\widehat{L}_B}) &= E_{\chi_{1,\hdots,\lambda}} \bigg(\sum_{l=1}^{\lambda} \chi_l \bigg)\nu^2_B + \mathrm{Var}_{\chi_{1,\hdots,\lambda}} \bigg(\sum_{l=1}^{\lambda} \chi_l \bigg) \eta_B\\
     &= m\lambda (\nu_B^2 + \eta_B^2).
\end{align*}
Furthermore, $\mathrm{Var}( \log \abs*{\widehat{L}_B}) < \infty$.   Lemma 7 in \cite{quiroz2021block} derives the result. 
\end{proof}

\section{Finite sample properties of the signed importance sampling estimator}\label{app: FiniteSampleBoundsIS}

Before proving Theorem \ref{thm: sum far from zero}, we state some results.

\begin{lemma}\label{lem: properties s chain} Suppose that Assumption \ref{ass: assumptions} holds. Then
\begin{enumerate}
    \item[(i)] $\mathrm{Var}_{\overline{\pi}}\left(S \right) = 1 - \mu^2,$ where $\mu = E_{\overline{\pi}}(S)$.
    \item[(ii)] $\underset{S \in \{-1, +1\}}{\max} \left|S - \mu\right| = 1 + |\mu| $.
    \item[(iii)] Suppose that $\left|\sum_{i=1}^N s^{(i)}\right| \leq cN$, for $0 < c < \left| \mu\right|$, with $\mu = E_{\overline{\pi}}(S)$. Then,
    \begin{align} \label{eq:concentration_inequality}
        \left|\sum_{i=1}^N s^{(i)} - \mu N \right| \geq (|\mu| - c)N.
    \end{align}
\end{enumerate}
\end{lemma}
\begin{proof}
For (i), we use $\mathrm{Var}_{\overline{\pi}}\left(S \right)= E_{\overline{\pi}}\left(S^2\right) - E_{\overline{\pi}}\left(S\right)^2$, with
$$E_{\overline{\pi}}\left(S^2\right) = 1^2\eta(\{+1\}) +  (-1)^2\eta(\{-1\}) = 1.$$

For (ii), verifying the two cases $S=-1$ and $S=+1$, together with $\mu \in [-1,1]$ gives the result.

For (iii), consider the case $\mu > 0$, such that $\mu=|\mu|$. Since $\sum_{i=1}^N s^{(i)} \in [-cN, cN]$, $$\mu N - \sum_{i=1}^N s^{(i)} \geq \mu N - cN,$$
it follows that $\left| \sum_{i=1}^N s^{(i)} - \mu N \right| \geq (|\mu| - c)N$. When $\mu <0$, 
$$\sum_{i=1}^N s^{(i)} - \mu N \geq -cN - \mu N,$$
since $\sum_{i=1}^N s^{(i)}\geq -cN$. Since $|\mu| = -\mu$ for $\mu <0$, it follows that
$$\left|\sum_{i=1}^N s^{(i)}  - \mu N \right| \geq (|\mu| - c)N.$$
\end{proof}

The following lemma is a concentration result for \eqref{eq:concentration_inequality}.
\begin{lemma}\label{lem:concentration_inquality} Suppose that Assumption \ref{ass: assumptions} holds. Then, for $0< c <|\mu|$,
\begin{align*}  
 \textstyle \Pr_{\overline{\pi}}\left(\left|\sum_{i=1}^N s^{(i)} - \mu N \right| \geq (|\mu| - c)N\right) & \leq 2\exp\left(-\frac{(|\mu| - c)^2 \delta N}{4(1-\mu^2) + 10(|\mu| - c)(1 + |\mu|)}\right), 
\end{align*}
where $\delta$ is the spectral gap of the Markov chain.
\end{lemma}
\begin{remark}
    Note that, with $v=|\mu| -c >0$,
    \begin{align*}  
 \textstyle \Pr_{\overline{\pi}}\left(\left|\frac{1}{N}\sum_{i=1}^N s^{(i)} - \mu \right| \geq v\right) & \leq 2\exp\left(-\frac{v^2 \delta N}{4(1-\mu^2) + 10v(1 + |\mu|)}\right), 
\end{align*}
which restates the concentration inequality in terms of the deviation of the sample mean from the population mean.
\end{remark}
\begin{proof}
Assumption~\ref{ass: assumptions} allows us to use Theorem 3.3 in \cite{paulin2015concentration} with $f=S(\Btheta,\Bu,\nu)$. Theorem 3.3 assumes a stationary reversible Markov chain on a general state space $\Omega$ with a stationary distribution, which we denote as $\overline{\pi}$ in Assumption~\ref{ass: assumptions}. With our notation, Theorem 3.3 assumes that $f=S(\Btheta,\Bu,\nu) \in L^2(\overline{\pi})$ with $|S(\Btheta, \Bu,\nu) - E_{\overline{\pi}}(S)|\leq C$ for every $\Btheta, \Bu,\nu \in \Omega$. It follows from Lemma \ref{lem: properties s chain}(ii) that $C=1+|\mu|$. The sign function is clearly in $L^2(\overline{\pi})$, since
$$\int_{\Btheta} \int_{\nu} \int_{\Bu} |S(\Btheta, \Bu,\nu)|^2\overline{\pi}(\Btheta,\Bu, \nu )d\Btheta d\nu d\Bu=\int_{\Btheta} \int_{\nu} \int_{\Bu} \overline{\pi}(\Btheta, \Bu, \nu )d\Btheta d\nu d\Bu=1<\infty.$$
We can now apply Equation (3.21) in \cite{paulin2015concentration} with $$t=(|\mu| -c)N,$$ 
and $V_f=\mathrm{Var}_{\overline{\pi}}(S)=1-\mu^2$ from Lemma \ref{lem: properties s chain}(i).
\end{proof}

We are now ready to prove Theorem \ref{thm: sum far from zero}.

\begin{proof}[Proof of Theorem \ref{thm: sum far from zero}]
To show that the probability in Theorem \ref{thm: sum far from zero} is close to 1 is equivalent to showing that the complementary event, i.e.\ $\left|\sum_{i=1}^N s^{(i)}\right| \leq cN$, has probability close to zero. By Part Lemma \ref{lem: properties s chain}(iii),  $$\left\{\left|\sum_{i=1}^N s^{(i)}\right| \leq cN \right\} \subseteq \left\{ \left|\sum_{i=1}^N s^{(i)} - \mu N \right| \geq (|\mu| - c)N\right\},$$
and thus 
\begin{align}\label{eq:Upper_bound_probability}
  \textstyle \Pr_{\overline{\pi}}\left(\left|\sum_{i=1}^N s^{(i)}\right| \leq cN \right) & \leq \textstyle \Pr_{\overline{\pi}}\left(\left|\sum_{i=1}^N s^{(i)} - \mu N \right| \geq (|\mu| - c)N\right)  \nonumber \\ 
  & \leq 2\exp\left(-\frac{(|\mu| - c)^2 \delta N}{4(1-\mu^2) + 10(|\mu| - c)(1 + |\mu|)}\right), 
\end{align}
where the last inequality follows from Lemma \ref{lem:concentration_inquality}. Note the exponential decay in terms of $N$ in  \eqref{eq:Upper_bound_probability} since $c < |\mu| \leq 1$. 
Now, to find $N_0$, 
\begin{align*}
    2\exp\left(-\frac{(|\mu| - c)^2 \delta N}{4(1-\mu^2) + 10(|\mu| - c)(1 + |\mu|)}\right) & \leq \varepsilon \\
    -\frac{(|\mu| - c)^2 \delta N}{4(1-\mu^2) + 10(|\mu| - c)(1 + |\mu|)}  & \leq \log \left (\frac{\varepsilon}{2}\right)  \\
    \frac{(|\mu| - c)^2 \delta N}{4(1-\mu^2) + 10(|\mu| - c)(1 + |\mu|)}  & \geq \log \left (\frac{2}{\varepsilon}\right), 
\end{align*}
which gives 
\begin{align}\label{eq:N_0}
    N_0 = \frac{4(1-\mu^2) + 10(|\mu| - c)(1 + |\mu|)}{(|\mu| - c)^2 \delta}  \log \left (\frac{2}{\varepsilon}\right).
\end{align}
\end{proof}

\section{Studying the effects of the simplified soft lower bound}\label{subsec:simplified_lb}

All the results in Lemmas \ref{lemma: block poisson est}-\ref{lemma: log variance},  except the unbiasedness property (Part (i) Lemma \ref{lemma: block poisson est}), are derived under the assumption $a=B(\Btheta)-m\lambda$. In addition, Lemma \ref{lemma: log variance} assumes that $\widehat{B}(\Btheta)$ is normal, while Lemma \ref{lemma: positive prob} in practice is used together with the normality assumption to compute the probability of a positive estimator when tuning the algorithm.

This section explores the quantities in the three lemmas using the following different choices of $a$:
\begin{itemize}
    \item[(a.)] $a_\mathrm{opt} = B(\Btheta) - m\lambda$,
    \item[(b.)] $\widehat{a}_\mathrm{opt} = \widehat{B}(\Btheta) - m\lambda$,
    \item[(c.)] $a_\mathrm{sub} = -1 - m\lambda$.
\end{itemize}
The aim is to show that, in particular (a.) and (c.), result in quantities used in the tuning that do not differ too much. Note that for case (a.) the analytical expressions in the lemmas are used to compute the quantities of interest ($B(\Btheta)$ is known in the simulation setup, see below), however, for (b.) and (c.), simulation is required (even if $B(\Btheta)$ is known). The simulation computes $1{,}000$ realisations of \eqref{eq: block_pois_est} and estimates the quantities by their sample versions. For example, to compute Lemma 3 for the soft lower bound $a_\mathrm{sub}$, $\xi_\ell \sim \mathrm{Pois}$ and $\widehat{B}^{(h,l)}\sim N$ are simulated from a Poisson and normal distribution (with the settings below), respectively, for $l=1,\dots,\lambda$. Then the log of the absolute value of \eqref{eq: block_pois_est} is computed with $a$ replaced by $a_\mathrm{sub}=-1-m\lambda$. Repeating this gives $1{,}000$ realisations of $\log|\widehat{L}_B|$ (with the sub-optimal soft lower bound $a_\mathrm{sub}$) and the sample variance estimates the corresponding $\sigma^2_{\log|\widehat{L}_B|}$. 

The simulation setup is as follows. We first note that $B(\Btheta)=-\nu Z(\Btheta)<0$ since $Z(\Btheta
)>0$ in doubly intractable problems and $\nu > 0$. The simulation uses an equally spaced grid $\exp(B) \in [0.01, 0.99]$. Note that, by construction, $B$ and $\exp(B)$ do not depend on $\Btheta$. Next, $\widehat{B} \sim N(B, \sigma_{\widehat{B}}^2)$ and we set $\sigma_{\widehat{B}}= \sqrt{2}$ by assuming that $\frac{\mathrm{Var}(\widehat{Z}_i(\Btheta))}{M Z^2(\Btheta)} = 1$ based on \eqref{eq: gamma_expression}, which implies the variance of the Monte Carlo estimate $M^{-1} \sum_{i=1}^M \widehat{Z}_i(\Btheta)$ equals to the squared value of the true value. The simulation uses the values of the hyperparameters $\lambda$ ($\lambda$ = 10) and $m$ ($m=1$) unless otherwise stated.

Plot (a) in Figure \ref{fig:BP-properties} compares the variance of the block-Poisson estimator ($\mathrm{Var}(\widehat{L}_B)$) against $\exp(B)$. As $\exp(B)$ gets closer to 1, the discrepancy of $\widehat{a}_{\mathrm{opt}}$, $a_{\mathrm{sub}}$ with $a_{\mathrm{opt}}$ becomes larger. However, for the more important quantity for the tuning, the variance of the log of the absolute value of the estimator, there is a small discrepancy between $a_{\mathrm{opt}}$ and $a_{\mathrm{sub}}$, except for small values of $\exp(B)$. From Plots (a) and (b), we can conclude that, compared with $\widehat{a}_{\mathrm{opt}}$, the usage of $a_{\mathrm{sub}}$ yields closer results for the variance of the block Poisson estimator and its logarithm in most cases.  
Plot (c) shows the probability of the BP estimator being positive against various $\lambda$ values (number of blocks). From the plot, we can see that even with a relatively small $\lambda$ (e.g., $\lambda >10$), the probability of the BP estimator being positive is almost close to 1, regardless of the choice of soft lower bound. Plot (d) represents $\mathrm{Var}(\log |\widehat{L}_B|)$, the variance of the logarithm of the absolute value of the BP estimator. The discrepancies among $\widehat{a}_{\mathrm{opt}}$, $a_\mathrm{sub}$ and $a_\mathrm{opt}$ get narrower in the case of a larger $\lambda$. To conclude, for a wide range of settings, the usage of $a_{\mathrm{sub}}$ yields a BP estimator with similar $\mathrm{Var}(\log |\widehat{L}_B|)$ and high probability of being positive to that of using $a_{\mathrm{opt}}$.

\begin{figure}[htbp]
    \centering
    \includegraphics[width = 0.8\linewidth]{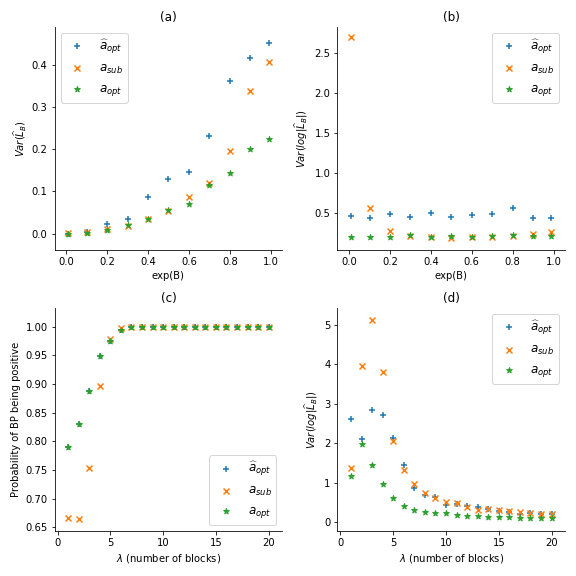}
    \caption{Illustrating some properties of the BP estimator for different choices of soft lower bound $a$.}
    \label{fig:BP-properties}
\end{figure}

\section{An approximate approach under strong assumptions}\label{subsec:bias_corrected}

This section proposes an alternative ``pseudo-marginal approach'' that is approximate and a fast alternative to our signed pseudo-marginal approach. Unlike our signed pseudo-marginal method, it is not simulation-consistent; however, it may be used in the tuning phase of the algorithm to obtain a rough approximation of the posterior. 

If the estimator $\widehat{Z}(\Btheta)$ is normally distributed with a known variance, then an unbiased almost surely positive estimator may be derived using the penalty method in \cite{ceperley1999penalty}. 

Suppose that $\widehat{Z}(\Btheta) \sim N(Z(\Btheta),\sigma^2_{\widehat{Z}(\Btheta)})$, where $\sigma^2_{\widehat{Z}(\Btheta)}$ is the variance of $\widehat{Z}(\Btheta)$. Then $\exp(-\nu \widehat{Z}(\Btheta))$ is log-normally distributed with expected value
$$E(\exp(-\nu \widehat{Z}(\Btheta)))=\exp\left(-\nu Z(\Btheta)+\frac{1}{2}\nu^2 \sigma^2_{\widehat{Z}(\Btheta)}\right).$$
Hence
\begin{equation}\label{eq:bias-corrected estimator}
    \exp\left(-\nu \widehat{Z}(\Btheta)-\frac{1}{2}\nu^2 \sigma^2_{\widehat{Z}(\Btheta)}\right)
\end{equation}
is a positive unbiased estimator of $\exp(-\nu Z(\Btheta))$. Thus, under the idealised assumptions that $\widehat{Z}(\Btheta)$ is normal with a known variance, we can use a pseudo-marginal algorithm to obtain samples from \eqref{eq:target_doubly_intractable}. However, in practice, $\widehat{Z}(\Btheta)$ is rarely normal and the variance must be estimated, so this method provides approximate samples in practice. It is outside the scope of this paper to study the resulting perturbation error. 

The advantage of this method is that it is much faster than the BP estimator, because it only requires a single estimate of the normalising function. However, it needs to be applied carefully due to the strong normality assumption. Section \ref{app:variability_Z_hat} shows that the normality assumption may be unrealistic for certain values of $\theta$ in the Ising model. 

\section{Implementation details of the signed block PMMH with the BP estimator}\label{sec: discussion of algo}

This section gives the implementation details of Algorithm \ref{algo: PMMH_DI}; it covers the construction of the BP estimator and how to choose its lower bound. Moreover, an alternative proposal suggested by an anonymous reviewer is discussed. 

\noindent \textbf{\textit{Construction of the BP estimator}}

To implement the BP estimator, we first fix the hyperparameters $\lambda, m$ and $a$. For each of the blocks $h$,  $h = 1,\dots,\lambda$, we sample $\chi_h \sim \text{Pois}(m)$. Depending on the value of $\chi_h$, we need to have the same number of $-\nu \widehat{Z}$ estimates. The computation can be done in parallel. We can draw $\chi_h$ for all the possible $h$ values at one time, and the total replications of $-\nu\widehat{Z}$ required are $\sum_{h=1}^{\lambda} \chi_h$. Parallel computation can also be implemented within the individual estimation process for $\widehat{Z}$ locally, where the calculation of $M$ particles is executed simultaneously. Particles here refer to samples used to estimate the normalising function.

\noindent \textbf{\textit{Choosing the lower bound in the BP estimator}}

The lower bound $a_{\mathrm{opt}} = -\nu Z - m\lambda$ for the BP estimator minimises the variance of the likelihood estimator. In the implementation,  we use $a_{\mathrm{sub}} =  -1 -m \lambda $, which comes from fixing $\nu$ at its expected value. For multiple observations (such as the Kent example), we set $a_{\mathrm{sub}} =  -n -m \lambda $, where $n$ is the number of observations.

\noindent \textbf{\textit{A deterministic proposal for the auxiliary variable $\nu$}}

An anonymous referee suggested an interesting deterministic proposal for $\nu$ based on its current value. The idea is to minimise the variability of the ratio in \eqref{eq: ar_PMMH}, similar to the idea of only updating a block of the $\Bu_{1:\lambda}$ at a time. The referee suggested the following deterministic proposal
$$\nu' = \nu\frac{\widehat{Z}_P(\Btheta,\Bu)}{\widehat{Z}_P(\Btheta',\Bu')}.$$
To account for the change in volume induced by a deterministic proposal, the acceptance probability must be adjusted by the determinant of the Jacobian of the inverse transformation (see e.g\ \citealp{green1995reversible}),
$$\left|\frac{d\nu}{d\nu'}\right| =  \frac{\widehat{Z}_P(\Btheta',\Bu')}{\widehat{Z}_P(\Btheta,\Bu)}.$$
The alternative version of the acceptance ratio in \eqref{eq: ar_PMMH} under the new proposal is
\begin{equation}\label{eq:acc_prob_referee} \frac{\lvert \widehat{\pi}(\Btheta', \nu', \Bu'_{1:\lambda}|\By,\lambda) \rvert}{\lvert \widehat{\pi}(\Btheta, \nu,\Bu_{1:\lambda}|\By,\lambda) \rvert}
    \dfrac{q(\Btheta| \Btheta')q(\Bu_{1:\lambda}|\Bu'_{1:\lambda})}{q(\Btheta'|\Btheta)q(\Bu'_{1:\lambda}|\Bu_{1:\lambda})}\frac{\widehat{Z}_P(\Btheta,\Bu)}{\widehat{Z}_P(\Btheta',\Bu')}.
\end{equation}
We experimented with this proposal in the Ising model and did not observe results that differed significantly from those with our proposal (independent exponential). Nevertheless, this deterministic proposal may prove useful in other doubly intractable models.

\section{The exchange algorithm and the Russian roulette method}\label{app: ex and rr intro}
\subsection{The exchange algorithm}
The exchange algorithm \citep{murray2006mcmc} constructs an augmented target distribution and updates the distribution with the Metropolis-Hastings algorithm. 

Consider the auxiliary variables $\By'$, with the likelihood $p(\By' | \Btheta') = \dfrac{f(\By'|\Btheta')}{Z(\Btheta')}$. The augmented target distribution is 
\[\pi(\Btheta,\Btheta',\By' | \By) \propto  \pi(\Btheta) q(\Btheta' | \Btheta) \dfrac{f(\By|\Btheta)}{Z(\Btheta)} \dfrac{f(\By' | \Btheta')}{Z(\Btheta')}.\]
For this augmented density, $\Btheta'$ is first generated from the proposal $ q(\Btheta' | \Btheta) $. Then the auxiliary variable $\By'$ is generated from  $ \dfrac{f(\By' | \Btheta')}{Z(\Btheta')}$. In the MH algorithm, consider the proposal that swaps $(\Btheta, \Btheta') \to (\Btheta', \Btheta)$. Then  $\By'$ is evaluated under $\Btheta$ and $\By$ is evaluated under  $\Btheta'$. The unknown normalising functions $Z(\Btheta)$ and $Z(\Btheta')$ in the denominator are cancelled out. The MH acceptance ratio of accepting $\Btheta'$ is then
\begin{equation*}
    \alpha = \min \left\{ 1 ,  \dfrac{\pi(\Btheta')}{\pi(\Btheta)}\dfrac{f(\By |\Btheta')}{f(\By' | \Btheta')} \dfrac{f(\By'| \Btheta)}{f(\By | \Btheta)} \dfrac{q(\Btheta | \Btheta')}{q(\Btheta' | \Btheta)}\right\}.
\end{equation*}

The exchange algorithm requires sampling $\By'$ from the likelihood without knowing the normalising function $Z(\Btheta)$. For the Ising model in Section \ref{sec: ising model}, the perfect sampling method is used \citep{propp1996exact}. For the Kent distribution, the sampling method uses the BACG acceptance-rejection method \citep{kent2013new}.

\subsection{The Russian roulette method}

\cite{lyne2015russian} use the Russian roulette (RR) method to solve the doubly intractable problem. The general idea is to express a known function containing the unknown constant by a geometric series. The series with infinite terms is truncated by the Russian roulette method to obtain an unbiased estimator. The sections below describe two variants of the RR method. The first one uses auxiliary variable $\nu$ in \eqref{eq:augmented_posterior}, the other modelling $1/Z(\Btheta)$ directly. 

\subsubsection{RR with the auxiliary variable (RR-aux)}

Recall the augmented posterior in \eqref{eq:augmented_posterior}, 
\begin{align}\label{eq:augmented_posterior_in_supp}
\pi(\Btheta,\nu | \By) \propto \pi(\Btheta) \dfrac{f(\By | \Btheta) }{Z(\Btheta)}  Z(\Btheta) \exp(-\nu Z(\Btheta)),    
\end{align}
where $\nu \sim \mathrm{Expon}(Z(\Btheta))$, i.e.,\ $p(\nu |
\Btheta) = Z(\Btheta) \exp(-\nu Z(\Btheta))$. The following part outlines the key points and algorithm of the RR method in \cite{lyne2015russian}.

Suppose $\widetilde{Z}(\Btheta)$ is an upper bound of $\widehat{Z}(\Btheta)$ with $E(\widehat{Z}(\Btheta)) = Z(\Btheta)$. The term $\exp(-\nu Z(\Btheta))$ can be expressed as:

\begin{align*}
    \exp(-\nu Z(\Btheta)) &= \exp( -\nu \widetilde{Z}(\Btheta) ) \exp(\nu \widetilde{Z}(\Btheta) - \nu Z(\Btheta))\\
    &=  \exp( -\nu \widetilde{Z}(\Btheta) ) \left(1+ \sum_{n=1}^{\infty}  \dfrac{\nu^n}{n!}(\widetilde{Z}(\Btheta) - Z(\Btheta))^n\right) \\
    &= \exp(-\nu Z(\Btheta)) E_{Z(\Btheta)|\nu} \left( 1 + \sum_{n=1}^{\infty} \dfrac{\nu^n}{n!} \prod_{i=1}^{n} (\widetilde{Z}(\Btheta) - \widehat{Z_i}(\Btheta))\right).
\end{align*}

An unbiased estimator of $\exp(-\nu Z(\Btheta))$ is 
\begin{align}\label{eq:Exp_est_RRaux}
    \widehat{\exp}(-\nu  Z(\Btheta)) & = \exp(-\nu \widetilde{Z}(\Btheta))  \left( 1 + \sum_{n=1}^{\infty} \dfrac{\nu^n}{n!} \prod_{i=1}^{n} (\widetilde{Z}(\Btheta) - \widehat{Z_i}(\Btheta))\right),
\end{align}
\[ \] which contains the summation over an infinite number of terms. The RR method truncates the summation into a finite number of terms, and the truncation is determined by a random stopping time $\tau \geq 1$: $ p_n = P(\tau \geq n)$ for all $n\geq 1$ and $p_1 = 1$. The stopping time could be also defined as $\tau = \inf \{ U_j > \mathcal{Q}_j\}$, where $U_j \sim \mathrm{Uniform}(0,1)$ and $\mathcal{Q}_j \in (0,1]$, which makes $p_n = \prod_{j=1}^{n-1} \mathcal{Q}_j$. \cite{lyne2015russian} select $\mathcal{Q}_j = \dfrac{\nu^j}{j!} \prod_{i=1}^{j} (\widetilde{Z}(\Btheta) - \widehat{Z_i}(\Btheta))$. It is not necessarily $\mathcal{Q}_j \in (0,1]$ for all $j$. Algorithm \ref{alg:RR} uses a threshold $r < 1$ to ensure $\mathcal{Q}_j$ is strictly less than 1. Using Russian roulette, the term $S$,
\[S = \left( 1 + \sum_{n=1}^{\infty} \dfrac{\nu^n}{n!} \prod_{i=1}^{n} (\widetilde{Z}(\Btheta) - \widehat{Z_i}(\Btheta))\right), \] 
is now replaced by the RR estimate $S_\tau$,
\begin{align}\label{eq:Russian_roulette_truncation}  
S_\tau & = \left( 1 + \sum_{n=1}^{\tau-1} \dfrac{\nu^n}{n!} \frac{\prod_{i=1}^{n} (\widetilde{Z}(\Btheta) - \widehat{Z_i}(\Btheta))}{\prod_{j=1}^{n-1} \mathcal{Q}_j}\right)
\end{align}
with $E(S_\tau) = S$.

\begin{algorithm}
\caption{The Russian roulette algorithm for RR-aux}\label{alg:RR}
\begin{algorithmic}
\State Require: $c_{\max}$, $r$, $\widetilde{Z}(\Btheta)$
\State Output: $S_\tau$
\State Let $s_n = \dfrac{\nu^n}{n!} \prod_{i=1}^{n} (\widetilde{Z}(\Btheta) - \widehat{Z_i}(\Btheta))$, then $S_n = 1 + \sum_{i=1}^{n} s_i$.

\State $S_\tau \gets 1$
\State $k \gets 1$
\State $w \gets 1$

\While{$k \leq c_{\max} $}
\If{$s_k < r$} 
    \State $\mathcal{Q}_k \gets s_k /r$
    \State $w \gets w \times \mathcal{Q}_k $
    \State $U\sim \mathrm{Uniform}(0,1)$
    \If{$\mathcal{Q}_k > U$} \Comment{accept, continue the algorithm}
    \State $S_\tau \gets S_\tau + s_k / w$
    \Else  \Comment{terminate}
    \State break;
    \EndIf
\Else
    \State $S_\tau  \gets S_\tau + s_k$ \Comment{$s_k$ is added to the series directly without RR}
\EndIf
\State $k \gets k+1$
\EndWhile
\end{algorithmic}
\end{algorithm}

A signed pseudo-marginal method can be used to target the absolute value of the estimated version of \eqref{eq:augmented_posterior_in_supp}, where the estimate is governed by random variables $\Bu$ that produce \eqref{eq:Russian_roulette_truncation}, including all normalising function estimates and the truncation $\tau$. Unlike our BP estimator, for Russian roulette estimators (including RR below), it is unclear how to induce a correlation to carry out correlated pseudo-marginal following the approach in \cite{tran2016block}, which assumes the likelihood estimator is a product of a fixed number of terms. Thus $\Bu'\sim p(\Bu')$ is proposed independently from $\Bu$ (and cancels the prior for $\Bu$ in the acceptance probability below). Consider the proposal $q(\nu', \Btheta' | \nu, \Btheta) = q(\Btheta' | \Btheta) p(\nu'| \Btheta') = q(\Btheta' | \Btheta) \widetilde{Z}(\Btheta') \exp(-\nu' \widetilde{Z}(\Btheta'))$, the MH acceptance ratio of $\nu', \Btheta'$ is
\begin{equation*}
    \alpha = \min\left\{1, \dfrac{\pi(\Btheta')}{\pi(\Btheta)} \dfrac{f(\By|\Btheta')}{f(\By|
\Btheta)} \dfrac{\widetilde{Z}(\Btheta)}{\widetilde{Z}(\Btheta')} \dfrac{\left( 1 + \sum_{n=1}^{\tau'-1} \dfrac{\nu'^n}{n!} \prod_{i=1}^{n} (\widetilde{Z}(\Btheta') - \widehat{Z_i}(\Btheta'))/\prod_{j=1}^{n-1}\mathcal{Q}'_j\right)}{\left( 1 + \sum_{n=1}^{\tau-1} \dfrac{\nu^n}{n!} \prod_{i=1}^{n} (\widetilde{Z}(\Btheta) - \widehat{Z_i}(\Btheta))/\prod_{j=1}^{n-1}\mathcal{Q}_j\right)} \dfrac{q(\Btheta|\Btheta')}{q(\Btheta'|\Btheta)} \right\}. \end{equation*}

\subsubsection{RR without the auxiliary variable (RR)}
This approach estimates $1/Z(\Btheta)$ unbiasedly without using the auxiliary variable $\nu$. The likelihood function $p(\By|\Btheta)$ is rewritten as

\[ p(\By|\Btheta) = \dfrac{f(\By|\Btheta)}{\widetilde{Z}(\Btheta)} \mathcal{C}(\Btheta) \left( 1+\sum_{n=1}^{\infty} \kappa(\Btheta)^n \right), \] with \[ \kappa(\Btheta)  =  1 - \mathcal{C}(\Btheta) \frac{Z(\Btheta)}{\widetilde{Z}(\Btheta)},\]
where $\mathcal{C}(\Btheta)$ ensures that $|\kappa(\Btheta)|<1$ so that the geometric series converges (recall that $\widetilde{Z}(\Btheta)$ is an upper bound of $\widehat{Z}(\Btheta)$).

An unbiased estimator of the likelihood function is then
\begin{align}\label{eq:RR_reciprocal}
    \widehat{p}(\By |\Btheta, \Bu)  & =  \dfrac{f(\By|\Btheta)}{\widetilde{Z}(\Btheta)} \mathcal{C}(\Btheta) \left( 1+\sum_{n=1}^{\infty} \prod_{i=1}^{n} (1+ \widehat{\kappa}_i(\Btheta))\right), 
\end{align}
where \[ \widehat{\kappa}_i(\Btheta)  =  1 - \mathcal{C}(\Btheta) \frac{\widehat{Z}_i(\Btheta)}{\widetilde{Z}(\Btheta)}.\]
The RR algorithm is similar to Algorithm \ref{alg:RR}, except that $\mathcal{C}(\Btheta)$ is selected to ensure $|\widehat{\kappa}_i(\Btheta)| < 1$. The MH acceptance ratio of $\Btheta'$ is
\begin{equation}\label{eq:acc_prob_RR_reciprocal}
    \alpha = \min\left\{1, \dfrac{\pi(\Btheta')}{\pi(\Btheta)} \dfrac{f(\By|\Btheta')}{f(\By|
\Btheta)} \dfrac{\widetilde{Z}(\Btheta)}{\widetilde{Z}(\Btheta')} \dfrac{\mathcal{C}(\Btheta)}{\mathcal{C}(\Btheta')} \dfrac{\left( 1+\sum_{n=1}^{\tau'-1} \prod_{i=1}^{n} (1 +  \widehat{\kappa}_i(\Btheta'))/\prod_{j=1}^{n-1} \mathcal{Q}'_j\right)}{\left( 1 + \sum_{n=1}^{\tau-1} \prod_{i=1}^{n} (1+ \widehat{\kappa}_i(\Btheta))/\prod_{j=1}^{n-1} \mathcal{Q}_j\right)}\dfrac{q(\Btheta|\Btheta')}{q(\Btheta'|\Btheta)} \right\}, \end{equation}
where $\mathcal{Q}_j = \prod_{i}^j \widehat{\kappa}_i(\Btheta)$. 

\section{Computational aspects of the estimators}\label{app:computational_aspects}
\subsection{Structure of estimators}
The three estimators used in signed-pseudo marginal, i.e.\ BP, RR-aux, and RR, all provide unbiased estimates of the likelihood function (and the posterior distribution, up to a proportionality constant). The first two do so by providing an unbiased estimator of $\exp(-\nu Z(\Btheta)$ in \eqref{eq:augmented_posterior}, while the latter unbiasedly estimates $1/Z(\Btheta)$ in \eqref{eq:target_doubly_intractable}.

The BP estimator of the exponential function (with $m=1$) has the structure
\begin{align}\label{eq:BP_struct}
\widehat{\exp}\left(-\nu Z(\Btheta)\right)_{\rm BP} & = A\prod_{l=1}^{\lambda} \eta_l(\Btheta), \,\, 
\eta_l(\Btheta) =\prod_{h=1}^{\chi_l} \dfrac{-\nu\widehat{Z}^{(h,l)}(\Btheta)-a}{\lambda}, \,\, A = \exp(a + \lambda),
\end{align}
see \eqref{eq: block_pois_est} for details. The structure of the RR-aux estimator of the exponential function is,
\begin{align}\label{eq:RR_aux_struct}
\widehat{\exp}\left(-\nu Z(\Btheta)\right)_{\rm RR-aux} & =   B\left( 1 + \sum_{n=1}^{\tau-1} g(n) \prod_{i=1}^{n} \overline{\eta}_i\right), \,\, \overline{\eta}_i = \widetilde{Z}(\Btheta) - \widehat{Z_i}(\Btheta),
\end{align}
where $g(n)=\dfrac{\nu^n}{n!\prod_{j=1}^{n-1}\mathcal{Q}_j}$ and $ B = \exp(-\nu \widetilde{Z}(\Btheta))$, see \eqref{eq:Exp_est_RRaux} and \eqref{eq:Russian_roulette_truncation} for details. Finally, the structure of the RR estimator of the reciprocal function is 
\begin{align}\label{eq:RR_struct}
\widehat{\frac{1}{Z(\Btheta)}}_{\rm RR} & =   C\left( 1 + \sum_{n=1}^{\tau-1} h(n) \prod_{i=1}^{n}\overline{\overline{\eta}}_i\right), \,\, \overline{\overline{\eta}}_i = 1 + \widehat{\kappa}_i(\Btheta), \,\, C = \frac{\mathcal{C}(\Btheta)}{\widetilde{Z}(\Btheta)},
\end{align}
where $h(n)=\dfrac{1}{\prod_{j=1}^{n-1}\mathcal{Q}_j}$, see \eqref{eq:RR_reciprocal} and \eqref{eq:acc_prob_RR_reciprocal} for details.

\subsection{Arithmetic scaling}\label{subsec:arithmetic_scaling}

We first note that for the estimators in \eqref{eq:BP_struct}, \eqref{eq:RR_aux_struct}, \eqref{eq:RR_struct}, computing $A$, $B$ and $C$ is fast as it depends on chosen hyperparameters (e.g.\ bounds), and the computation is trivial to execute. The more costly component is computing $\eta_l$, $\overline{\eta_i}$, and $\overline{\overline{\eta_i}}$, whose cost can be approximated with that of estimating the normalising function. For the BP,  $\eta_l$ involves on average $1$ estimate of the normalising function (recall $\chi_l\sim\mathrm{Pois}(1)$), and for the Russian roulette approaches, both $\overline{\eta_i}$ and $\overline{\overline{\eta_i}}$ involve 1 estimate of the normalising function.   

Besides $A$, $B$, and $C$, the computation involves, for each estimator, 
$$\prod_{l=1}^{\lambda} \eta_l(\Btheta), \,\, 1 + \sum_{n=1}^{\tau-1} h(n) \prod_{i=1}^{n}\overline{\overline{\eta}}_i, \,\, \text{ and } 1 +\sum_{n=1}^{\tau-1} h(n) \prod_{i=1}^{n}\overline{\overline{\eta}}_i.$$
We now study the order of the arithmetic operations for each estimator. Note that the block-Poisson estimator (first term) has a simple product structure, whereas RR and RR-aux have a nested product structure that is summed. Thus, the complexity of the block-Poisson is $\mathcal{O}(\lambda)$. For RR and RR-aux, the summation is over terms that are $O(n)$, which means there exists a constant $k$ such that $O(n) \leq kn$. Thus, 
\begin{align*}
    \sum_{n=1}^{\tau-1}\mathcal{O}\left(n\right) & \leq \sum_{n=1}^{\tau - 1}kn \\
    & = k\frac{(\tau - 1)\tau}{2},
\end{align*}
and hence the complexity of the RR and RR-aux estimators are $1+\mathcal{O}(\tau^2) = \mathcal{O}(\tau^2)$. 

It may be argued that $\mathcal{O}(\tau^2)$ is the cost of a naive implementation of RR-aux and RR. To see this, note that the products within the sum are nested. Thus, if we at each $n$ of the sum, reuse the already computed $n-1$ terms, the complexity of the term within the sum is $\mathcal{O}(1)$, and therefore the complexity of this more efficient implementation of the RR-aux and RR estimators is also linear, i.e.\  $\mathcal{O}(\tau)$. 

\subsection{Parallelisation and vectorisation}
The product in the BP estimator consists of purely independent terms and can be implemented using vectorised multiplication chains, leveraging modern processor pipelines for parallelisation via single instruction, multiple data (SIMD) \citep[Ch. 4]{warne2022BAtutorial,hennessy2011computer} at an $\mathcal{O}(\lambda)$ cost. The sum in both RR and RR-aux can also be computed via purely independent terms (if the product is fully recomputed for each $n$) and thus leverage SIMD, however, at an $\mathcal{O}(\tau^2)$ cost, corresponding to the sub-optimal implementation.

A practical remark regarding the parallelisation of these estimators is that, although each $\eta_l$ term in the BP estimator involves a random number of sub-terms, the total number of $\eta_l$ terms is fixed at $\lambda$ (non-random). This ensures a predictable global workload, facilitating efficient allocation of parallel computational resources. In contrast, while each $\overline{\eta}_i$ or $\overline{\overline{\eta}}_i$ in the RR and RR-aux estimators involves a single computation, the total number of such terms is random and governed by the truncation level $\tau$. As a result, parallel implementation of RR and RR-aux requires precomputing a conservative upper bound $n_{\max}$ of estimates of the normalisation function, which leads to waste when $n_{\max} > \tau - 1$.

We note that the more efficient linear arithmetic implementation of RR-aux and RR in Section \eqref{subsec:arithmetic_scaling}, on the other hand, imposes a sequential computation pattern and prevents vectorisation: each step depends on the completion of the previous product extension (and sum update). This introduces strict dependencies and limits opportunities for hardware parallelism.

\subsection{Practical execution-time considerations}\label{subsec:practical_execution_considerations}
From a practical execution time perspective, BP estimators may execute substantially faster, not due to fewer arithmetic operations, but due to computational parallelism and vectorised operations. This performance advantage becomes more pronounced when the estimates of the normalising function can be obtained cheaply relative to the arithmetic operations, such as in the Kent distribution example. In contrast, for the Ising example, where obtaining the normalising functions is the primary computational bottleneck, the arithmetic operations and their parallel execution constitute a small part of the overall computing time, rendering less pronounced computational gains for the block-Poisson estimator.

\section{Details of the Ising model}\label{app: ising model}

\subsection{An unbiased estimator for the normalising function}
This section supplements the material on the annead importance sampling (AIS) in Section \ref{sec: ising model}.  The likelihood function is $p(\By|\theta) = f(\By|\theta)/Z(\theta)$, with $f(\By|\theta) = \exp(\theta S(\By)) $.

Consider the following intermediate kernel of the likelihood function 
\begin{equation}
    f_i(\By|\theta) = f(\By|\theta)^{\beta_i} p(\By)^{1-\beta_i}, i = 0,\dots,N,\notag
\end{equation}
where $0 = \beta_0 < \beta_1 < \dots <\beta_{N-1} <\beta_N = 1$ and $p(\By) = 0.5^{L\times L}$. In the case of $\beta_0$, sampling from the prior density $f_N(\By|\theta)$ of $\beta_0$ is straightforward. By gradually increasing $\beta_i$, the generated data samples are drawn from the desired likelihood function $f_N(\By|\theta) = f(\By|\theta)$ after $N$ steps without knowing the normalising function. The algorithm starts by sampling $M$ particles (samples) from $f_0(\By|\theta)$, and proceeds with a certain transition probability to a new configuration for $\beta_i$ ($i=1,\dots,N-1)$ and terminates when $\beta_N=1$ is reached.

The transition to a new configuration $\By_{i+1,m}$ ($m=1,\dots, M$) from the current configuration $\By_{i,m}$ is completed by the following Gibbs update. 
\begin{enumerate}
    \item [1] Select one random location $i,j$ out of a $L\times L$ grid.
    \item [2] Change the corresponding value of $y_{i,j} $ with probability
    \begin{equation*}
        p(y_{i,j}=1) = \dfrac{1}{1 + \exp(- \beta_{i+1} \theta \sum y_{\rm neighbour}) 
        },
    \end{equation*}
    where $y_{\rm neighbour}$ refers to the adjacent points in all four cardinal directions: left, right, up, and down of $y_{i,j}$.
    \item [3] Set the new configuration as $y_{i+1,m}$.
\end{enumerate}

The final weight associated with particle $m$ (omitting $\theta$) is
\begin{equation*}
    w^{(m)} = \dfrac{f_1(\By_{0,m})}{f_0(\By_{1,m})} \dfrac{f_2(\By_{1,m})}{f_1(\By_{1,m})} \dots  \dfrac{f_{n}(\By_{n-1,m})}{f_{n-1}(\By_{n-1,m})}.
\end{equation*}
The average of the importance weights $\sum_{m=1}^M w^{(m)} /N $ converges to the ratio of $Z_1(\theta)/Z_0(\theta)$, where $Z_i(\theta)$ corresponds the normalising function of $f_i(y|\theta) $.

We work with  $\log w^{(m)}$ to avoid overflow problems. Parallel computation is possible because the particles (the samples to estimate the normalising functions) are independent. Our description follows the supplementary code in \cite{park2018bayesian}, which uses OpenMP to implement the parallel computation. The re-evaluation of $S(\By)$ from scratch can be computationally costly as it involves $O(L^2)$ operations for each combination of $\beta_i$ and particle $m$. We modify the evaluation process by adding or subtracting the local updates of the selected location only. Such changes reduce the complexity to $O(1)$ and consequently decrease the computational time substantially.

\subsection{The variability of the normalising function}\label{app:variability_Z_hat}

The estimate $\widehat{Z}$ and its variance are crucial in hyperparameter tuning for the proposed algorithm. As the Ising model involves one parameter, it is feasible to study the variability of $\widehat{Z}$ by simulation.  Figure \ref{fig:histogram_Ztheta}  shows the estimates of the scaled $\widehat{Z}$ under different $\theta$ values, where $\widehat{Z}$ is rescaled by dividing by the sample mean of the replications. Each histogram is generated by 1{,}000 independent replications, each of which uses 100 particles in AIS (samples in AIS) with 4{,}000 intermediate transitions equally spaced between 0 and 1. The horizontal axis refers to the scaled $\widehat{Z}(\theta)$. As $\theta$ increases, the distribution of the scaled $\widehat{Z}$ is heavily skewed and the normality assumption appears to be invalid for $\theta > 0.4$. Such a violation explains the overestimation by the bias-corrected estimator for $\theta = 0.43$ in the example in Section \ref{sec: ising model}. Examining the range of the horizontal axis, the magnitude also
increases sharply with $\theta$, which implies that a larger $\theta$ is associated with more variability in $\widehat{Z}$. Hence, more particles (samples) are required to estimate $\widehat{Z}$ as  $\theta$ increases.

\begin{figure}[htbp]
    \centering
    \includegraphics[width = 0.8\linewidth]{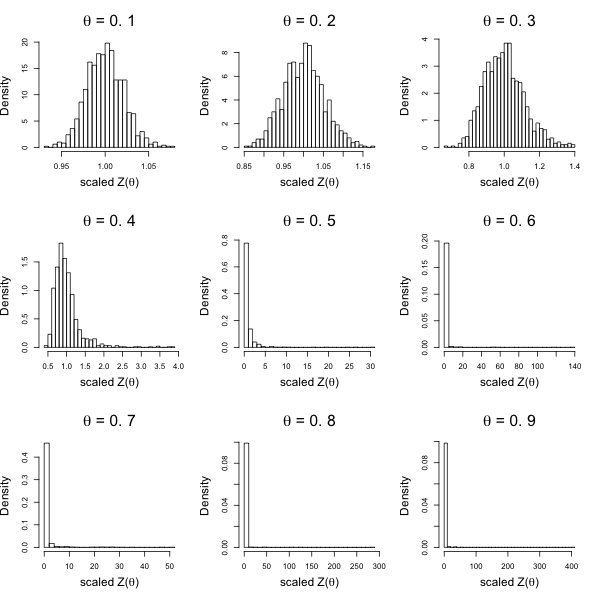}
    \caption{Histograms of scaled $Z(\theta)$ estimates on a $10 \times 10$ 2D Ising model. }
    \label{fig:histogram_Ztheta}
\end{figure}

\subsection{Hyperparameter values of RR and RR-aux}
Hyperparameters for RR and RR-aux in Algorithm \ref{alg:RR} include $c_{\max}$, $r$, and $\mathcal{C}(\theta)$ (RR only). Another term to determine is $\widetilde{Z}(\theta)$, which is the upper bound of $\widehat{Z}(\theta)$.

In RR and RR-aux, $c_{\max} =50$, which is the maximum number of terms for summation in RR and RR-aux. When the number of iterations is greater than 50, the algorithm terminates and returns the series, i.e.,\ a summation over 50 terms.  The threshold $r=0.6$ is used to compare with each single term $s_k$. A smaller $r$ value is associated with longer computing time, but it provides an estimator with smaller variance.  RR uses $\mathcal{C}(\theta) = 0.4$, which ensures $\widehat{\kappa}_i(\theta)$ is strictly less than 1. A smaller $\mathcal{C}(\theta)$ leads to a larger $\widehat{\kappa}_i(\theta)$, which results in the estimator with smaller variance but with a longer computing time. We double the number of importance samples in AIS to calculate $\widetilde{Z}(\theta)$ for a large $\theta$; 200 for $\theta = 0.2$, 1,000 for $\theta = 0.43$.

 The hyperparameters are crucial for the efficiency of RR and RR-aux methods. However, the tuning rules are not provided in \cite{lyne2015russian}. We could select a much smaller $\mathcal{C}(\Btheta)$ and $r$ to ensure the convergence of RR-aux when $\theta = 0.43$, but it leads to a substantially longer computing time. 

 \subsection{Additional results for the simulation study}\label{subsec_appendix:additional_Ising}

Figure \ref{fig:Ising_traceplots} shows trace plots of the MCMC draws of $\theta$ for the Ising model. The figure shows good mixing for all methods, except for RR-aux which gets stuck around iteration $1{,}000$. 

\begin{figure}[h!]
    \centering
    \includegraphics[width=0.8\linewidth]{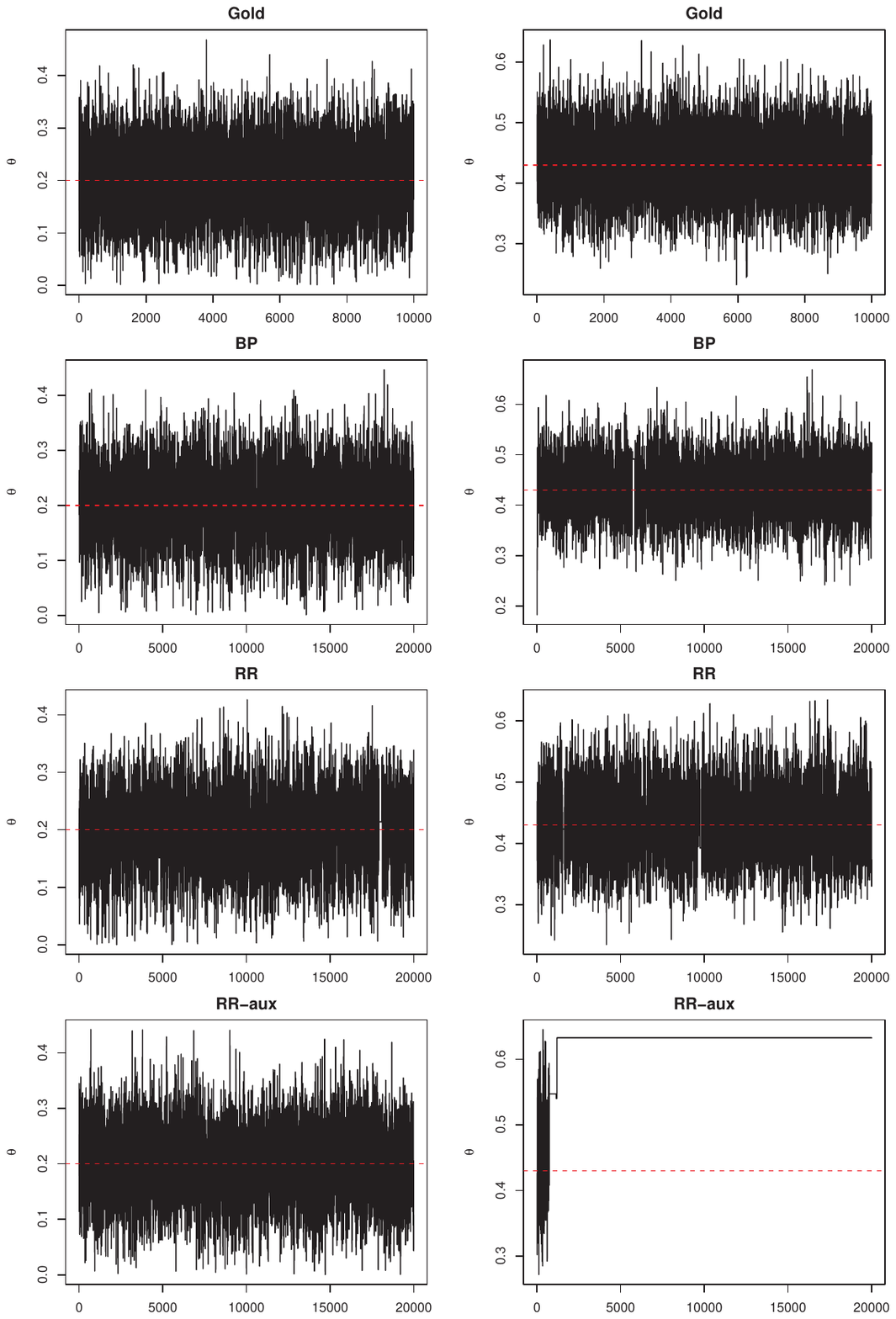}
    \caption{Trace plots of MCMC draws of $\theta$. The horizontal axis refers to the iteration index. The red dotted line represents true $\theta$ value (left panel: $\theta = 0.2$, right panel: $\theta = 0.43$). Note that for the Gold result, only 10,000 posterior samples are plotted after thinning from 1,000,000 iterations.}
    \label{fig:Ising_traceplots}
\end{figure}

\section{Details of the Kent distribution}\label{app: kent distr}

\begin{figure}[t]
    \centering
    \begin{minipage}[b]{0.48\textwidth}
        \includegraphics[width=\textwidth]{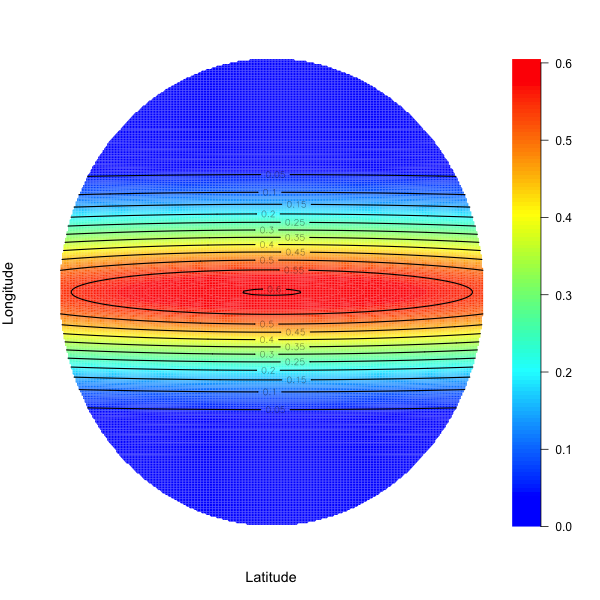}
        \label{fig:}
    \end{minipage}
    \hfill
    \begin{minipage}[b]{0.48\textwidth}
\includegraphics[width=\textwidth]{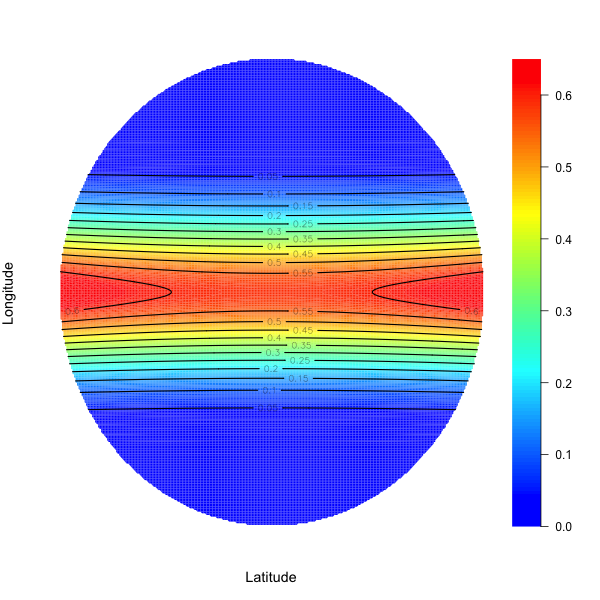}
        \label{fig:sub2}
    \end{minipage}
\caption{Contour plot of a Kent distribution with $\kappa =5$, the corresponding $\beta$ values are 2.45 (left) and 2.55 (right), corresponding to $\beta/\kappa$  = 0.49 and 0.51, respectively.}
    \label{fig:kent_uni_and_bimodal}
\end{figure}

\subsection{Posterior distribution}\label{eq:Kent_posterior}
 In Example \ref{Ex: Kent} of Section \ref{sec: doubly intract prob}, the density function of the Kent distribution is 
 \begin{align*}
    p(\By|\Btheta) = \frac{1}{c(\kappa,\beta)} \exp\left\{\kappa \boldsymbol{\gamma}_1^\top \cdot \By + \beta \left[(\boldsymbol{\gamma}_2^\top \cdot \By)^2 - (\boldsymbol{\gamma}_3^\top \cdot \By)^2\right] \right\} = \frac{f(\By |\Btheta)}{c(\kappa,\beta)},
\end{align*}
with $\By = \{ y_1,y_2,y_3\}$, $\sum_{i=1}^3 y_i^2 = 1$ and $\Btheta = \{ \boldsymbol{\gamma}_1,\boldsymbol{\gamma}_2,\boldsymbol{\gamma}_3,\beta,\kappa \}$; $\kappa > 0, 0\leq \beta < \kappa/2$. The restriction on $\beta$ ensures that the data distribution is unimodal, see Figure \ref{fig:kent_uni_and_bimodal} for an example when $\beta/\kappa$ is close to the boundary (left panel) and when it falls outside (right panel). The parameters $\{ \boldsymbol{\gamma}_1, \boldsymbol{\gamma}_2, \boldsymbol{\gamma}_3\}$ form a 3-dimensional orthonormal matrix with $\boldsymbol{\gamma}_i, i=1,2,3$ is a $3\times 1$ vector. For a 3-dimensional $\mathrm{FB}_5$ distribution, the normalising function is 
$$c(\kappa,\beta) = 2\pi \sum_{j=0}^{\infty} \dfrac{\Gamma(j+0.5)}{\Gamma(j+1)} \beta^{2j} (0.5\kappa)^{-2j-0.5} I_{2j+0.5}(\kappa).$$

We now extend the joint posterior distribution of $\Btheta$  and the auxiliary variables $\nu$ in \eqref{eq:augmented_posterior} to the case of multiple observations. This is the augmented posterior that BP and RR-aux use in the signed pseudo-marginal approach. Assume $n$ independent observations $\By_{1:n}:= \{\By_1, \ldots, \By_n\}$ from an $\mathrm{FB}_5$ distribution, together with the auxiliary variables $\nu_i \sim \text{Exp}(c(\kappa,\beta)), i= 1,\dots,n$. The posterior distribution is 
\begin{align}\label{eq:augmented_auxiliary_multiple_n}
    \pi(\Btheta, \nu_{1:n}|\By_{1:n}) & \propto \pi(\Btheta) \prod_{i = 1}^n f(\By_i |\Btheta)  \exp\left(-\nu_i c(\kappa,\beta)\right) \notag \\ 
    &=\pi(\Btheta) \exp\left(-\sum_{i=1}^n \nu_i c(\kappa,\beta)\right) \prod_{i = 1}^n f(\By_i |\Btheta).
\end{align}
Note that we only need to estimate the normalising function $c(\kappa,\beta)$ once, regardless of the number of observations. BP and RR-aux (and the approximate method in Section \ref{subsec:bias_corrected} if used in the tuning procedure) use this posterior distribution, and the only difference lies in the estimation method of $\exp(\cdot)$. 

Without the auxiliary variable $\nu$, the normalising function appears as a reciprocal instead of an exponent. 
This is the posterior that RR uses in the signed pseudo-marginal approach. For multiple observations,
\begin{align}\label{eq:augmented_recip_multiple_n}
    \pi(\Btheta|\By_{1:n}) & \propto \pi(\Btheta) \prod_{i = 1}^n \left( 
    \frac{f(\By_i |\Btheta)}{c(\kappa,\beta)} \right) 
 =\pi(\Btheta) c^{-n}(\kappa,\beta) \prod_{i = 1}^n f(\By_i |\Btheta).
\end{align}
Note that, unlike \eqref{eq:augmented_auxiliary_multiple_n}, where the argument of the exponent can be estimated unbiasedly when $E\left( \widehat{c}(\kappa, \beta)\right)=c(\kappa, \beta)$, the argument of the reciprocal in \eqref{eq:augmented_recip_multiple_n} cannot, since $E\left( \widehat{c}(\kappa, \beta)^n\right) \neq c(\kappa, \beta)^n$ in this case. Thus, for RR, each term $c^{-1}(\kappa, \beta)$ needs to be estimated individually and independently for each observation, which dramatically increases the computation. This is the reason RR is classified as not scalable to multiple observations in Table \ref{tab:methods summary}. Finally, we note that \eqref{eq:augmented_recip_multiple_n} is also the target of the exchange algorithm.

\subsection{Implementation details}

For all Bayesian methods except for Exchange (where $c(\kappa,\beta)$ cancels), we compute the first 3 terms exactly in $c(\kappa,\beta)$, i.e.\ $K=3$,  and perform a truncation of the remaining terms. 

In the simulation study in Section \ref{subsec: kent_dist}, BP selects $m = 1$ and chooses the number of blocks $\lambda = 50, 50, 100$ for $n = 10, 100$ and $1{,}000$ respectively.  If used for the tuning procedure, we fix the number of samples as 50 to obtain the variance of $c(\kappa,\beta)$ in the implementation of the approximate method in Section \eqref{subsec:bias_corrected}. For RR-aux, $r = 0.6$, $\mathcal{C}(\Btheta) = 0.4$ (RR only) and $c_{\max} = 50$ (Algorithm \ref{alg:RR}). The upper bound for the estimator of $c(\kappa,\beta)$ is set to 110\% of the summation of the first 10 truncated terms in  $c(\kappa,\beta)$.

The MLE method is based on our modification of the function \texttt{kent.mle} of the R package \texttt{Directional}, where the original version uses the moment estimates of $\gamma$'s. We use an optimiser on the transformed parameters to obtain the MLE of all the parameters in the modified version.

In the empirical study, as the number of observations is small to moderate (ranging from 9 to 101), we select $m = 1$, $\lambda = 20$ for BP, $r=0.2$, $\mathcal{C}(\theta)$ = 0.8 (RR-only) for RR and RR-aux.

\subsection{Additional results for the simulation study} \label{app: results of FB5_sim}
Tables \ref{tab:FB5_simulation_beta} and \ref{tab:FB5_simulation_beta_over_kappa} show additional results for the estimation of $\beta$ and $\beta/\kappa$ in the simulation study, respectively. Similar conclusions to those in Section \ref{subsec: kent_dist} of the main paper are obtained, where our method BP outperforms the Bayesian alternatives (RR, RR-aux, and Exchange), particularly for larger $n$.
\begin{table}[]
    \centering
    \small{
    \begin{tabular}{r l l r | l l r|l l r}
    \toprule
               & \multicolumn{9}{c}{$\beta/\kappa$ = 0.01} \\
           \midrule
               & \multicolumn{3}{c}{$n=10$} & \multicolumn{3}{c}{$n=100$} &\multicolumn{3}{c}
 {$n=1{,}000$}\\
  & $\mathrm{RMSE}_\beta$ & $\mathrm{ESS}_{\beta} $ &  $\mathrm{ESS}_{\beta} $/s& $\mathrm{RMSE}_\beta$ & $\mathrm{ESS}_{\beta} $ &  $\mathrm{ESS}_{\beta} $/s  & $\mathrm{RMSE}_\beta$ & $\mathrm{ESS}_{\beta} $ &  $\mathrm{ESS}_{\beta} $/s \\
    \cmidrule(lr){2-4}\cmidrule(lr){5-7}\cmidrule(lr){8-10} 
BP& 1.36 & 119.8 &13.1 & 0.40 & 112.7 &12.2 & 0.10 & 110.3 &11.3\\
RR& 1.44 & 14.1 &0.3 & 0.03 & 4.5 & $<$0.1 & 0.01 & 3.9 & $<$0.1\\
RR-aux& 1.33 & 141.0 &12.4 & 0.41 & 150.7 &4.6 & 0.10 & 35.9 &0.6\\
Moment& 1.55 & - &- & 0.22 & - &- & 0.05 & - &-\\
MLE& 2.46 & - &- & 0.43 & -&- & 0.11 & - &-\\
Exch.& 2.70 & 103.8 &8.5 & 1.00 & 136.1 &9.6 & 0.62 & 107.7 &4.5\\
\midrule
           & \multicolumn{9}{c}{$\beta/\kappa$ = 0.25} \\
           \midrule
               & \multicolumn{3}{c}{$n=10$} & \multicolumn{3}{c}{$n=100$} &\multicolumn{3}{c}
 {$n=1{,}000$}\\
  & $\mathrm{RMSE}_\beta$ & $\mathrm{ESS}_{\beta} $ &  $\mathrm{ESS}_{\beta} $/s& $\mathrm{RMSE}_\beta$ & $\mathrm{ESS}_{\beta} $ &  $\mathrm{ESS}_{\beta} $/s  & $\mathrm{RMSE}_\beta$ & $\mathrm{ESS}_{\beta} $ &  $\mathrm{ESS}_{\beta} $/s \\
    \cmidrule(lr){2-4}\cmidrule(lr){5-7}\cmidrule(lr){8-10} 
BP& 0.66 & 99.3 &11.0 & 0.37 & 121.8 &13.3 & 0.13 & 129.6 &13.2\\
RR& 0.73 & 13.7 &0.3 & 0.35 & 3.8 & $<$0.1  & 0.18 & 2.7 & $<$0.1 \\
RR-aux& 0.64 & 114.1 &10.1 & 0.37 & 157.7 &4.8 & 0.13 & 31.2 &0.5\\
Moment& 1.06 & -&- & 0.60 & -&- & 0.64 & -&-\\
MLE& 2.03 & -&- & 0.38 & -&- & 0.14 & -&-\\
Exch.& 2.30 & 84.3 &6.7 & 1.58 & 79.1 &5.4 & 1.47 & 64.8 &2.6\\
\midrule
           & \multicolumn{9}{c}{$\beta/\kappa$ = 0.49} \\
           \midrule
               & \multicolumn{3}{c}{$n=10$} & \multicolumn{3}{c}{$n=100$} &\multicolumn{3}{c}
 {$n=1{,}000$}\\
  & $\mathrm{RMSE}_\beta$ & $\mathrm{ESS}_{\beta} $ &  $\mathrm{ESS}_{\beta} $/s& $\mathrm{RMSE}_\beta$ & $\mathrm{ESS}_{\beta} $ &  $\mathrm{ESS}_{\beta} $/s  & $\mathrm{RMSE}_\beta$ & $\mathrm{ESS}_{\beta} $ &  $\mathrm{ESS}_{\beta} $/s \\
    \cmidrule(lr){2-4}\cmidrule(lr){5-7}\cmidrule(lr){8-10} 
BP& 1.19 & 77.4 &8.5 & 0.39 & 102.4 &11.3 & 0.13 & 86.0 &8.8\\
RR& 1.15 & 11.0 &0.3 & 0.39 & 2.9 & $<$0.1  & 0.73 & 2.8 & $<$0.1 \\
RR-aux& 1.19 & 88.3 &8.0 & 0.38 & 134.7 &4.1 & 0.14 & 15.4 &0.3\\
Moment& 1.25 & -&- & 1.45 & -&- & 1.50 & -&-\\
MLE& 1.71 & -&- & 0.54 & -&- & 0.25 & -&-\\
Exch. & 2.51 & 68.3 &5.4 & 2.21 & 103.5 &6.8 & 1.91 & 100.7 &3.9\\
\bottomrule
\\
    \end{tabular}}
\caption{ Simulation results when estimating $\beta$  using 100 independent replications of an $\mathrm{FB}_5$ distribution. The RMSE numbers are with respect to the true value $\beta=0.05,1.25,2.45$. The $\mathrm{ESS_\beta}$/s is the effective sample size (altered with sign correction) per second.}
\label{tab:FB5_simulation_beta}
\end{table}

\begin{table}[]
    \centering
    \small{\begin{tabular}{r l l l | l l l|l l l}
    \toprule
               & \multicolumn{9}{c}{$\beta/\kappa$ = 0.01} \\
           \midrule
               & \multicolumn{3}{c}{$n=10$} & \multicolumn{3}{c}{$n=100$} &\multicolumn{3}{c}
 {$n=1{,}000$}\\
  & $\mathrm{RMSE}_{\frac{\beta}{\kappa}}$ & $\mathrm{ESS}_{\frac{\beta}{\kappa}} $ & $\mathrm{ESS}_{\frac{\beta}{\kappa}}$/s& $\mathrm{RMSE}_{\frac{\beta}{\kappa}}$ & $\mathrm{ESS}_{\frac{\beta}{\kappa}} $ & $\mathrm{ESS}_{\frac{\beta}{\kappa}}$/s  & $\mathrm{RMSE}_{\frac{\beta}{\kappa}}$ & $\mathrm{ESS}_{\frac{\beta}{\kappa}} $ & $\mathrm{ESS}_{\frac{\beta}{\kappa}}$/s \\
    \cmidrule(lr){2-4}\cmidrule(lr){5-7}\cmidrule(lr){8-10} 
BP& 0.22 & 147.3 &16.2 & 0.08 & 115.4 &12.5 & 0.02 & 110.4 &11.3\\
RR& 0.21 & 117.6 &2.8 & 0.01 & 4.0 &$<0.1$ & 0.00 & 4.3 &$<0.1$\\
RR-aux& 0.22 & 170.6 &15.0 & 0.08 & 155.0 &4.7 & 0.02 & 35.9 &0.6\\
Moment& 0.16 & -&- & 0.04 & -&- & 0.01 & -&-\\
MLE& 0.28 & -&- & 0.08 & -&- & 0.02 & -&-\\
Exch.& 0.26 & 140.3 &11.5 & 0.18 & 145.6 &10.2 & 0.12 & 109.7 &4.6\\
\midrule
           & \multicolumn{9}{c}{$\beta/\kappa$ = 0.25} \\
           \midrule
               & \multicolumn{3}{c}{$n=10$} & \multicolumn{3}{c}{$n=100$} &\multicolumn{3}{c}
 {$n=1{,}000$}\\
  & $\mathrm{RMSE}_{\frac{\beta}{\kappa}}$ & $\mathrm{ESS}_{\frac{\beta}{\kappa}} $ & $\mathrm{ESS}_{\frac{\beta}{\kappa}}$/s& $\mathrm{RMSE}_{\frac{\beta}{\kappa}}$ & $\mathrm{ESS}_{\frac{\beta}{\kappa}} $ & $\mathrm{ESS}_{\frac{\beta}{\kappa}}$/s  & $\mathrm{RMSE}_{\frac{\beta}{\kappa}}$ & $\mathrm{ESS}_{\frac{\beta}{\kappa}} $ & $\mathrm{ESS}_{\frac{\beta}{\kappa}}$/s \\
    \cmidrule(lr){2-4}\cmidrule(lr){5-7}\cmidrule(lr){8-10} 
BP& 0.05 & 121.1 &13.4 & 0.07 & 125.4 &13.7 & 0.02 & 131.1 &13.3\\
RR& 0.08 & 112.7 &2.7 & 0.06 & 3.6 &$<0.1$ & 0.04 & 2.7 &$<0.1$\\
RR-aux& 0.05 & 135.7 &12.1 & 0.07 & 163.5 &5.0 & 0.02 & 31.2 &0.5\\
Moment& 0.10 & -&- & 0.12 & -&- & 0.13 & -&-\\
MLE& 0.15 & -&- & 0.06 & -&- & 0.03 & -&-\\
Exch.& 0.06 & 110.4 &8.8 & 0.12 & 98.9 &6.7 & 0.14 & 79.3 &3.2\\
\midrule
           & \multicolumn{9}{c}{$\beta/\kappa$ = 0.49} \\
           \midrule
               & \multicolumn{3}{c}{$n=10$} & \multicolumn{3}{c}{$n=100$} &\multicolumn{3}{c}
 {$n=1{,}000$}\\
  & $\mathrm{RMSE}_{\frac{\beta}{\kappa}}$ & $\mathrm{ESS}_{\frac{\beta}{\kappa}} $ & $\mathrm{ESS}_{\frac{\beta}{\kappa}}$/s& $\mathrm{RMSE}_{\frac{\beta}{\kappa}}$ & $\mathrm{ESS}_{\frac{\beta}{\kappa}} $ & $\mathrm{ESS}_{\frac{\beta}{\kappa}}$/s  & $\mathrm{RMSE}_{\frac{\beta}{\kappa}}$ & $\mathrm{ESS}_{\frac{\beta}{\kappa}} $ & $\mathrm{ESS}_{\frac{\beta}{\kappa}}$/s \\
    \cmidrule(lr){2-4}\cmidrule(lr){5-7}\cmidrule(lr){8-10} 
BP& 0.21 & 98.4 &10.9 & 0.06 & 110.2 &12.1 & 0.01 & 119.4 &12.2\\
RR& 0.19 & 263.4 &6.3 & 0.07 & 3.0 &$<0.1$ & 0.11 & 3.8 &$<0.1$\\
RR-aux& 0.21 & 115.2 &10.4 & 0.06 & 144.5 &4.4 & 0.02 & 22.5 &0.4\\
Moment& 0.26 & -&- & 0.28 & -&- & 0.28 & 0.0 &0.0\\
MLE& 0.13 & -&- & 0.09 & -&- & 0.04 & 0.0 &0.0\\
Exch.& 0.18 & 92.4 &7.4 & 0.03 & 112.0 &7.3 & 0.03 & 117.0 &4.5\\
\bottomrule
\\
    \end{tabular}}
\caption{ Simulation results when estimating $\beta/\kappa$ using 100 independent replications of an $\mathrm{FB}_5$ distribution. The RMSE numbers are with respect to the true values of $\beta/\kappa$. The $\mathrm{ESS}_{\frac{\beta}{\kappa}}$/s is the effective sample size (altered with sign correction) per second.  }
\label{tab:FB5_simulation_beta_over_kappa}
\end{table}

\subsection{Additional results for the empirical study} \label{app: results of FB5}
Figures \ref{fig:FB5_kappa} and \ref{fig:FB5_ratio} show additional results for the estimation of $\beta$ and $\beta/\kappa$ in real data application, respectively. Similar conclusions to those in Section \ref{sec: empirical study} of the main paper are obtained, where our method BP outperforms the Bayesian alternatives (RR, RR-aux, and Exchange).

\begin{table}{}
\small{\begin{tabular}{lccc|ccc}
\toprule
Paleo & \multicolumn{3}{c}{Group 1 (n=9)} &\multicolumn{3}{c}{Group 2 (n=24)} \\
    & $\beta$ & $\mathrm{ESS}_{\beta}$ & $\mathrm{ESS}_{\beta}$/s & $\beta$ & $\mathrm{ESS}_{\beta}$ & $\mathrm{ESS}_{\beta}$/s  \\
\cmidrule(lr){2-4} \cmidrule(lr){5-7}
       BP &  2.68 &   183.91 &         7.29 &  4.97 &   198.29 &        20.76 \\
       RR &  3.03 &    18.98 &         0.53 &  4.02 &   182.33 &         2.14 \\
   RR-aux &  3.09 &   136.99 &         9.63 &   3.3 &   144.31 &         7.97 \\
 Exchange &  4.61 &    82.11 &         6.12 &  6.65 &    54.03 &         4.15 \\
   Moment &  4.03 &      - &          - &  5.88 &      - &          - \\
      MLE &  4.55 &      - &          - &  6.44 &      - &          - \\
  \midrule
Magnetic & \multicolumn{3}{c}{Group 1 (n=62)} &\multicolumn{3}{c}{Group 2 (n=62)} \\
    & $\beta$ & $\mathrm{ESS}_{\beta}$ & $\mathrm{ESS}_{\beta}$/s & $\beta$ & $\mathrm{ESS}_{\beta}$ & $\mathrm{ESS}_{\beta}$/s \\
   \cmidrule(lr){2-4} \cmidrule(lr){5-7}
       BP &   7.15 &   166.73 &        16.69 &  14.18 &   126.19 &        12.98 \\
       RR &   7.62 &    25.77 &         0.15 &  14.88 &     1.53 &         0.01 \\
   RR-aux &   7.43 &   174.93 &          6.4 &  14.51 &    48.32 &         1.66 \\
 Exchange &  21.92 &   101.14 &         6.57 &  51.89 &    66.67 &         4.34 \\
   Moment &   4.55 &      - &          - &   8.87 &      - &          - \\
      MLE &   8.24 &      - &          - &  15.57 &      - &          - \\
      \midrule
Sandstone & \multicolumn{3}{c}{Group 1 (n=35)} &\multicolumn{3}{c}{Group 2 (n=13)} \\
    & $\beta$ & $\mathrm{ESS}_{\beta}$ & $\mathrm{ESS}_{\beta}$/s & $\beta$ & $\mathrm{ESS}_{\beta}$ & $\mathrm{ESS}_{\beta}$/s \\
   \cmidrule(lr){2-4} \cmidrule(lr){5-7}
                    BP &  1.54 &   142.63 &         5.66 &   8.86 &    95.27 &        10.08 \\
       RR &  1.24 &    86.89 &         0.78 &   7.56 &   121.59 &         1.76 \\
   RR-aux &   1.8 &   185.85 &         8.86 &   7.65 &    96.16 &         6.37 \\
 Exchange &  3.31 &   109.17 &          7.6 &  13.69 &     4.92 &         0.38 \\
   Moment &  2.07 &      - &          - &  18.45 &      - &          - \\
      MLE &  2.42 &      - &          - &  20.15 &      - &          - \\
      \midrule
Stone & \multicolumn{3}{c}{Group 1 (n=101)} &\multicolumn{3}{c}{Group 2 (n=101)} \\
    & $\beta$ & $\mathrm{ESS}_{\beta}$ & $\mathrm{ESS}_{\beta}$/s & $\beta$ & $\mathrm{ESS}_{\beta}$ & $\mathrm{ESS}_{\beta}$/s \\
   \cmidrule(lr){2-4} \cmidrule(lr){5-7}
                    BP &  0.57 &    80.23 &         3.21 &  1.04 &    32.84 &         3.41 \\
       RR &  0.48 &    10.55 &         0.03 &  0.97 &    10.53 &         0.04 \\
   RR-aux &   0.6 &     75.8 &         2.05 &  1.07 &    37.73 &         1.03 \\
 Exchange &  2.22 &    56.21 &         3.95 &  1.44 &    64.99 &         5.02 \\
   Moment &  0.23 &      - &          - &  0.41 &      - &          - \\
      MLE &   0.6 &      - &          - &   1.1 &      - &          - \\
\bottomrule
\end{tabular}}
\caption{Results for the Kent model of the four datasets when estimating $\beta$. All the chains ran for 10,000 iterations. The term $\beta$ refers to the posterior mean for the Bayesian methods and estimates for the frequentist methods. The sign correction is applied for BP, RR and RR-aux. }\label{tab:FB5_empirical_beta}
\end{table}

\begin{table}{}
\small{\begin{tabular}{lccc|ccc}
\toprule
Paleo & \multicolumn{3}{c}{Group 1 (n=9)} &\multicolumn{3}{c}{Group 2 (n=24)} \\
    & $\kappa$ & $\mathrm{ESS}_{\kappa}$ & $\mathrm{ESS}_{\kappa}$/s & $\kappa$ & $\mathrm{ESS}_{\kappa}$ & $\mathrm{ESS}_{\kappa}$/s  \\
\cmidrule(lr){2-4} \cmidrule(lr){5-7}
       BP &  18.93 &    195.43 &          7.75 &  36.04 &    214.66 &         22.47 \\
       RR &  19.95 &    191.09 &          5.31 &   33.1 &    138.81 &          1.63 \\
   RR-aux &   18.8 &    150.38 &         10.57 &  33.79 &    124.84 &           6.9 \\
 Exchange &  21.62 &      72.4 &           5.4 &  38.37 &     52.11 &             4 \\
   Moment &  26.54 &       - &           - &  39.84 &       - &           - \\
      MLE &  26.65 &       - &           - &  40.02 &       - &           - \\
  \midrule
Magnetic & \multicolumn{3}{c}{Group 1 (n=62)} &\multicolumn{3}{c}{Group 2 (n=62)} \\
     & $\kappa$ & $\mathrm{ESS}_{\kappa}$ & $\mathrm{ESS}_{\kappa}$/s & $\kappa$ & $\mathrm{ESS}_{\kappa}$ & $\mathrm{ESS}_{\kappa}$/s  \\
   \cmidrule(lr){2-4} \cmidrule(lr){5-7}
        BP &  15.01 &     161.8 &          16.2 &   29.09 &    125.53 &         12.91 \\
       RR &  15.92 &     25.33 &          0.15 &   31.04 &      1.72 &          0.01 \\
   RR-aux &  15.43 &    183.07 &          6.69 &   29.63 &     48.07 &          1.65 \\
 Exchange &  45.45 &    101.12 &          6.57 &  107.28 &     63.17 &          4.12 \\
   Moment &  12.99 &       - &           - &   23.22 &       - &           - \\
      MLE &  16.49 &       - &           - &   31.13 &       - &           - \\
      \midrule
Sandstone & \multicolumn{3}{c}{Group 1 (n=35)} &\multicolumn{3}{c}{Group 2 (n=13)} \\
   & $\kappa$ & $\mathrm{ESS}_{\kappa}$ & $\mathrm{ESS}_{\kappa}$/s & $\kappa$ & $\mathrm{ESS}_{\kappa}$ & $\mathrm{ESS}_{\kappa}$/s  \\
   \cmidrule(lr){2-4} \cmidrule(lr){5-7}
                     BP &  20.54 &    232.08 &          9.21 &  49.05 &    117.05 &         12.38 \\
       RR &   20.5 &    131.06 &          1.17 &  44.27 &    141.78 &          2.05 \\
   RR-aux &  20.69 &    157.71 &          7.52 &  47.11 &    118.79 &          7.87 \\
 Exchange &  20.74 &     80.44 &           5.6 &  54.33 &     46.45 &           3.6 \\
   Moment &  22.36 &       - &           - &  68.94 &       - &           - \\
      MLE &  22.44 &       - &           - &  70.18 &       - &           - \\
      \midrule
Stone & \multicolumn{3}{c}{Group 1 (n=101)} &\multicolumn{3}{c}{Group 2 (n=101)} \\
    & $\kappa$ & $\mathrm{ESS}_{\kappa}$ & $\mathrm{ESS}_{\kappa}$/s & $\kappa$ & $\mathrm{ESS}_{\kappa}$ & $\mathrm{ESS}_{\kappa}$/s  \\
   \cmidrule(lr){2-4} \cmidrule(lr){5-7}
                           BP &  4.27 &     83.03 &          3.33 &  2.13 &     33.15 &          3.45 \\
       RR &  4.25 &     25.15 &          0.08 &  2.07 &     12.77 &          0.05 \\
   RR-aux &   4.2 &    129.21 &           3.5 &   2.2 &     37.41 &          1.02 \\
 Exchange &   5.3 &     45.59 &           3.2 &  2.95 &     62.47 &          4.83 \\
   Moment &  4.29 &       - &           - &  1.99 &       - &           - \\
      MLE &  4.32 &       - &           - &  2.19 &       - &           - \\
\bottomrule
\end{tabular}}
\caption{Results for the Kent model of the four datasets when estimating $\kappa$. All the chains ran for 10,000 iterations. The term $\kappa$ refers to the posterior mean for the Bayesian methods and estimates for the frequentist methods. The sign correction is applied for BP, RR and RR-aux. }\label{tab:FB5_empirical_kappa}
\end{table}
\begin{figure}[ht!]
    \centering\includegraphics[width=0.8\linewidth]{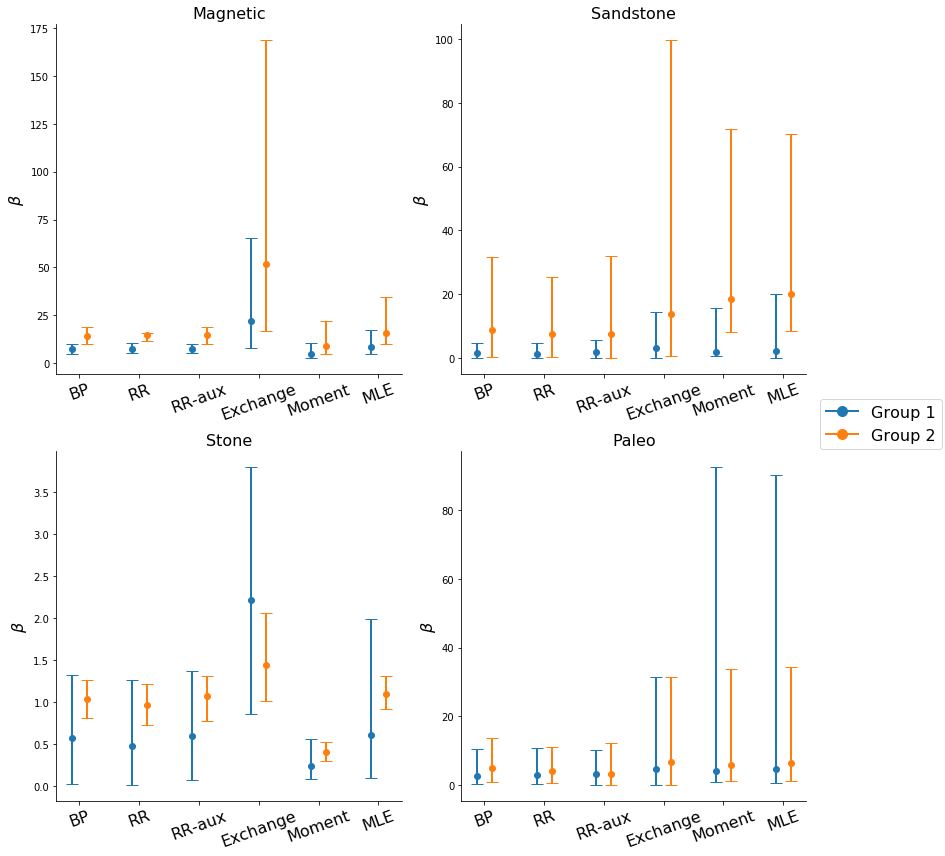}
    \caption{Posterior mean (estimate for Moment and MLE) of $\beta$ for two groups of the four datasets. For Bayesian methods, the error bar represents the 95\% credible interval of the posterior distributions. For MLE and Moment, the error bar shows the 95\% confidence intervals of the bootstrapped samples.}
    \label{fig:FB5_kappa}
\end{figure}

\begin{figure}[ht!]
    \centering
    \includegraphics[width=0.8\linewidth]{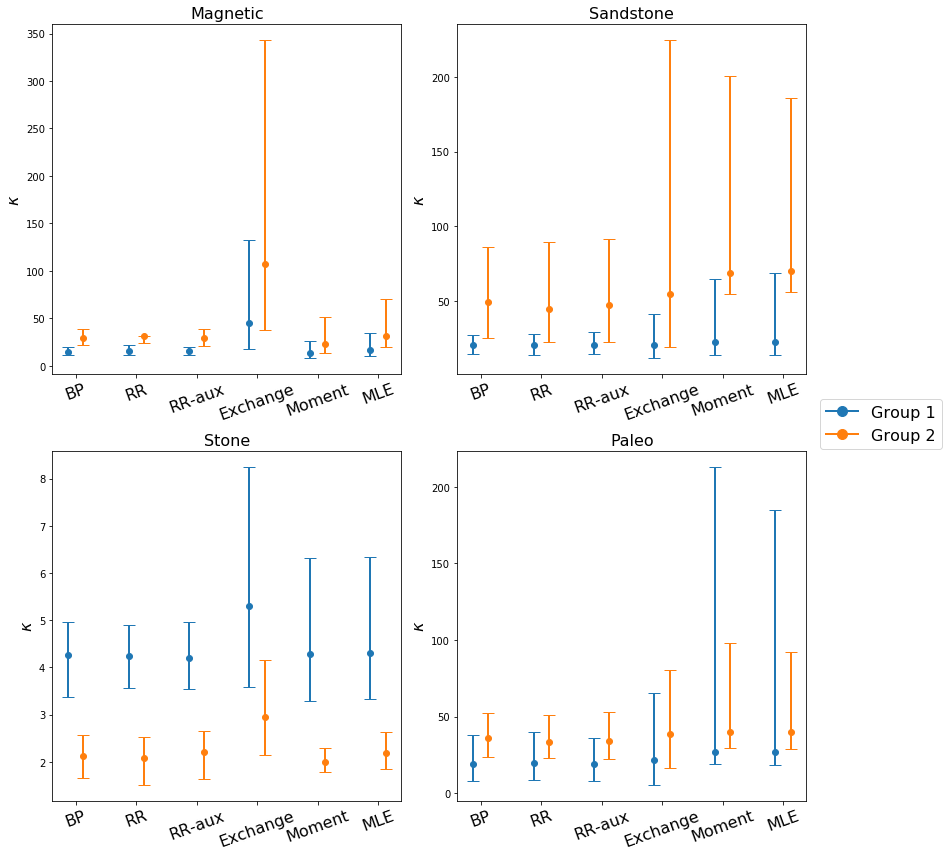}
    \caption{Posterior mean (estimate for Moment and MLE) of $\kappa$ for two groups of the four datasets. For Bayesian methods, the error bar represents the 95\% credible interval of the posterior distributions. For MLE and Moment, the error bar shows the 95\% confidence intervals of the bootstrapped samples.}
    \label{fig:FB5_ratio}
\end{figure}

\end{document}